\documentclass[10pt]{article}
\usepackage{fullpage,graphicx,subfigure,mathdots,mathpazo,color}
\usepackage{amsmath,amscd,tikz,mathrsfs}
\usepackage[normalem]{ulem}
\usepackage{amsmath}
\usepackage{setspace,booktabs}

\usepackage{epsfig,amsmath,graphicx,amssymb,overpic}

\usepackage{graphicx}
\usepackage{dcolumn}
\usepackage{bm}
\usepackage{graphicx}
\usepackage{subfigure}
\usepackage{epsfig,amsmath,graphicx,amssymb,overpic,cite}

\usepackage{graphicx}
\usepackage{dcolumn}
\usepackage{bm}
\usepackage{graphicx}
\usepackage{subfigure}
\usepackage{amsthm}

\newcommand{\ee}{\mathrm{e}}

\usepackage[T1]{fontenc}
\usepackage{amsmath, amsfonts, amssymb, amsthm}
\usepackage{geometry} 
\usepackage{hyperref} \hypersetup{colorlinks=true, linkcolor=blue, citecolor=red}
\usepackage{enumitem}

\usepackage{graphicx}
\usetikzlibrary{positioning,arrows.meta}

\usepackage{graphicx}      
\usepackage{titlesec}
\def\bee{\begin{eqnarray}}
\def\ene{\end{eqnarray}}
\def\bes{\begin{subequations}}
\def\ees{\end{subequations}}
\def\no{{\nonumber}}
\usepackage{color}

\def\i{\mathrm{i}}
\def\v{\vspace{0.06in}}
\def\be{\begin{equation}}
\def\ee{\end{equation}}
\def\bee{\begin{eqnarray}}
\def\ene{\end{eqnarray}}
\def\bes{\begin{subequations}}
\def\ees{\end{subequations}}
\def\no{\nonumber}
\def\d{\displaystyle}
\def\i{\mathrm{i}}

\theoremstyle{plain}
\newtheorem{theorem}{Theorem}[section]

\newtheorem{proposition}[theorem]{Proposition}

\newtheorem{RH}[theorem]{Riemann-Hilbert Problem}

\theoremstyle{definition}
\newtheorem{definition}[theorem]{Definition}
\newtheorem{assumption}[theorem]{Assumption}

\newtheorem{remark}[theorem]{Remark}

\allowdisplaybreaks[4]

\numberwithin{equation}{section}

\begin{document}

\baselineskip=14pt \renewcommand {\thefootnote}{\dag}
\renewcommand
{\thefootnote}{\ddag} \renewcommand {\thefootnote}{ }

\pagestyle{plain}


\begin{center}
\baselineskip=18pt \leftline{} \vspace{0.05in} 
{\large\bf Long-time asymptotics of the integrable defocusing Wadati–Konno–Ichikawa equation with  a finite-genus algebro-geometric background} \\[0.2in]
\end{center}

\begin{center}
 {{Taohua\, Luo}$^{1,2}$,\,\,  Zhenya\, Yan$^{3,2,1,*}$,\,\,  Guoqiang\, Zhang$^{1}$}
\footnote{$^{*}${\it Email address}: zyyan@mmrc.iss.ac.cn (Corresponding author)}  \\[0.1in]
\baselineskip=11pt {\footnotesize  $^1${\it State Key Laboratory of Mathematical Sciences, Academy of Mathematics and Systems Science, \\ \footnotesize Chinese Academy of Sciences,  Beijing 100190, China}}\\
{\footnotesize  $^2${\it School of Mathematical Sciences, University of Chinese Academy of Sciences, Beijing 100049, China}}\\
{\footnotesize  $^3${\it School of Mathematics and Information Science, Zhongyuan University of Technology, Zhengzhou 450007, China}} \\
\end{center}

\v \v

 {\bf Abstract.} We study the finite-genus algebro-geometric solutions of the Wadati–Konno–Ichikawa (WKI) equation with the saturable nonlinearity and long-time asymptotic behaviors of their short-range perturbations. First, for both the focusing and defocusing reductions, we formulate the finite-genus Baker–Akhiezer functions as explicitly solvable the matrix Riemann–Hilbert (RH) problems on the complex spectral plane and obtain theta-function representations together with the reconstruction formulae for the WKI field and the reciprocal coordinate. We then consider the Cauchy problem of the defocusing WKI equation on a finite-genus algebro-geometric background. We construct the scattering data and  RH problem, and perform a Deift-Zhou nonlinear steepest descent analysis. The space-time plane is divided into two transition regions, a Zakharov–Manakov (ZM) region, and a fast-decay region. The leading term is a phase-shifted finite-genus WKI solution. The transition corrections are governed by a Painlev\'e-XXXIV model, while the ZM radiation is described by parabolic-cylinder functions. The reciprocal-coordinate asymptotics are obtained simultaneously.

\vspace{1em}




\begin{spacing}{1.2}
\baselineskip=13pt
\tableofcontents
\vspace{-0.15in}
\end{spacing}


\baselineskip=14pt

\section{Introduction}

\quad  In this paper, we consider the integrable Wadati-Konno-Ichikawa (WKI) equation with the saturable nonlinearity
\begin{equation}\label{WKI}
\mathrm{i}q_t+
\left(
\frac{q}{\sqrt{1+\sigma |q|^2}}
\right)_{xx}=0,
\quad (x,t)\in\mathbb{R}\times\mathbb{R}^+,
\quad  \mathrm{i}=\sqrt{-1},
\end{equation}
which was presented by Wadati, Konno, and Ichikawa \cite{Wadati1979b} from a new spectral problem, differing from the standard AKNS framework \cite{Ablowitz1974}, where $q=q(x,t)$ is a complex envelope field, and the cases $\sigma=1\, (-1)$ correspond to the focusing (defocusing) reduction. In particular, as $\sigma=0$, the nonlinear WKI equation (\ref{WKI}) reduces to the linear Schr\"odinger equation $\i q_t+q_{xx}=0$, which can be further reduced to the diffusion/heat equation $q_{\tau}-q_{xx}=0$ under the imaginary-time transform $\tau=it$. In contrast with polynomial nonlinearities of positive flows (e.g., nonlinear Schr\"odinger (NLS) equation \cite{Zakharov1972}) in the AKNS hierarchy \cite{Ablowitz1974}, the WKI nonlinearity has a saturable form.
The WKI hierarchy also appears in geometric and mechanical settings; in particular, closely related reductions describe nonlinear
transverse oscillations of elastic beams under tension \cite{Ichikawa1981}.
The WKI equation (\ref{WKI}) can be rewritten as the Hamiltonian form \cite{Wadati1979b}
\bee
 \i q_t=\frac{\partial^2}{\partial x^2}\frac{\delta {\mathcal H}_0}{\delta \bar{q}},\qquad {\mathcal H}_0=2\sigma \int_{\mathbb{R}}\left(1-\sqrt{1+\sigma|q|^2}\right)dx,
\ene
where $\overline q$ denotes the complex conjugate of $q$. The potential form of (\ref{WKI}) can be written as (let $q\to q_x$ and integrate it w.r.t $x$) \cite{Chen2021}
\begin{equation}\label{WKI2}
\mathrm{i}q_t+
\left(
\frac{q_x}{\sqrt{1+\sigma |q_x|^2}}
\right)_{x}=0,
\qquad (x,t)\in\mathbb{R}\times\mathbb{R}^+,
\qquad \sigma=\pm1,
\end{equation}

As $|q|<1$ and $\sigma=\pm 1$, the WKI equation (\ref{WKI}) can be rewritten as the infinite-degree differential polynomial form
\bee \label{wki-exp}
 \i q_t+q_{xx}+\left(-\frac{1}{2}\sigma |q|^2q+\frac{3}{8} |q|^4q+\cdots+(-\sigma)^n\frac{(2n-1)!!}{(2n)!!} |q|^{2n}q+\cdots\right)_{xx}=0,
\ene
which with $n=1$ generates the the truncated approximate nonlinear model \cite{Shimizu1980}
\bee \label{WKI-n}
 \i q_t+q_{xx}-\frac{1}{2}\sigma (|q|^2q)_{xx}=0,
\ene
which can be rewritten as the Hamiltonian form
\bee
 \i q_t=\frac{\partial^2}{\partial x^2}\frac{\delta {\mathcal H}_0}{\delta \bar{q}},\qquad {\mathcal H}_0= \int_{\mathbb{R}}\left(\frac{1}{4}\sigma |q|^4-|q|^2\right)dx.
\ene
Of course, one can also truncate (\ref{wki-exp}) to generate other new higher-degree forms, for example,
\bee
 \i q_t+q_{xx}-\frac{1}{2}\sigma (|q|^2q)_{xx}+\frac{3}{8} (|q|^4q)_{xx}=0.
 \ene

Eq. (\ref{WKI-n}) differs from  both the NLS equation \cite{Zakharov1972}
\bee
 \i q_t+q_{xx}-\frac{1}{2}\sigma |q|^2q=0,
\ene
and the derivative NLS equation \cite{kaup1978}
\bee
 \i q_t+q_{xx}-\frac{1}{2}\i\sigma  (|q|^2q)_{x}=0,
\ene
Notice that the stationary equation of (\ref{WKI-n}) with $q(x,t)=\hat{q}(x)e^{-\i\mu t}$,
\bee
 \hat{q}_{xx}+\mu \hat{q}-\frac{1}{2}\sigma (|\hat{q}|^2\hat{q})_{xx}=0,
\ene
is similar to the travelling-wave forms with the transform $u(x,t)=u(\xi),\, \xi=x-vt$ of the short-pulse equation \cite{schafer2004}
\bee
 u_{xt}-u+\frac{1}{6}\sigma (u^3)_{xx}=0
\ene
and its complex form \cite{feng2015,feng2016}.
\bee
 u_{xt}+u+\frac{1}{2}\sigma(|u|^2u_x)_{x}=0.
\ene

The WKI equation (\ref{WKI}) is completely integrable and admits the Lax representation
\begin{equation}\label{Lax}
\Psi_x=X(x,t;z)\Psi,
\qquad
\Psi_t=T(x,t;z)\Psi,
\end{equation}
where $\Psi=\Psi(x,t;z)$ is the complex matrix-valued eigenfunction, $z\in\mathbb{C}$ is the spectral parameter,
\begin{equation}
X=-\mathrm{i}z\sigma_3+zQ,
\qquad
T=
\begin{pmatrix}
-\dfrac{2\mathrm{i}z^2}{p}
& \dfrac{2z^2q}{p}+\mathrm{i}z\left(\dfrac{q}{p}\right)_x \\[3mm]
-\dfrac{2z^2\sigma\overline q}{p}+\mathrm{i}z\sigma\left(\dfrac{\overline q}{p}\right)_x
& \dfrac{2\mathrm{i}z^2}{p}
\end{pmatrix},
\end{equation}
with the potential matrix
\begin{equation}\label{eq:Q-p-intro}
Q=Q(x,t)=
\begin{pmatrix}
0&q(x,t)\\
-\sigma\overline q(x,t) &0
\end{pmatrix},
\qquad
p=p(x,t)=\sqrt{1+\sigma|q(x,t)|^2},
\end{equation}
and three Pauli matrices are
\begin{equation}
\sigma_1=
\begin{pmatrix}0&1\\1&0\end{pmatrix},
\qquad
\sigma_2=
\begin{pmatrix}0&-\mathrm{i}\\ \mathrm{i}&0\end{pmatrix},
\qquad
\sigma_3=
\begin{pmatrix}1&0\\0&-1\end{pmatrix}.
\end{equation}
 From the analytical point of view, the most characteristic feature of \eqref{WKI} is that the spectral parameter multiplies the potential in the spatial part of the Lax pair \cite{Lax1968}. This nonstandard dependence is responsible for several phenomena that do not occur in the usual Zakharov--Shabat problem \cite{Zakharov1972}, including a reciprocal change of the spatial variable, simultaneous spectral analysis near $z=0$ and $z=\infty$, and a parametric reconstruction of the physical solution.

Some integrable structures of the WKI system (\ref{WKI}) have been investigated from several complementary directions. Shimizu and Wadati developed its inverse-scattering theory and obtained explicit solitary waves of (\ref{WKI}) \cite{Shimizu1980}, whose higher-pole solutions were found.
 Relations between the WKI and AKNS spectral schemes were subsequently established through gauge and reciprocal transformations \cite{Ishimori1982,Wadati1983}. The hierarchy possesses infinitely many conservation laws and admits Hamiltonian and algebraic formulations \cite{Wadati1979b}. Based on the C.Neumann and Bargmann constraints, respectively, Qiao \cite{qiao93,qiao95} nonlinearized the WKI spectrum problem (\ref{Lax}) into the Hamilton systems. Qu and Zhang \cite{qu2005} re-deduced the WKI equation (\ref{WKI}) from the motions of curves in Euclidean geometry $E^3$. The ${\rm sl}(2)$ WKI spectral problem (\ref{Lax}) was also extended to a ${\rm so}(3)$ soliton hierarchy \cite{Mawx2014}. The B\"acklund transforms were studied for the WKI equation \cite{BT1,BT2,BT3,BT4}. The hodograph and Darboux transformations were used to obtain focusing and defocusing multisolitons, breathers, and rogue waves in \cite{Zhang2017}. Direct inverse-scattering and blow-up properties of WKI solitons were studied in \cite{Liu2017}. A matrix Riemann--Hilbert formulation for the WKI equation with simple and higher-order poles was developed in \cite{Zhang2019}, producing determinant representations of soliton and positon-type solutions. Moreover, the higher-order pole solutions \cite{zhang2020} of the WKI equation were also found from the Gelfand–Levitan–Marchenko (GLM) equation via the Olmedilla’s method \cite{Olme1987}. The existence of global solution was shown for the focusing WKI equation with small initial data in the smooth function space \cite{shima2016}. The properties of scattering data were studied for the second flow of the WKI system with box-like initial value \cite{Tu2021}. Long-time analysis has also advanced substantially: the potential WKI equation was treated by nonlinear steepest descent in \cite{Chen2021}, and soliton resolution for decaying WKI data in weighted Sobolev spaces was established by a $\bar\partial$-steepest descent method in \cite{Li2022a}. For finite-density initial data, soliton resolution and asymptotic stability of $N$-soliton solutions in solitonic space-time regions were proved in \cite{Li2022b}. These works provide a rather complete picture for zero or constant nonzero backgrounds. To the best of our knowledge, the dynamics of perturbations of nonconstant quasi-periodic WKI backgrounds has remained largely unexplored.

Finite-genus, or finite-gap solutions form one of the central classes of exact solutions in the theory of integrable systems. Their origin can be traced to the periodic spectral problem for the Korteweg-de Vries (KdV) equation developed by Novikov \cite{Novikov1974} in 1974 and to the theta-function formulae of Its and Matveev in 1975 \cite{Its1975}. The general algebro-geometric theory was further developed by Dubrovin, Matveev, and Novikov \cite{Dubrovin1976} and by Krichever \cite{Krichever1976,Krichever1977}. In this framework, a finite-genus solution is encoded by a compact algebraic curve, marked points, local parameters, a nonspecial divisor, normalized Abelian differentials, and a point of the associated Jacobian variety (see, e.g., Refs. \cite{Belokolos1994,Belokolos1986}). The corresponding Baker--Akhiezer (BA) function is a meromorphic function or vector on the underlying Riemann surface with prescribed divisor and essential singularities. It simultaneously linearizes the nonlinear flow on the Jacobian and solves the auxiliary linear equations of the Lax pair. Besides yielding explicit quasi-periodic solutions, finite-genus BA functions play a second, equally important role: they provide global model solutions in nonlinear steepest descent analyses \cite{Deift1993}. In 1990, the Lax pair nonlinearization method was presented by Cao and Geng \cite{Cao1990a, Cao1990b} to find algebraic-geometric solutions of integrable systems \cite{Cao1999,Zhou1997}.
In 2003, Gesztesy and Holden \cite{Ges2003}, based on the spectral analysis and compact Riemann surfaces, presented
an effective method to construct algebro-geometric solutions of certain integrable hierarchies.

Though these above-mentioned techniques can be used to construct algebraic-geometric solutions of integrable nonlinear equation, where the BA functions are formulated directly on a two-sheeted hyperelliptic Riemann surface,
they did not find these algebraic-geometric solutions via the Riemann-Hilbert (RH) problems, in which the long-time asymptotics of solutions of integrable systems can be conveniently analyzed via the Deift-Zhou onlinear steepest descent method \cite{Deift1993}.
For asymptotic analysis, a matrix formulation on the ordinary complex spectral plane is often more convenient. A systematic realization of this idea was given by Kotlyarov and Shepelsky for the nonlinear Schr\"odinger (NLS) equation \cite{Kotlyarov2017}. Their construction encodes the two sheets of the spectral curve into the columns of a unimodular $2\times2$ matrix that is analytic in the complex plane cut along finitely many spectral arcs and satisfies piecewise constant jumps. After separating the exponential dependence on the space-time variables, the resulting model problem can be solved explicitly in terms of fourth-root functions, Abel maps, and Riemann theta functions. This planar BA formalism has since been adapted to the Maxwell--Bloch system \cite{Kotlyarov2018}, the derivative NLS equation \cite{Zhao2020}, the KdV equation \cite{Zhao2023}, the Hirota equation \cite{Cao2024}, and the extended mKdV equation \cite{cheng2026}. The advantage of the planar formulation is not merely notational. It places the exact finite-genus background and the inverse-scattering RH problem in the same analytic setting, making the background matrix directly available as a global parametrix in analyses of oscillatory RH problems \cite{Deift1993}. Moreover,
a numerical method was presented to study the finite-genus solutions of the KdV equation via Riemann–Hilbert problems \cite{trog13a,trog13b}.  Though the finite-gap and algebro-geometric WKI flows has been constructed in \cite{Li2016}, the planar matrix Riemann--Hilbert representation was not constructed for the finite-genus WKI Baker-Akhiezer function.

Up to now, the Deift-Zhou onlinear steepest descent method \cite{Deift1993}, building on earlier inverse-scattering asymptotics such as \cite{Its1981}, has become the principal tool for obtaining precise large-time expansions of integrable equations. For finite-genus backgrounds,
Kamvissis and Teschl \cite{Kamvissis2012} studied short-range perturbations of periodic and algebro-geometric Toda lattices and obtained modulated finite-gap asymptotics on the underlying curve. Mikikits-Leitner and Teschl \cite{MikikitsLeitner2012} established the analogous theory for perturbed finite-gap KdV solutions. In matrix problems of NLS type, the deformations can instead be carried out on the ordinary complex plane; see, e.g., the analysis of step-like oscillatory NLS backgrounds in \cite{BoutetdeMonvel2021}. Most recently, Fan {\it et al}\cite{Fan2026} obtained long-time asymptotics for the defocusing NLS equation with a finite-genus algebro-geometric background in two families of transition regions, a Zakharov--Manakov region, and a fast-decay region. Their work identifies a phase-shifted finite-genus background as the leading term and a Painlev\'e XXXIV correction in the transition regimes. Related planar finite-genus asymptotic analyses have recently been carried out for the focusing NLS equation \cite{Ma2026a} and for the focusing mKdV equation with discrete spectrum \cite{Ma2026b}. Painlev\'e XXXIV model problems previously appeared in critical random-matrix and orthogonal-polynomial asymptotics \cite{Its2008,Its2009}.
To the best of our knowledge, a planar matrix Riemann--Hilbert representation of the finite-genus WKI Baker--Akhiezer function, combined with a nonlinear steepest descent analysis for the defocusing WKI equation on a finite-genus algebro-geometric background, has not previously been developed.

In this paper, we will investigate these two issues. The first objective of the present paper is to develop such a planar finite-genus formalism for the WKI equation \eqref{WKI} via the RH problem. This requires modifications that are specific to the WKI spectral problem. Following the idea used in the Camassa-Holm equation \cite{constantin2003,Monvel2009,Monvel2006},
in terms of the function $p(x,t)$ in \eqref{eq:Q-p-intro},  a natural reciprocal spatial variable is
\begin{equation}\label{eq:y-intro}
y=y(x,t)
=
x+\int_{-\infty}^{x}\bigl(p(s,t)-1\bigr)\,ds=x+\int_{-\infty}^{x}\bigl(\sqrt{1+\sigma|q(s,t)|^2}-1\bigr)\,ds,
\end{equation}
that is, $\partial y(x,t)/\partial x=\sqrt{1+\sigma|q(x,t)|^2}$.

The second objective of the present paper is to study the long-time asymptotics of the  solution for Cauchy problem of the defocusing ($\sigma=-1$) WKI equation with the reciprocal coordinate \eqref{eq:y-intro} under a short-range perturbation of a finite-genus algebro-geometric background
\begin{equation}\label{background}
q(x,0)=q_0(x)\sim q^{(\mathrm{AG})}(x,0),
\qquad x\to\pm\infty,
\end{equation}
where the finite-genus algebro-geometric solution of the defocusing WKI equation is written as
\begin{equation}\label{AGSolution}
q^{(\mathrm{AG})}(x,t)
=
q^{(\mathrm{AG})}\bigl(x,t;\mathbf E,\widehat{\mathbf E},\boldsymbol\phi\bigr),
\end{equation}
where
$\mathbf E=(E_0,\ldots,E_n),\,\, \widehat{\mathbf E}=(\widehat E_0,\ldots,\widehat E_n),\,\, \boldsymbol\phi=(\phi_1,\ldots,\phi_n),$
with $n\in\mathbb N_0$ and and the corresponding hyperelliptic Riemann surface is defined by
\begin{equation}\label{RSurface}
	w^2(z)=\prod_{j=0}^{n}(z-E_j)(z-\widehat E_j),
\end{equation}
where $E_0<\widehat E_0<E_1<\widehat E_1<\cdots<E_n<\widehat E_n.$

In the defocusing case $(\sigma=-1)$, the regularity of this transformation is tied to the constraint $|q|<1$. Unlike the standard NLS reconstruction \cite{Zakharov1972,Deift1993}, in which the potential is recovered from the large-$z$ expansion of the RH solution, the WKI field and the physical coordinate are reconstructed from the expansion at the zero spectral point. Consequently, the RH problem is normalized at $z=\infty$, whereas the nonlinear solution is recovered from $z\to0$.

For completeness, the boundary condition is extended to $t\geq0$ by assuming that the perturbation remains integrable relative to the algebro-geometric background:
\begin{equation}\label{qRestrict}
\int_{\mathbb R}
\left|q(s,t)-q^{(\mathrm{AG})}(s,t)\right|\,ds<\infty,
\qquad t\geq0.
\end{equation}
The precise short-range and spectral regularity assumptions used in the analysis are specified in Section~\ref{mainresults}.

The two main parts of the paper are therefore logically inseparable. The planar finite-genus construction supplies the exact background Baker--Akhiezer matrix and the global parametrix for the nonlinear steepest descent analysis. Conversely, the asymptotic analysis demonstrates why the planar representation is particularly effective: the same theta-function matrix that produces exact finite-genus WKI solutions also controls the leading dynamics of their short-range perturbations.

The some highlights of the present work are summarized as follows:
\begin{itemize} 

\item [i)] We formulate the finite-genus Baker--Akhiezer functions for both focusing $(\sigma=1)$ and defocusing $(\sigma=-1)$ WKI reductions (\ref{WKI}) as explicitly solvable the matrix RH problems on the complex $z$-plane. Two natural contour and homology configurations are treated, and theta-function formulae for the corresponding background matrices are obtained.

\item [ii)] We derive reconstruction formulae for both the finite-genus field $q^{(\mathrm{AG})}(y,t)$ and the physical coordinate $x(y,t)$ from the zero-spectral-point expansion. This makes the reciprocal WKI geometry compatible with a planar RH formulation normalized at infinity.

\item [iii)] For short-range perturbations of a finite-genus defocusing WKI background (\ref{background}), we construct the background Jost solutions, scattering data, and the inverse RH problem, including the band jumps and fourth-root endpoint behavior.

\item [iv)] We carry out a Deift-Zhou nonlinear steepest descent analysis in two transition regions, a Zakharov-Manakov region, and a fast-decay region. The leading term is a phase-shifted finite-genus WKI solution; the transition corrections are governed by Painlev\'e XXXIV and have order $t^{-1/3}$, the Zakharov-Manakov radiation has order $t^{-1/2}$, and the fast-decay error is $O(t^{-1})$.

\item [v)] We propagate the asymptotic estimates through the $z=0$ reconstruction and simultaneously recover the nonlinear field and the reciprocal coordinate. This step differs from one of the standard finite-genus NLS problem and is essential for interpreting the asymptotics in the original WKI variables.
\end{itemize}

\subsection{Main results}\label{mainresults}


\begin{figure}[!t]
	\centering
	\includegraphics[scale=0.5]{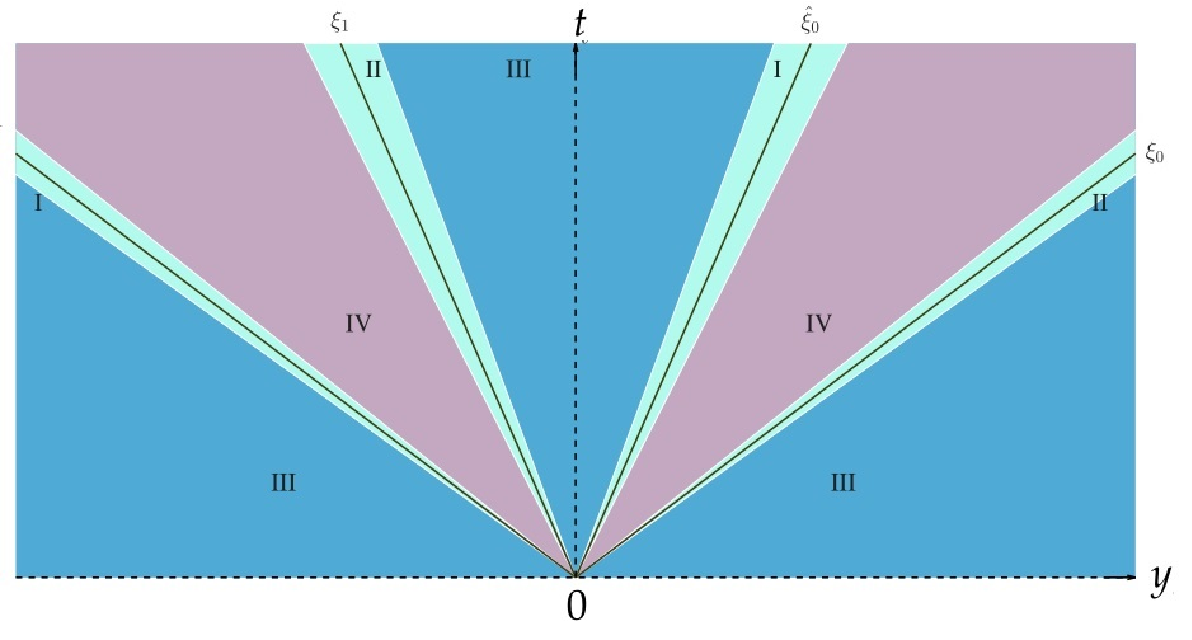}
	\vspace{-0.1in}\caption{Four distinct asymptotic regions in $(y,t)$-space for genus $n=1$.}
	\label{fig:WKI7}
\end{figure}

Let $E_j$ and $\widehat{E}_j$ \, $(j=0,1,\ldots,n)$ be fixed numbers with $E_0<\widehat E_0<E_1<\widehat E_1<\cdots<E_n<\widehat E_n$ for the defocusing case, and $\{E_j,\widehat{E}_j\}_{j=0}^{n}$ being distinct points in $\mathbb{C}$ for the focusing case, and
\begin{equation} \label{fg}
	\begin{aligned}
		f(z)=\int_{\widehat{E}_0}^z\frac{\prod_{k=0}^n(s-z_k^f)}{w(s)}\mathrm{d}s,\qquad
		g(z)=\int_{\widehat{E}_0}^z\frac{4\prod_{k=0}^{n+1}(s-z_k^g)}{w(s)}\mathrm{d}s,
	\end{aligned}
\end{equation}
be two hyperelliptic integrals, where $w^2(z)=\prod_{j=0}^{n}(z-E_j)(z-\widehat E_j)$.  Then these constants $z_k^f$ and $z_k^g$ are uniquely determined by the conditions
\begin{equation} \label{fg-p}
	\begin{cases}
		\begin{array}{lll}
			f(z)=z+\mathcal{O}(1), &  g(z)=2z^2+\mathcal{O}(1), & z\to\infty, \v\\
			\int_{E_j}^{\widehat{E}_j}\mathrm{d}f=\int_{E_j}^{\widehat{E}_j}\mathrm{d}g=0, & j=0,\ldots,n.
		\end{array}
	\end{cases}
\end{equation}
respectively. Moreover, $\left(\{z_k^f\}_{k=0}^n\cup\{z_k^g\}_{k=0}^{n+1}\right)\cap\{E_j,\widehat{E}_j\}_{j=0}^n=\emptyset$.
It follows from \eqref{fg} and \eqref{fg-p} that $f(E_j)=f(\widehat{E}_j)$ and $g(E_j)=g(\widehat{E}_j)$ for $j=0,1,\ldots,n$.

The points $\xi_j$ and $\widehat{\xi}_j$, $j=0,\ldots,n$, are precisely the critical values associated with the phase function $\theta(z,\xi)$ defined in \eqref{theta}. They are given by
\begin{equation} \label{xi}
	\xi_j=-\frac{4\prod_{k=0}^{n+1}(E_j-z_k^g)}{\prod_{k=0}^n(E_j-z_k^f)},\qquad
	\widehat{\xi}_j=-\frac{4\prod_{k=0}^{n+1}(\widehat{E}_j-z_k^g)}{\prod_{k=0}^n(\widehat{E}_j-z_k^f)}.
\end{equation}
These points are used to divide the $(y,t)$-half-plane into several distinct space--time regions (see the following Definition \ref{def1}).

\begin{definition} \label{def1}
	For any positive constant $C$, we let $\xi=y/t$, and consider \eqref{xi} to divide the $(y,t)$ space into the four regions:
	\begin{itemize}
		\item Transition region I:\, $\cup_{j=0}^{n}\{\xi : |\xi-\widehat{\xi}_j|t^{2/3}\leq C\}$.
		
		\item Transition region II:\, $\cup_{j=0}^{n}\{\xi : |\xi-\xi_j|t^{2/3}\leq C\}$.
		
		\item Zakharov--Manakov region III:\, $\xi\in(-\infty,\widehat{\xi}_n)\cup_{j=1}^{n}(\xi_j,\widehat{\xi}_{j-1})\cup(\xi_0,+\infty)$.
		
		\item Fast-decay region IV:\, $\xi\in\cup_{j=0}^{n}(\widehat{\xi}_j,\xi_j)$.
	\end{itemize}
	\label{def-1}
\end{definition}

An illustration of the four regions described above is given in Figs.~\ref{fig:WKI7} and~\ref{fig:WKI8} for genus $n=1,2$. We derive the long-time asymptotics of $q(y,t)$ in these regions under the following assumptions.

\begin{assumption}\label{ass-1}
	Throughout this paper, we impose the following assumptions:
	\begin{itemize}
		\item There exists a positive constant $C$ such that $q_0(x)=q^{(AG)}(x,0)$ for $|x|>C$; that is, the initial data $q_0(x)$ coincides with the background outside a compact set.
		
		\item Let $r_i$, $i=1,2$, defined in \eqref{r}, be the reflection coefficients associated with the initial data. We assume that $|r_i(\eta)|=1$ for $\eta\in\{E_j,\widehat{E}_j\}_{j=0}^n$ and that $r_i$ is bounded on $\cup_{j=0}^n(E_j,\widehat{E}_j)$.
		
		\item Throughout the analysis of the defocusing WKI equation (\ref{WKI}), we assume that $|q(x,t)|<1$.
	\end{itemize}
\end{assumption}

\begin{figure}[!t]
	\centering
	\includegraphics[scale=0.5]{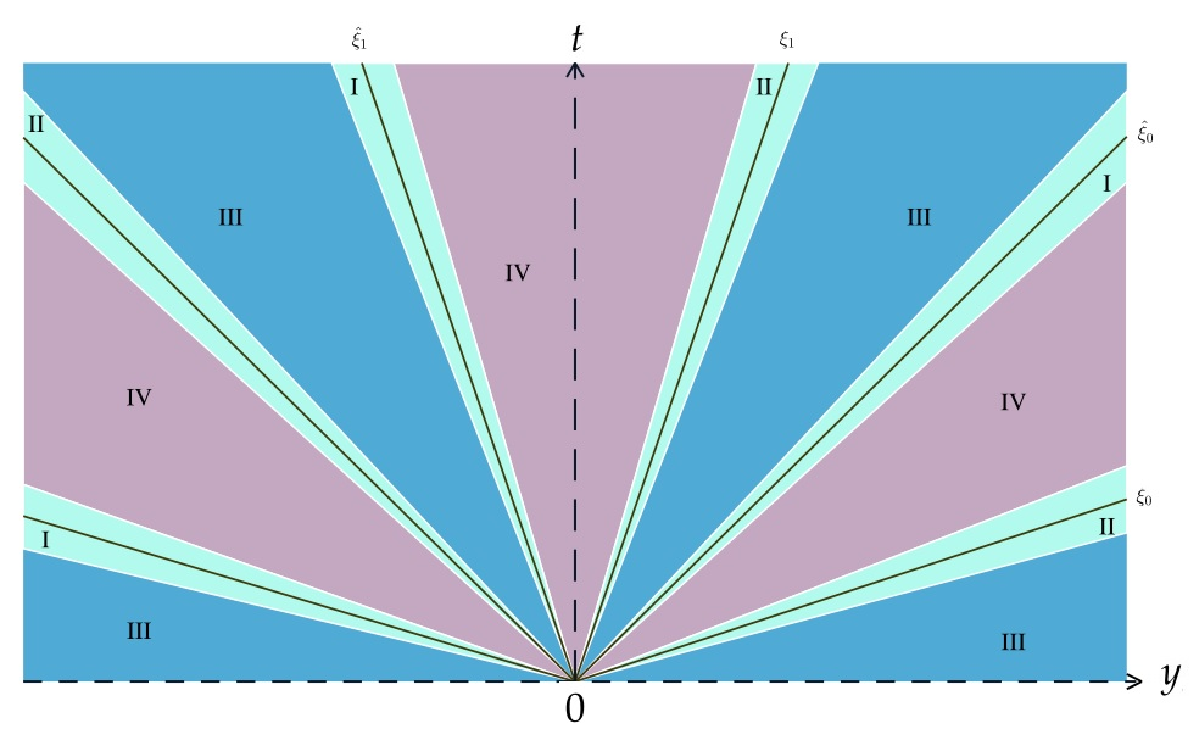}
	\vspace{-0.1in}\caption{Four distinct asymptotic regions in $(y,t)$-space for genus $n=2$.}
	\label{fig:WKI8}
\end{figure}

\begin{theorem} \label{theom-AG}
For the $M(y,t;z)$ satisfying the RH problem \ref{RH-A-M}, the  finite-genus algebro-geometric solution of the WKI equation (\ref{WKI}) is given by (\ref{AG-A}) with (\ref{AG-A2}).
\end{theorem}

\begin{theorem} \label{theom}
	For the given finite-genus algebro-geometric solution $q^{(AG)}(x,t)=q^{(AG)}(x,t;\boldsymbol{E},\boldsymbol{\widehat{E}},\boldsymbol{\phi})$ of the defocusing WKI equation \eqref{WKI} with $\sigma=-1$ found in Theorem \ref{theom-AG}, the solution of the Cauchy problem \eqref{WKI}, \eqref{background}, subject to the constraint \eqref{qRestrict} and Assumption~\ref{ass-1}, has the following long-time asymptotic behaviors in the regions I-IV specified in Definition~\ref{def-1}:
	\begin{itemize}
		\item {} {\bf Asymptotics in transition region I:}
		Let $j_1\in\{0,\ldots,n\}$ be fixed and $|\xi-\widehat{\xi}_{j_1}|t^{2/3}\leq C$. Then
		\begin{equation}
			\begin{aligned}
				q(x,t)=
				e^{-2\delta(0)}q^{(AG)}(x,t;\boldsymbol{E},\boldsymbol{\widehat{E}},\boldsymbol{\phi}-\boldsymbol{\delta})
				-e^{-2\delta(0)}
				\left(\frac{\partial}{\partial y}\mathcal{F}_{12}\right)\left(1+\mathrm{i}\frac{\partial}{\partial y}\mathcal{F}_{11}\right)
				+\mathcal{O}(t^{-\epsilon}),
			\end{aligned}
		\end{equation}
		with
		\begin{equation}
			x=y+\mathrm{i}\mathcal{F}_{11}+\mathcal{O}(t^{-\epsilon}),
		\end{equation}
		where the real vector $\boldsymbol{\delta}=(\delta_1,\ldots,\delta_n)$ and the function $\delta$ are defined by \eqref{delta_1} and \eqref{delta1}, respectively, $\mathcal{F}=\left(\mathcal{F}_{ij}\right)_{i,j=1,2}$ is given by \eqref{mathcalF}.
		
		\item {} {\bf Asymptotics in transition region II:}
		Let $j_1\in\{0,\ldots,n\}$ be fixed and $|\xi-\xi_{j_1}|t^{2/3}\leq C$. Then
		\begin{equation}
			\begin{aligned}
				q(x,t)=
				e^{-2\delta(0)}q^{(AG)}(x,t;\boldsymbol{E},\boldsymbol{\widehat{E}},\boldsymbol{\phi}-\boldsymbol{\delta})
				+e^{-2\delta(0)}
				\left(\frac{\partial}{\partial y}\mathcal{\tilde{F}}_{12}\right)\left(1+\mathrm{i}\frac{\partial}{\partial y}\mathcal{\tilde{F}}_{11}\right)
				+\mathcal{O}(t^{-\epsilon}),
			\end{aligned}
		\end{equation}
		with
		\begin{equation}
			x=y+\mathrm{i}\mathcal{\tilde{F}}_{11}+\mathcal{O}(t^{-\epsilon}),
		\end{equation}
		where the real vector $\boldsymbol{\delta}=(\delta_1,\ldots,\delta_n)$ and the function $\delta$ are defined by \eqref{delta_2} and \eqref{delta2}, respectively, $\mathcal{\tilde{F}}$ is given by \eqref{mathcaltildeF}.
		
		\item {} {\bf Asymptotics in the Zakharov--Manakov region III:}
		For $\xi\in(-\infty,\widehat{\xi}_n)\cup_{j=1}^n(\xi_j,\widehat{\xi}_{j-1})\cup(\xi_0,+\infty)$,
		\begin{equation}
			\begin{aligned}
				q(x,t)=
				e^{-2\delta(0)}q^{(AG)}(x,t;\boldsymbol{E},\boldsymbol{\widehat{E}},\boldsymbol{\phi}-\boldsymbol{\delta})
				+e^{-2\delta(0)}
				\left(\frac{\partial}{\partial y}\mathcal{J}_{12}\right)\left(1+\mathrm{i}\frac{\partial}{\partial y}\mathcal{J}_{11}\right)
				+\mathcal{O}(t^{-1}),
			\end{aligned}
		\end{equation}
		with
		\begin{equation}
			x=y+\mathrm{i}\mathcal{J}_{11}+\mathcal{O}(t^{-1}),
		\end{equation}
		where the real vector $\boldsymbol{\delta}=(\delta_1,\ldots,\delta_n)$ and the function $\delta$ are defined by \eqref{delta_3} and \eqref{delta3}, respectively, $\mathcal{J}$ is given by \eqref{mathcalJ}.
		
		\item {} {\bf Asymptotics in the fast-decay region IV:}
		For $\xi\in\cup_{j=0}^{n}(\widehat{\xi}_{j},\xi_{j})$,
		\begin{equation}
			\begin{aligned}
				q(x,t)=
				e^{-2\delta(0)}q^{(AG)}(x,t;\boldsymbol{E},\boldsymbol{\widehat{E}},\boldsymbol{\phi}-\boldsymbol{\delta})
				+\mathcal{O}(t^{-1}),
			\end{aligned}
		\end{equation}
		with
		\begin{equation}
			x=y+\mathcal{O}(t^{-1}).
		\end{equation}
	\end{itemize}
\end{theorem}

\begin{remark} i) The above-found results can also be extended to other integrable higher-order WKI flows \cite{Wadati1979b,Li2016}, for example,
the third-order WKI flow
\begin{equation}\label{WKI-2}
q_t+
\left(
\frac{q_x}{(1+\sigma |q|^2)^{3/2}}
\right)_{xx}=0,
\end{equation}
and fourth-order WKI flow
\begin{equation}\label{WKI-2}
\mathrm{i}q_t-
\left(
\frac{4(1+\sigma|q|^2)[q_{xx}(2+\sigma|q|^2)-2\sigma q|q_x|^2-\sigma q^2\bar{q}_{xx}-3\sigma \bar{q}q_{x}^2+5q(|q|^4)_x]}
{32(1+\sigma |q|^2)^{7/2}}
\right)_{xx}=0.
\end{equation}
ii) Moreover, we shall address the long-time asymptotics of the Cauchy problem of the {\it focusing} WKI equation with a finite-genus algebro-geometric background found in Sec. 3 in another work.
\end{remark}

The rest of the paper is organized as follows.
Section~\ref{Preliminaries} reviews the hyperelliptic surface, the classical Baker--Akhiezer data, the WKI gauge transformation, and the reconstruction at $z=0$. In Section~\ref{Representation}, we construct the scalar phase functions and solve the two planar finite-genus model Riemann-Hilbert problems explicitly in terms of Riemann theta functions (i.e., the proof of Theorem \ref{theom-AG}). Section~\ref{JostFuc} develops the Jost solutions and scattering data for the defocusing finite-genus background and analyzes the signature table of the phase. Sections~\ref{Region-I} and~\ref{Region-II} are devoted to the two transition regions and to the associated Painlev\'e XXXIV local parametrices. In Section~\ref{RegionIII-IV}, we study the Zakharov-Manakov and fast-decay regions by means of parabolic-cylinder and small-norm model problems. Finally, Section~\ref{Proof-Regions} evaluates the transformed Riemann-Hilbert solutions at $z=0$ and completes the proof of  Theorem \ref{theom} for the asymptotic formulae for $q(x,t)$ and $x(y,t)$.

\section{Preliminaries}\label{Preliminaries}

In this section, we briefly recall some basic concepts and results from the theory of hyperelliptic Riemann surfaces \cite{Belokolos1994,Farkas-Riemann92,Ges2003,Kotlyarov2017}. We give the construction of the Baker-Akhiezer function for the WKI equation (\ref{WKI}) via the corresponding  RH problem.

Let $\{E_j,\widehat{E}_j\}_{j=0}^{n}$ be $2n+2$ distinct points in $\mathbb{C}$, and
$\mathcal{X}$ be the Riemann surface of genus $n$ defined by the equation
\begin{equation}
\mathcal{X}:\,\, w^2(z)=P(z), \qquad	P(z)=\prod_{j=0}^{n}(z-E_j)(z-\widehat{E}_j),
\end{equation}
with cuts along the arcs $\Gamma_j=(E_j,\widehat{E}_j)$ connecting the branch points
$E_j$ and $\widehat{E}_j$, $j=0,1,2,\ldots,n$. The Riemann surface
$\mathcal{X}$ can be viewed as a double covering of the complex $z$-plane:
two copies of the $z$-plane are glued along the cuts $\Gamma_j$. The upper and
lower sheets of $\mathcal{X}$ are denoted by $\mathcal{X}_+$ and
$\mathcal{X}_-$, respectively, and are fixed by the relations
\begin{equation}
	\sqrt{P(z)}
	=\pm z^{n+1}\left(1+\mathcal{O}(z^{-1})\right),
	\qquad
	z=\pi(\mathcal{P})\to\infty,
	\quad
	\mathcal{P}\in\mathcal{X}_{\pm},
\end{equation}
where $z=\pi(\mathcal{P})$ is the standard projection of
$\mathcal{P}=(z,w)\in\mathcal{X}$ onto the Riemann sphere $\mathbb{P}^{1}$.
Thus each point in the $z$-plane, except for the branch points, has two
preimages $\mathcal{P}_{\pm}\in\mathcal{X}_{\pm}$. We use $\infty^{\pm}$ to stand for the preimages of
$z=\infty$ on $\mathcal{X}_{\pm}$, respectively. After
adjoining the two points $\infty^{+}$ and $\infty^{-}$, $\mathcal{X}$ becomes
a compact Riemann surface of genus $n$. The function $\sqrt{P(z)}$ then becomes
a meromorphic function on the compact Riemann surface $\mathcal{X}$, with
$2n+2$ zeros at $E_j$ and $\widehat{E}_j$, $j=0,1,2,\ldots,n$, and two poles at
$\infty^{+}$ and $\infty^{-}$, each of multiplicity $n+1$.

According to  the new spatial variable (\ref{eq:y-intro}), the vector-valued Baker-Akhiezer function $\psi(y,t; P)=(\psi_1, \psi_2)^T$ for the WKI equation (\ref{WKI}) is defined by the following conditions \cite{Belokolos1994}:
\begin{itemize}
	\item $\psi(y,t; P)$ is meromorphic on $\mathcal{X}\setminus(\{\infty^{-}\}\cup\{\infty^{+}\})$ with pole divisor
	 $
		\mathcal{D}=\sum_{j=1}^n P_j,$
		which is non-special in the sense that $\pi(P_j)\neq E_j,\widehat{E}_j,\quad \pi(P_j)\neq \pi(P_k),\quad j\neq k.$
		
	\item $\psi(y,t; P)$ has the following asymptotic behaviors:
	\begin{equation}
	 \psi(y,t; P)=\left\{\begin{array}{ll}
    \d\left[(1, 0)^T+\mathcal{O}(z^{-1})\right]e^{-izy-2iz^2t},
		& z=\pi(P),\quad P\to\infty^{-}, \v\\
	\d cz\left[(0, 1)^T+\mathcal{O}(z^{-1})\right]e^{izy+2iz^2t},
		& z=\pi(P),\quad P\to\infty^{+},
	\end{array}\right.
 \end{equation}
	where $c$ is a complex constant.
\end{itemize}

\begin{figure}[!t]
	\centering
	\includegraphics[scale=0.3]{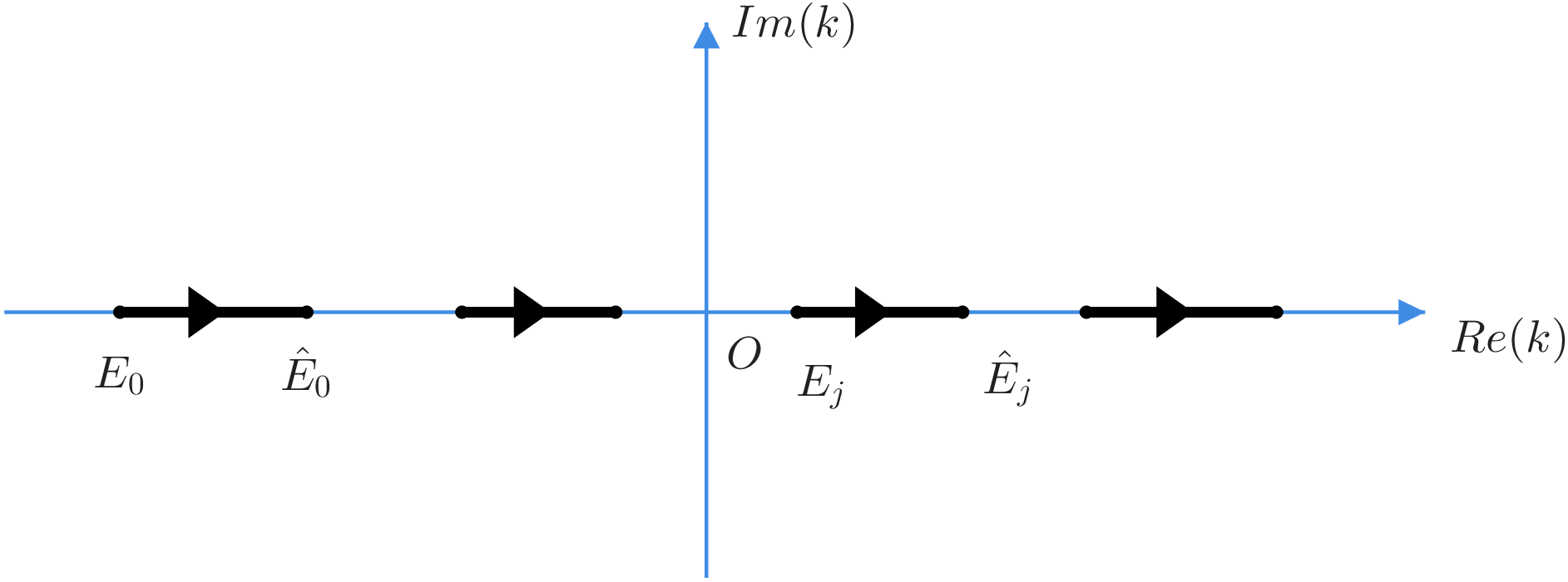}
	\caption{The jump contour on the $z$-plane for the defocusing WKI equation.}
	\label{fig:WKI1}
\end{figure}


Next, we recall the RH formalism for the WKI equation \cite{Li2022a,Li2022b,Liu2017,shima2016}. The RH problem corresponding to the solution of the WKI equation under zero boundary conditions is determined by a piecewise smooth simple curve $\Gamma$ and a $2\times2$ matrix-valued jump function on $\Gamma$, as follows.

\begin{RH}\label{widetildeM}
	The function $\widetilde{M}(z)$ satisfies the following RH problem:

	\begin{enumerate}[label=(\textbf{$\widetilde{M}$\arabic*}), leftmargin=*]
		\item \textbf{Analyticity:} $\widetilde{M}(z)=\widetilde{M}(y,t;z)$ is analytic in $\mathbb{C}\setminus\Gamma$, where $\Gamma$ is a piecewise smooth oriented curve in the complex plane.
		
		\item \textbf{Jump condition:} $\widetilde{M}(z)$ satisfies the jump relation
		\bee
		\widetilde{M}_-(z)=\widetilde{M}_+(z)e^{-i\widetilde\theta(z)\widehat{\sigma}_3}J(z), \quad z\in \Gamma,
		\ene
		where $e^{\widehat{\sigma}_3}A=e^{\sigma_3}Ae^{-\sigma_3}$ and
		$\widetilde\theta(z)=\widetilde\theta(y,t;z)=zy+2z^2t$ with $y$ being defied by \eqref{eq:y-intro}.

		\item \textbf{Asymptotic behaviors:}
		\begin{equation}\label{widetildeM Asymptotic}
			\widetilde{M}(z)=\left\{\begin{array}{ll}
          I+\mathcal{O}(z^{-1}),& z\to\infty, \v\\
		 		\widetilde{M}_{0}(y,t)
			\left[
			I+zH(x(y,t),t)+O(z^{2})
			\right], &  z\to 0,
          \end{array}\right.
		\end{equation}
		where
		\begin{equation}\no
			\widetilde{M}_{0}(y,t)
			=
			e^{\eta(x(y,t),t)\sigma_{3}}
			G^{-1}(x(y,t),t),\quad H=\int_{-\infty}^{x}[Q(s,t)+\mathrm{i}\bigl(p(s,t)-1\bigr)\sigma_{3}]\,ds,
		\end{equation}
			\begin{equation} \no
			G(x,t)
			=
			\sqrt{\frac{p+1}{2p}}
			\begin{pmatrix}
				1 & \mathrm{i}(1-p)/\bar{q} \v\\
				\mathrm{i}(1-p)/q & 1
			\end{pmatrix},\quad
				\eta(x,t)
   =\int_{-\infty}^{x}
			\frac{q\bar{q}_{s}-q_{s}\bar{q}}{4p(p+1)}\,ds.
		\end{equation}
		\item \textbf{Local behavior:} $\widetilde{M}(z)$ has inverse fourth-root singularities at the endpoints of $\Gamma$.
	\end{enumerate}
\end{RH}

Combining \eqref{widetildeM Asymptotic} with \eqref{eq:y-intro}, we obtain the following result:
\begin{proposition}
	Suppose that $\widetilde{M}(y,t;z)$ satisfies the Riemann--Hilbert problem \ref{widetildeM}. Then the solution $q(x,t)$ of the WKI equation \eqref{WKI} is given by the following reconstruction formulae:
	\begin{equation} \label{q}
		q(x,t)=\lim_{z\to 0}
		\frac{\partial_y\left[\widetilde{M}(y,t;0)^{-1}\widetilde{M}(y,t;z)\right]_{12}}{z}\left(1+\mathrm{i}\lim_{\substack{z\to 0}}
		\frac{\partial_y\left[\widetilde{M}(y,t;0)^{-1}\widetilde{M}(y,t;z)\right]_{11}}{z}\right)^{-1},
	\end{equation}
 with	\begin{equation} \label{x}
		x(y,t)=y+\mathrm{i}\lim_{\substack{z\to 0}}
		\frac{\left[\widetilde{M}(y,t;0)^{-1}\widetilde{M}(y,t;z)\right]_{11}-1}{z}.
	\end{equation}
\end{proposition}

\begin{proof}
	Similar to the Camassa-Holm equation \cite{Monvel2006}, one can introduce the improved transformation \cite{Li2022a,Li2022b}
	\begin{equation}
		\widetilde{M}(y,t;z)=G(x,t)e^{d_+\widehat{\sigma}_3}\mu(x,t;z)e^{-d_-\sigma_3}e^{\mathrm{i}(zy+2z^2t)\sigma_3},
	\end{equation}
	where
	\begin{equation} \no
		d_\mp=\pm
		\int_{\mp\infty}^{x}
		\frac{{\rm Im}(q\bar{q}_x)}{2p(p+1)}(s,t)\,ds,\qquad
		d=d_++d_-=
		\int_{-\infty}^{+\infty}
		\frac{{\rm Im}(q\bar{q}_x)}{2p(p+1)}(s,t)\,ds.
	\end{equation}
	We then define $\mathcal{K}(x,t;z)$ by
	\begin{equation}
		\mathcal{K}(x,t;z)
		=
		\mathrm{i}z\left[
		x+\int_{-\infty}^x(p(s)-1)\,ds
		\right]
		+2\mathrm{i}z^2t.	
	\end{equation}
	It follows that
	\begin{equation}
	\mathcal{K}_x(x,t;z)=\mathrm{i}zp,\qquad
	\mathcal{K}_t(x,t;z)
	=
	\left(
	2\mathrm{i}z^2
	+
	z\frac{{\rm Im}(q\bar{q}_x)}{p^2}
	\right),
	\end{equation}
	which are compatible, that is, $\mathcal{K}_{xt}=\mathcal{K}_{tx}$. This compatibility relation can be written as
	\begin{equation}
		\mathrm{i}p_t
		=
		\left(
		\frac{{\rm Im}(q\bar{q}_x)}{p^2}
		\right)_x.
	\end{equation}
	which is the conversation law of the WKI equation (\ref{WKI}) \cite{Wadati1979b}.
	
	Then the equivalent Lax pair for $\widetilde{\mu}$ in \eqref{Lax} can be written as
	\begin{equation} \label{mu-lax}
		\left\{\begin{array}{l}
			\mu_x+\mathcal{K}_x[\sigma_3,\mu]
			=
			-e^{-d_+\widehat{\sigma}_3}U\mu,\v\\
			\mu_t+\mathcal{K}_t[\sigma_3,\mu]
			=
			-e^{-d_+\widehat{\sigma}_3}V\mu,
		\end{array}\right.
	\end{equation}
	where
	\begin{equation} \no
		U=
		\begin{pmatrix}
			0
			&
			-\dfrac{\mathrm{i}q\left[p{\rm Im}(q\bar{q}_x)-|q|_x^2\right]}
			{4p^2(p^2-1)}
			\\[1.2ex]
			\dfrac{\mathrm{i}\bar{q}\left[p{\rm Im}(q\bar{q}_x)+|q|_x^2\right]}
			{4p^2(p^2-1)}
			&
			0
		\end{pmatrix},\qquad
			V=
		\begin{pmatrix}
			0 & v_{12} \v\\
			v_{21} & 0
		\end{pmatrix},
	\end{equation}
	with
	\begin{equation} \no
	v_{12}
	=
	\frac{
		\mathrm{i}q\left[\bar{q}{\rm Im}(q\bar{q}_x)-2\bar{q}_x(p-1)\right]
	}{2p^3(p-1)\bar{q}}\,z
	-
	\frac{
		\mathrm{i}q\left[\bar{q}{\rm Im}(q\bar{q}_t)-2\bar{q}_t(p-1)\right]
	}{4p^2(p-1)\bar{q}},
	\end{equation}
	\begin{equation}  \no
	v_{21}
	=
	\frac{
		\mathrm{i}\bar{q}\left[q{\rm Im}(q\bar{q}_x)+2q_x(p-1)\right]
	}{2p^3(p-1)q}\,z
	+
	\frac{
		\mathrm{i}\bar{q}\left[q{\rm Im}(q\bar{q}_t)+2q_t(p-1)\right]
	}{4p^2(p-1)q}.
	\end{equation}
	The Jost solutions of (\ref{mu-lax}) is written as the Volterra integral equations
	\begin{equation}
	\mu_{\pm}=I-\int_{\pm \infty}^xe^{[\mathcal{K}(s,t;z)-\mathcal{K}(x,t;z)]\widehat{\sigma}_3}e^{-d_+(s,t)\widehat{\sigma}_3}U(s,t;z)\mu_{\pm}(s,t;z)ds.
	\end{equation}
	
	According the compatibility condition of (\ref{mu-lax}),  $\mu_{xt}= \mu_{tx}$, substituting the reconstruction formulae for $x(y,t)$ in \eqref{x} and $q(x,t)$ in \eqref{q}, one recovers the WKI equation (\ref{WKI}).
\end{proof}

The Baker-Akhiezer function, reformulated in terms of the Riemann-Hilbert problem, is defined by
\begin{equation}
	\Psi^{(AG)}(y,t;z)
	=
	\widetilde{M}(y,t;z)e^{-\mathrm{i}(zy+2z^2t)\sigma_3},
\end{equation}
as in Definition~\ref{BA}.

\begin{definition}[Baker--Akhiezer function]\label{BA}
	Let $\Gamma_j=(E_j,\widehat{E}_j)$, $j=0,1,\ldots,n$, be oriented arcs connecting $E_j$ and $\widehat{E}_j$; see Figs.~\ref{fig:WKI1} and~\ref{fig:WKI2}. A $2\times2$ matrix $\Psi^{(AG)}(y,t;z)$ is called the Baker-Akhiezer function for the WKI equation \eqref{WKI} if it satisfies the following Riemann--Hilbert problem:
	\begin{RH}\quad
		\begin{enumerate}[label=(\textbf{$\Psi^{(AG)}$\arabic*}), leftmargin=*]
			\item \textbf{Analyticity:} $\Psi^{(AG)}(y,t;z)$ is analytic in $\mathbb{C}\setminus\Gamma$, where
			 $\Gamma=\left(\cup_{j=0}^n\Gamma_j\right)\cup\left(\cup_{j=1}^n\widehat{\Gamma}_j\right)$,\,
				$\Gamma_j=(E_j,\widehat{E}_j),$ \, $\widehat{\Gamma}_j=(\widehat{E}_{j-1},E_j).$

			\item \textbf{Jump condition:} $\Psi^{(AG)}_y{\Psi^{(AG)}}^{-1}$ and $\Psi^{(AG)}_t{\Psi^{(AG)}}^{-1}$ are entire functions. Moreover, $\Psi^{(AG)}$ satisfies
			\begin{equation} \no
				\Psi^{(AG)}_-(y,t;z)=\Psi^{(AG)}_+(y,t;z)J(z),\quad z\in\Gamma,
			\end{equation}
			where $J(z)$ is a piecewise constant matrix on $\Gamma$:
			
			For Case A,
			\begin{equation}\label{caseA}
				J(z)=
				\begin{pmatrix}
					0 & ie^{i\phi_j}\\
					ie^{-i\phi_j} & 0
				\end{pmatrix},
				\quad z\in(E_j,\widehat{E}_j),\quad j=0,1,\ldots,n.
			\end{equation}
			
			For Case B,
			\begin{equation}\label{caseB}
				J(z)=
				\begin{cases}
					i\sigma_1, 
					& z\in\Gamma_j=(E_j,\widehat{E}_j),\quad j=0,1,\ldots,n,\\[1.2em]
					e^{i\widehat{\phi}_j\sigma_3}
					& z\in\widehat{\Gamma}_j=(\widehat{E}_{j-1},E_j),\quad j=1,\ldots,n.
				\end{cases}
			\end{equation}
			
			Here $\phi_j$ and $\widehat{\phi}_j$ are real constants, with $\phi_0=0$.
			
			\item \textbf{Normalization at infinity:} $\Psi^{(AG)}(x,t;z)=\left(I+\mathcal{O}(z^{-1})\right)
				e^{-i(zy+2z^2t)\sigma_3}$ as $z\to\infty$.
			
			\item \textbf{Local behavior:} $\Psi^{(AG)}(y,t;z)$ has inverse fourth-root singularities at $E_j$ and $\widehat{E}_j$, $j=0,\ldots,n$.
			
			\item \textbf{Symmetry condition:} $\Psi^{(AG)}(y,t;z)$ satisfies the following symmetry relations:
			\begin{align} \no
			\text{Defocusing WKI equation}:\quad &\Psi^{(AG)}(y,t;z)=\sigma_1\overline{\Psi^{(AG)}(y,t;\bar{z})}\sigma_1,\\
			\text{Focusing WKI equation}: \quad &	\Psi^{(AG)}(y,t;z)=\sigma_2\overline{\Psi^{(AG)}(y,t;\bar{z})}\sigma_2. \no
			\end{align}
		\end{enumerate}
	\end{RH}
\end{definition}

\begin{figure}[!t]
	\centering
	\includegraphics[scale=0.3]{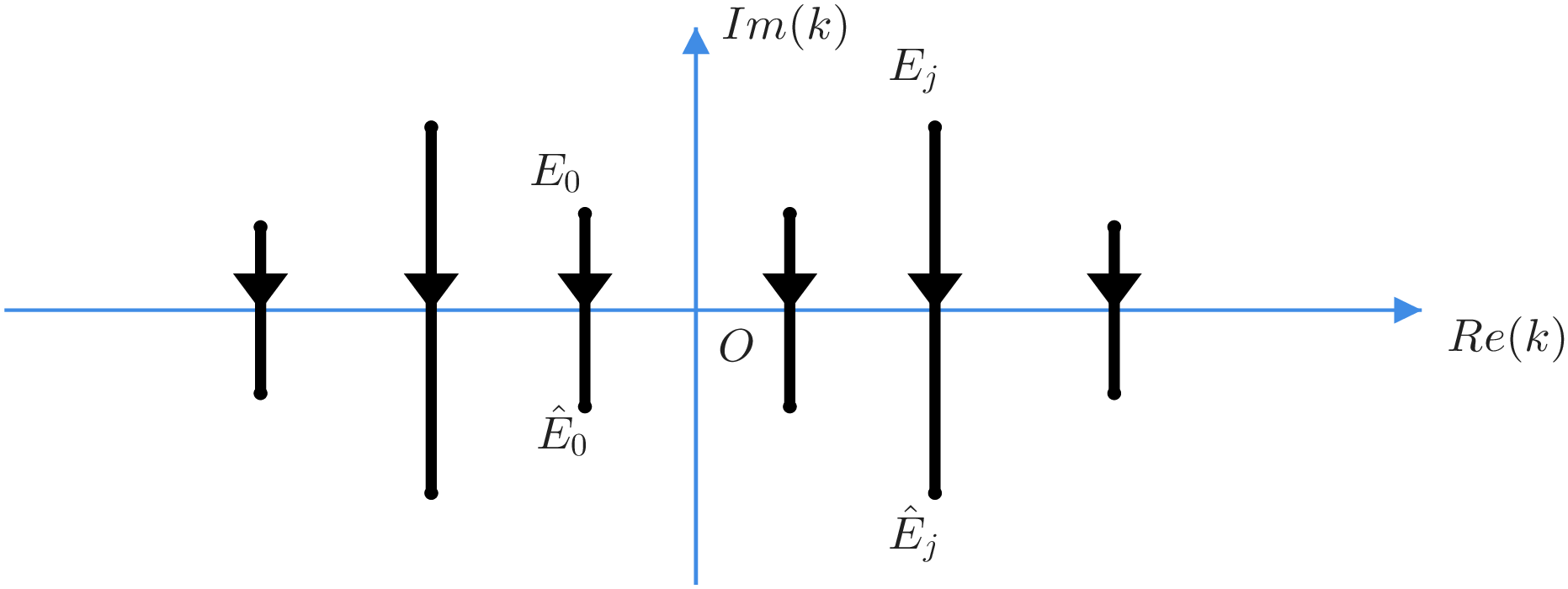}
	\caption{The jump contour in the $z$-plane for the focusing WKI equation.}
	\label{fig:WKI2}
\end{figure}

\section{The finite-genus algebro-geometric solutions}\label{Representation}

In this section, we construct the finite-genus algebro-geometric solutions of the WKI equation (\ref{WKI}), that is, we give the proof of Theorem \ref{theom}.
The existence of the Baker--Akhiezer function is proved by an explicit construction of $\Psi^{(AG)}$, obtained by solving the original RH problem. To transform the initial RH problem into a form that admits an explicit solution, we seek $\Psi^{(AG)}(y,t,z)$ in the form
\begin{equation} \label{psi-AG}
	\Psi^{(AG)}(y,t,z)
	=
	e^{\mathrm{i}(f_0y+g_0t)\sigma_3}
	M(y,t,z)
	e^{-\mathrm{i}(f(z)y+g(z)t)\sigma_3},	
\end{equation}
where the constants $f_0$, $g_0$, the scalar functions $f(z)$, $g(z)$, and the matrix $M(y,t,z)$ are to be determined. From the definition of $\Psi^{(AG)}$, the functions $f$ and $g$ satisfy the following Riemann--Hilbert problem.

\begin{RH}
	The functions $f(z)$ and $g(z)$ satisfy the following RH problem:
	\begin{enumerate}[label=(\textbf{$f$-$g$-\arabic*}), leftmargin=*]
		\item \textbf{Analyticity:} $f(z)$ and $g(z)$ are analytic in
		$\mathbb{C}\setminus\Gamma$.
      		
		\item \textbf{Jump condition:} $f(z)$ and $g(z)$ satisfy the jump conditions
		\begin{equation}
		 \left\{
         \begin{array}{lll}	f_+(z)+f_-(z)=C_j^f, &	g_+(z)+g_-(z)=C_j^g, &
		 z\in\Gamma_j,\,j=0,\ldots,n, \v\\
			f_+(z)-f_-(z)=\widehat{C}_j^f, &			g_+(z)-g_-(z)=\widehat{C}_j^g, &
			z\in\widehat{\Gamma}_j,\, j=1,\ldots,n,
        \end{array}\right.
        \end{equation}
       where $C_j^f,\,C_j^g,\,\widehat{C}_j^f,\, \widehat{C}_j^g$ are real constants.
		
		\item \textbf{Asymptotic behavior:} As $z\to\infty$, $f(z)$ and $g(z)$ have the asymptotic behaviors:
		\begin{equation}\label{f}
			f(z)=z+f_0+\mathcal{O}\left(z^{-1}\right), \qquad
			g(z)=2z^2+g_0+\mathcal{O}\left(z^{-1}\right).
		\end{equation}

\item \textbf{Symmetry relations:} $\overline{f(\overline{z})}=f(z)$,\quad $\overline{g(\overline{z})}=g(z)$.
	\end{enumerate}
\end{RH}

\subsection{Constructions of $f(z)$ and $g(z)$}
In this subsection, we construct the phase functions on the $z$-plane. The phase functions $f(z)$ and $g(z)$ are scalar functions defined by the Cauchy-type integrals \cite{Kotlyarov2017}
\begin{equation}
	f(z)=
	\begin{cases}
		\displaystyle
		\omega(z)\left(1+\frac{1}{2\pi\mathrm{i}}
		\int_{\Gamma_0}
		\frac{C_0^f}{\omega_+(\xi)(\xi-z)}\,d\xi\right),
		& n=0,\\[2ex]
		\displaystyle
		\frac{w(z)}{2\pi\mathrm{i}}
		\left(
		\sum_{j=0}^{n}\int_{\Gamma_j}
		\frac{C_j^f}{w_+(\xi)(\xi-z)}\,d\xi
		+
		\sum_{l=1}^{n}\int_{\widehat{\Gamma}_l}
		\frac{\widehat{C}_l^f}{w(\xi)(\xi-z)}\,d\xi
		\right),
		& n\geq 1,
	\end{cases}
\end{equation}
and
\begin{equation}
	g(z)=
	\begin{cases}
		\displaystyle
		w(z)\left(
		2z+(E_0+\widehat{E}_0)
		+\frac{1}{2\pi\mathrm{i}}\int_{\Gamma_0}
		\frac{C_0^g}{w_+(\xi)(\xi-z)}\,d\xi
		\right),
		& n=0,\\[3ex]
		\displaystyle
		w(z)\left(
		2+\frac{1}{2\pi\mathrm{i}}
		\left(
		\sum_{j=0}^{1}\int_{\Gamma_j}
		\frac{C_j^g}{w_+(\xi)(\xi-z)}\,d\xi+\sum_{j=0}^{1}\int_{\widehat\Gamma_j}
		\frac{\widehat{C}_j^g}{w(\xi)(\xi-z)}\,d\xi
		\right)
		\right),
		& n=1,\\[3ex]
		\displaystyle
		\frac{w(z)}{2\pi\mathrm{i}}
		\left(
		\sum_{j=0}^{n}\int_{\Gamma_j}
		\frac{C_j^g}{w_+(\xi)(\xi-z)}\,d\xi+\sum_{j=0}^{n}\int_{\widehat\Gamma_j}
		\frac{\widehat{C}_j^g}{w(\xi)(\xi-z)}\,d\xi
		\right),
		& n\geq 2,
	\end{cases}
\end{equation}
where $C_j^f$, $C_j^g$, $j=0,1,\ldots,n$, and $\widehat{C}_j^f$, $\widehat{C}_j^g$, $j=1,2,\ldots,n$, must satisfy the system of linear equations generated by the asymptotic behaviors \eqref{f}  as $z\to\infty$:
\begin{equation}
	\sum_{j=0}^{n}\int_{\Gamma_j}
	\frac{\xi^m C_j^f}{w_+(\xi)}\,d\xi
	+
	\sum_{l=1}^{n}\int_{\widehat{\Gamma}_l}
	\frac{\xi^m \widehat{C}_l^f}{w(\xi)}\,d\xi
	=
	-2\pi\mathrm{i}\delta_{m,n-1},
\end{equation}
for $m=0,1,2,\ldots,n-1$. Thus we have $n$ equations for $3n+2$ unknowns. An analogous system is obtained for the function $g(z)$:
\bee
\left\{\begin{array}{l}
	\d\sum_{j=0}^{n}
	\int_{\Gamma_j}\frac{\xi^kC_j^g\,d\xi}{w_+(\xi)}+\sum_{l=1}^{n}
	\int_{\widehat\Gamma_l}\frac{\xi^k\widehat{C}_l^g\,d\xi}{w(\xi)}
	=0,
	\qquad k=0,\ldots,n-3, \v\\
	\d\sum_{j=0}^{n}
	\int_{\Gamma_j}\frac{\xi^{n-2}C_j^g\,d\xi}{w_+(\xi)}+\sum_{l=1}^{n}
	\int_{\widehat\Gamma_l}\frac{\xi^{n-2}\widehat{C}_l^g\,d\xi}{w(\xi)}
	=-4\pi\mathrm{i}, \v\\
	\d \sum_{j=0}^{n}
	\int_{\Gamma_j}\frac{\xi^{n-1}C_j^g\,d\xi}{w_+(\xi)}+\sum_{l=1}^{n}
	\int_{\widehat\Gamma_l}\frac{\xi^{n-1}\widehat{C}_l^g\,d\xi}{w(\xi)}
	=-2\pi\mathrm{i}\sum_{j=0}^{n}(E_j+\widehat{E}_j).
\end{array} \right.
\ene

To reduce the number of unknowns to the number of equations, it is natural to consider the following two cases:
\bee
\left\{
\begin{array}{l}
{\rm Case\,\, A:} \quad C_0^f=C_0^g=\widehat{C}_j^f=\widehat{C}_j^g=0,
\qquad j=1,2,\ldots,n; \v\\
{\rm Case\,\, B:}\quad C_j^f=C_j^g=0, \qquad j=0,1,2,\ldots,n.
\end{array}\right.
\ene

These two cases correspond to Cases A and B in Definition~\ref{BA}.

\subsection{Explicit construction of $M(y,t;z)$ in Case A and algebro-geometric solution}

In this subsection, we seek an explicit representation of $M(y,t,z)$ by solving the associated Riemann--Hilbert problem.

\begin{RH} \label{RH-A-M}
	In Case A, $M(y,t,z)$ is the solution of the following RH problem:
	\begin{enumerate}[label=(\textbf{$M$\arabic*}), leftmargin=*]
		\item \textbf{Analyticity:} $M(z)$ is analytic in $\mathbb{C}\setminus\Gamma_A$, where $\Gamma_A=\bigcup_{j=0}^{n}(E_j,\widehat{E}_j)$.
		
		\item \textbf{Jump condition:} $M(z)$ satisfies the jump relation $M_-(z)=M_+(z)J_A(z)$, $z\in\Gamma_A$, where the jump matrix is given by
		\begin{equation} \label{eq:M_jump}
		J_A(z)=\left\{
           \begin{array}{ll}
            i\sigma_1
, & z\in(E_0,\widehat{E}_0), \v\\
			\begin{pmatrix}
				0 & \mathrm{i}e^{-\mathrm{i}yC_j^f-\mathrm{i}tC_j^g-\mathrm{i}\phi_j}\\
				\mathrm{i}e^{\mathrm{i}yC_j^f+\mathrm{i}tC_j^g+\mathrm{i}\phi_j} & 0
			\end{pmatrix}, & z\in(E_j,\widehat{E}_j),\,\, j=1,2,\ldots,n.
        \end{array}\right.
		\end{equation}

		\item \textbf{Normalization at infinity:} As $z\to\infty$, $M(z)=I+\mathcal{O}(z^{-1})$.
		
		\item \textbf{Local behavior:} $M(z)$ is bounded on $\Gamma_A$ except at the branch points $E_j$ and $\widehat{E}_j$, where it has at most inverse fourth-root singularities.
	\end{enumerate}
\end{RH}

First, define
\begin{equation}
	\varkappa(z)=
	\prod_{j=0}^{n}
	\sqrt[4]{\frac{z-E_j}{z-\widehat{E}_j}},
	\qquad
	z\in\mathbb{C}\setminus\Gamma_A,
\end{equation}
where the cuts are taken along $(E_j,\widehat{E}_j)$, $j=0,\ldots,n$, and the branch is fixed by the condition $\varkappa(\infty)=1$. Then
\begin{equation}
	\varkappa_-(z)=\mathrm{i}\varkappa_+(z),
	\qquad z\in\Gamma_A,
	\tag{3.1}
\end{equation}
and
\bee
\varkappa(z)=\left\{\begin{array}{ll}
 (z-\widehat{E}_j)^{-1/4}+\mathcal{O}(1), & z\to\widehat{E}_j, \v\\
 \d 1+\sum_{j=0}^{n}\frac{\widehat{E}_j-E_j}{4z}+\mathcal{O}(z^{-2}),& z\to\infty.
 \end{array}\right.
 \ene
We choose the $\mathbf{a}_l$-cycles to be ovals on the upper sheet of $\mathcal{X}$ encircling the intervals $(E_l,\widehat{E}_l)$, $l=1,2,\ldots,n$. The $\mathbf{b}_l$-cycles start from $(E_0,\widehat{E}_0)$, go to the interval $(E_l,\widehat{E}_l)$ on the upper sheet, and then return to the starting point on the lower sheet; see Figs.~\ref{fig:WKI3} and~\ref{fig:WKI4}.

\begin{figure}[!t]
	\centering
	\includegraphics[scale=0.3]{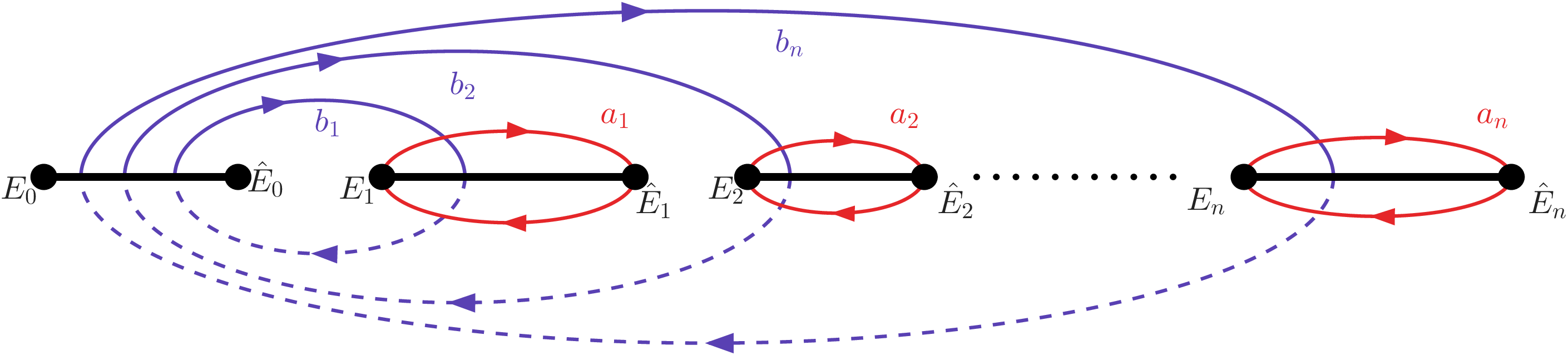}
	\caption{Case A: cuts on the real axis for the defocusing WKI equation.}
	\label{fig:WKI3}
\end{figure}

\begin{figure}[!t]
	\centering
	\includegraphics[scale=0.3]{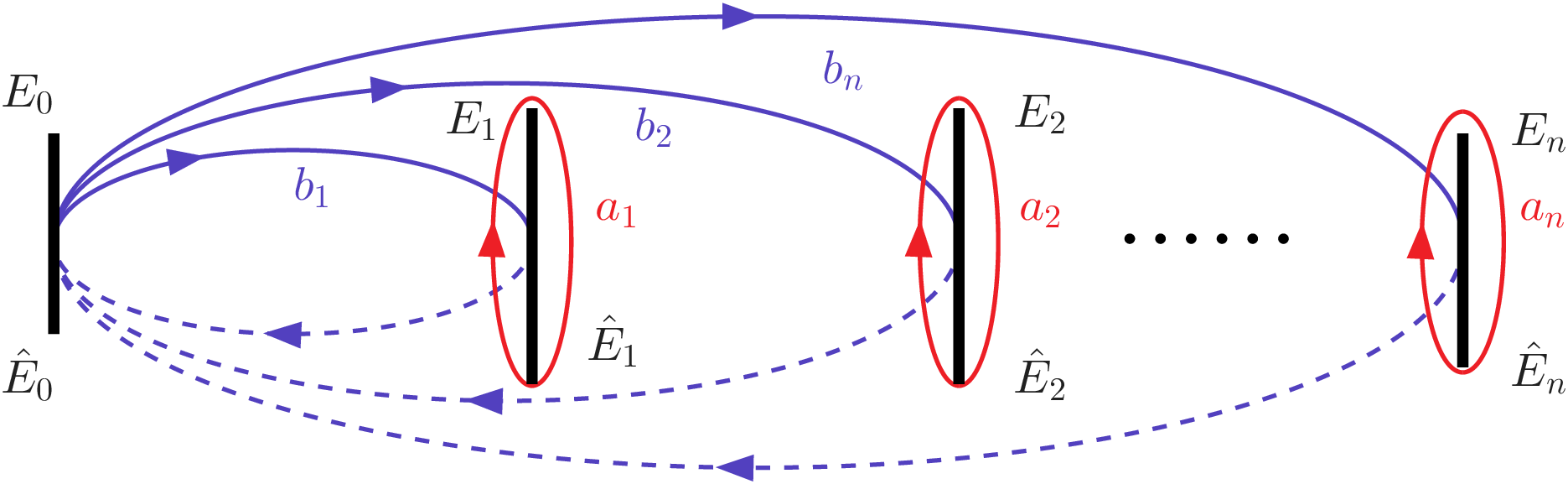}
	\caption{Case A: vertical cuts for the focusing WKI equation.}
	\label{fig:WKI4}
\end{figure}

Next, we introduce the Abelian integrals
\begin{equation}
	\omega_j(z)=\int_{\widehat{E}_0}^{z}\psi_j(s)\,ds,\quad j=1,2,\ldots,n,
\end{equation}
where $d\omega_j(\mathcal{P})$ form a basis of holomorphic differentials on $\mathcal{X}$:
\begin{equation}
\no
	\psi_j(z)=
	\frac{\displaystyle\sum_{i=1}^{n}c_{ji}z^{n-i}}
	{\sqrt{P(z)}}.
\end{equation}
The coefficients $c_{jl}$ are uniquely determined by the normalization relations
\begin{equation}
	\int_{\mathbf{a}_l}d\omega_j(\mathcal{P})=2\int_{E_l}^{\widehat{E}_l}\psi_{j+}(z)\,dz=\delta_{jl},\quad j,l=1,2,\ldots,n.
\end{equation}

The normalized holomorphic differentials define the $b$-period matrix
\begin{equation}
	B_{jl}=\int_{\mathbf{b}_l} d\omega_j(\mathcal{P})=2\sum_{k=1}^{l}\int_{\widehat{E}_{k-1}}^{E_k}\psi_j(z)\,dz.
\end{equation}
This matrix is symmetric and has positive definite imaginary part.

Associated with the matrix $B$ is the Riemann theta function, defined for $\mathbf{u}\in\mathbb{C}^n$ by the Fourier series
\begin{equation}
	\Theta(u_1,\ldots,u_n)
	=
	\sum_{\mathbf{l}\in\mathbb{Z}^n}
	\exp\left\{
	\pi\mathrm{i}(B\mathbf{l},\mathbf{l})
	+2\pi\mathrm{i}(\mathbf{l},\mathbf{u})
	\right\},
\end{equation}
where $(\mathbf{l},\mathbf{u})=l_1u_1+\cdots+l_nu_n$. The theta function is even, i.e., $\Theta(-\mathbf{u})=\Theta(\mathbf{u})$, and satisfies the periodicity relations
\begin{equation}
	\Theta(\mathbf{u}\pm\mathbf{e}_j)=\Theta(\mathbf{u}),
	\qquad
	\Theta(\mathbf{u}\pm B\mathbf{e}_j)
	=
	e^{\mp 2\pi\mathrm{i}u_j-\pi\mathrm{i}B_{jj}}\Theta(\mathbf{u}),
\end{equation}
where $\mathbf{e}_j=(0,\ldots,1,\ldots,0)$ is the $j$-th unit vector in $\mathbb{C}^n$, and $B\mathbf{e}_j$ is the $j$-th column of $B$.

Let $\Lambda\subset\mathbb{C}^n$ denote the lattice generated by integer linear combinations of the vectors $\mathbf{e}_j$ and $B\mathbf{e}_j$, $j=1,2,\ldots,n$. The Jacobian variety of $\mathcal{X}$ is then the complex torus $\operatorname{Jac}\{\mathcal{X}\}=\mathbb{C}^n/\Lambda$. The Abel map $\mathbf{A}:\mathcal{X}\to\operatorname{Jac}\{\mathcal{X}\}$ is defined by
\begin{equation}
	A_j(\mathcal{P})=\int_{\mathcal{P}_0}^{\mathcal{P}}d\omega_j(\mathcal{Q}),\quad j=1,2,\ldots,n,
\end{equation}
where the point $\mathcal{P}_0$ is fixed by the condition $\pi(\mathcal{P}_0)=\widehat{E}_0$, and $\mathcal{Q}$ denotes the integration variable. The Abel map is also extended to integral divisors $\mathcal{D}=\mathcal{P}_1+\cdots+\mathcal{P}_m$ by summation:
\begin{equation}
	\mathbf{A}(\mathcal{D})=\mathbf{A}(\mathcal{P}_1)+\cdots+\mathbf{A}(\mathcal{P}_m).
\end{equation}
For any non-special integral divisor $\mathcal{D}=\mathcal{P}_1+\cdots+\mathcal{P}_n$ of degree $n$, there exists a vector $\mathbf{w}(\mathcal{D})$ such that the Riemann theta function $\Theta(\mathbf{A}(\mathcal{P})+\mathbf{w}(\mathcal{D}))$, defined on $\mathcal{X}$ cut along the cycles $\mathbf{a}_j$ and $\mathbf{b}_j$, has precisely $n$ zeros at $\mathcal{P}_j$, $j=1,\ldots,n$. The vector $\mathbf{w}(\mathcal{D})$ is defined by
\begin{equation}
	\mathbf{w}(\mathcal{D})=-\mathbf{A}(\mathcal{D})-\mathbf{K},
\end{equation}
where the Riemann constant vector $\mathbf{K}$ is defined componentwise, modulo the lattice $\Lambda$, by
\begin{equation}  \no
	K_j=
	\frac{1}{2}
	-\frac{B_{jj}}{2}
	+
	\sum_{\substack{l=1\\ l\ne j}}^{n}
	\int_{\mathbf{a}_l}
	\left(
	\int_{\mathcal{P}_0}^{\mathcal{Q}}\omega_j
	\right)\omega_l.
\end{equation}
In the hyperelliptic case,
\begin{equation}
	K_j=\frac{1}{2}\sum_{l=1}^{n}B_{lj}-\frac{j}{2}.
\end{equation}

The Abelian integrals $\mathbf{A}(z)$, considered on the upper sheet of $\mathcal{X}$, have the following properties:
\bee
 \begin{array}{lll}
	\mathbf{A}_-(z)-\mathbf{A}_+(z)	=0,	&	z\in(\widehat{E}_{j-1},E_j), &	j=1,\ldots,n, \v\\
	\mathbf{A}_-(z)+\mathbf{A}_+(z)=B\mathbf{e}_j,	&	z\in(E_j,\widehat{E}_j),&	j=0,1,\ldots,n.
\end{array}
\ene
Now let
\begin{equation}
	F_s(z)
	=
	\frac{\theta(\mathbf{A}(z)+\mathbf{c}+\mathbf{d}_s)}
	{\theta(\mathbf{A}(z)+\mathbf{d}_s)},
	\quad
	H_s(z)
	=
	\frac{\theta(-\mathbf{A}(z)+\mathbf{c}+\mathbf{d}_s)}
	{\theta(-\mathbf{A}(z)+\mathbf{d}_s)},
	\quad
	z\in\mathbb{C}\setminus\Gamma_A,\quad s=1,2,
\end{equation}
where $	\mathbf{c}=(c_0,c_1,\ldots,c_j,\ldots,c_n),\quad j=0,1,\ldots,n,$
and $\mathbf{d}_1=-\mathbf{w}(\mathcal{D}),	\quad	\mathbf{d}_2=\mathbf{w}(\mathcal{D}).$
Then, by the jump relations for the Abel map, we obtain
\bee
\begin{array}{lll}
	F_{s-}(z)=F_{s+}(z),& H_{s-}(z)=H_{s+}(z), & z\in(\widehat{E}_{j-1},E_j),\,\,\, s=1,2;\,\, j=1,\ldots,n, \v\\
	F_{s-}(z)=e^{-2\pi\mathrm{i}c_j}H_{s+}(z),& H_{s-}(z)=e^{2\pi\mathrm{i}c_j}F_{s+}(z),&
		z\in(E_j,\widehat{E}_j),\,\,\, s=1,2;\,\, j=0,1,\ldots,n,
\end{array}
\ene
where $c_0=0$.

Next, define the $2\times2$ matrix-valued function
\begin{equation}
	N(z)=
	\begin{pmatrix}
		F_1(z) & H_1(z)\\
		F_2(z) & H_2(z)
	\end{pmatrix}.
\end{equation}
Then $N(z)$ has the jump relation: $N_-(z)=N_+(z)J_{AN}(z)$,  where the jump matrix $J_{AN}(z)$ is defined as
\bee
 J_{AN}(z)=\left\{\begin{array}{ll}
	 I, &	z\in(\widehat{E}_{j-1},E_j),\,\,\, j=1,\ldots,n, \v\\
	\begin{pmatrix}
		0 & e^{2\pi\mathrm{i}c_j}\\
		e^{-2\pi\mathrm{i}c_j} & 0
	\end{pmatrix}, &	z\in(E_j,\widehat{E}_j),\,\,\, j=0,1,\ldots,n,
 \end{array}\right.
 \ene
with $c_0=0$.

Finally, we define the matrix-valued function $M(z)$ as \cite{Kotlyarov2017}
\begin{equation}\label{M}
	\begin{aligned}
		M(z)=
		\frac{1}{2}
		\begin{pmatrix}
			F_1^{-1}(\infty) & 0\\
			0 & H_2^{-1}(\infty)
		\end{pmatrix}
		\begin{pmatrix}
			(\varkappa(z)+\varkappa^{-1}(z))F_1(z)
			&
			(\varkappa(z)-\varkappa^{-1}(z))H_1(z)
			\\[1ex]
			(\varkappa(z)-\varkappa^{-1}(z))F_2(z)
			&
			(\varkappa(z)+\varkappa^{-1}(z))H_2(z)
		\end{pmatrix}.
	\end{aligned}
\end{equation}

It is straightforward to verify that $M(z)$ is analytic in $\mathbb{C}\setminus\bigcup_{j=0}^{n}(E_j,\widehat{E}_j)$, satisfies $M(z)=I+\mathcal{O}\left(z^{-1}\right)$ as $z\to\infty$, and has the jumps
\begin{equation}
	M_-(z)
	=
	M_+(z)
	\begin{pmatrix}
		0 & \mathrm{i}e^{2\pi\mathrm{i}c_j}\\
		\mathrm{i}e^{-2\pi\mathrm{i}c_j} & 0
	\end{pmatrix},
	\qquad
	z\in(E_j,\widehat{E}_j),\quad j=0,1,\ldots,n.
\end{equation}
These jumps agree with the jump conditions of the Riemann--Hilbert problem \ref{RH-A-M} if we set
\begin{equation}
	c_j=-\frac{yC_j^f+tC_j^g+\phi_j}{2\pi},
	\qquad j=1,\ldots,n,
\end{equation}
recalling that $c_0=0$. Hence $M(y,t,z)$ is the solution of the Riemann--Hilbert problem.

By the reconstruction formulae, we can find the finite-genus algebro-geometric solution $q^{(AG)}(x,t)$ of the WKI equation (\ref{WKI})
\begin{equation} \label{AG-A}
	\begin{aligned}
		q^{(AG)}(x,t)=
		\lim_{z\to0}
		\frac{
			\partial_y\left(e^{2\mathrm{i}(f_0y+g_0t)}(M_{22}(y,t;0)M_{12}(y,t;z)-M_{12}(y,t;0)M_{22}(y,t;z))\right)
		}{z}\\
		\qquad\quad \times\left(1+\mathrm{i}\lim_{\substack{z\to 0}}
		\frac{\partial_y\left(M_{22}(y,t;0)M_{11}(y,t;z)-M_{12}(y,t;0)M_{21}(y,t;z)\right)}{z}\right)^{-1}
	\end{aligned}
\end{equation}
with
\begin{equation} \label{AG-A2}
	\begin{aligned}
		x(y,t)
		=y+\mathrm{i}\lim_{z\to0}
		\frac{
			M_{22}(y,t;0)M_{11}(y,t;z)-M_{12}(y,t;0)M_{21}(y,t;z)-1
		}{z}.
	\end{aligned}
\end{equation}
Here $M$ is defined in \eqref{M}.

\subsection{Explicit construction of $M(y,t;z)$ in Case B and algebro-geometric solution}

In the long-time asymptotic analysis below, we use Case A. Case B is included for completeness and for comparison with the focusing-type contour configuration.

In this subsection, we seek an explicit representation of $M(y,t;z)$ by solving the associated Riemann--Hilbert problem.

\begin{RH} \label{RH-B-M}
	In Case B, $M(y,t;z)$ is the solution of the following Riemann--Hilbert problem:
	\begin{enumerate}[label=(\textbf{$M$\arabic*}), leftmargin=*]
		\item \textbf{Analyticity:} $M(z)$ is analytic in $\mathbb{C}\setminus\Gamma$.
		
		\item \textbf{Jump condition:} $M(z)$ satisfies the jump relation $M_-(z)=M_+(z)J_B(z)$, $z\in\Gamma$, where the jump matrix is given by
		\begin{equation} \label{eq:M_jump2}
			J_B(z)=\left\{
           \begin{array}{ll}
			\mathrm{i}\sigma_1
             ,
			&
			z\in\Gamma_j=(E_j,\widehat{E}_j),
			\quad j=0,1,2,\ldots,n, \v\\
 			e^{-\mathrm{i}(y\widehat{C}_j^f+t\widehat{C}_j^g+\widehat{\phi}_j)\sigma_3}
       ,
			 &
			z\in\widehat{\Gamma}_j=(\widehat{E}_{j-1},E_j),\, j=1,2,\ldots,n.
          \end{array} \right.
          \end{equation}
		
		\item \textbf{Normalization at infinity:} As $z\to\infty$, $M(z)=I+\mathcal{O}(z^{-1})$.
		
		\item \textbf{Local behavior:} $M(z)$ is bounded on $\Gamma$ except at the branch points $E_j$ and $\widehat{E}_j$, where it has at most inverse fourth-root singularities.
	\end{enumerate}
\end{RH}

In Case B, we choose another homology basis. The $\mathbf{a}_l$-cycles start from $\widehat{E}_{l-1}$, go to $E_l$ on the upper sheet, and then return to the starting points on the lower sheet. The $\mathbf{b}_l$-cycles are ovals on the upper sheet enclosing the intervals $(E_0,\widehat{E}_{l-1})$, $l=1,2,\ldots,n$; see Figs.~\ref{fig:WKI5} and~\ref{fig:WKI6}.

\begin{figure}[!t]
	\centering
	\includegraphics[scale=0.35]{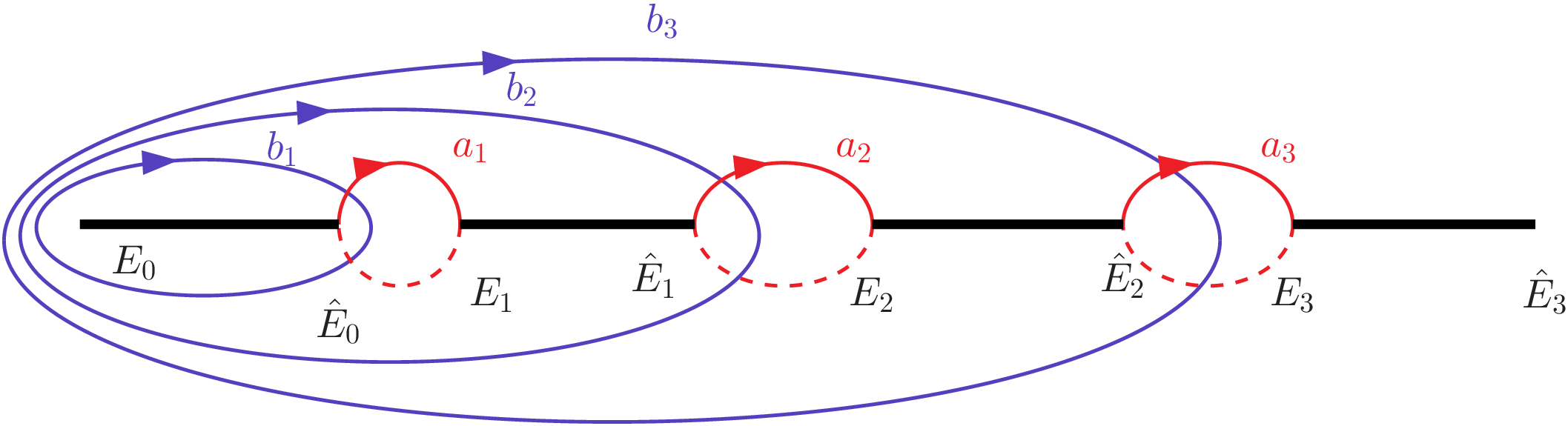}
	\caption{Case B: cuts on the real axis for the defocusing WKI equation.}
	\label{fig:WKI5}
\end{figure}

\begin{figure}[!t]
	\centering
	\includegraphics[scale=0.35]{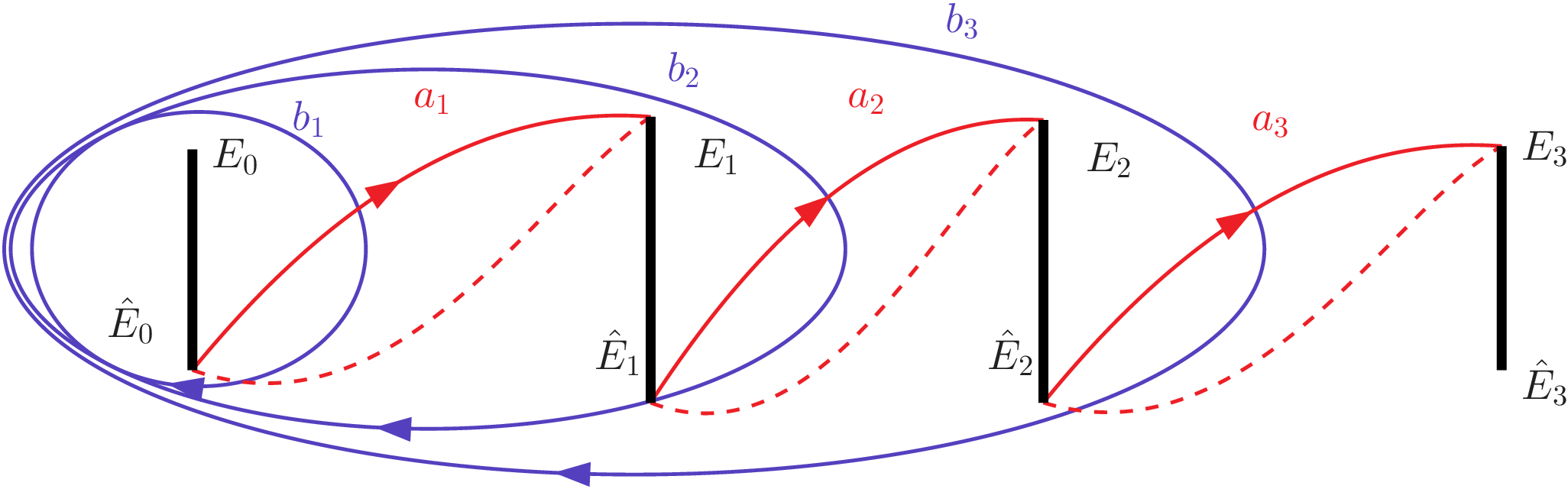}
	\caption{Case B: vertical cuts for the focusing WKI equation.}
	\label{fig:WKI6}
\end{figure}

The Abelian integrals $\omega_j(z)$,
\begin{equation}
	\omega_j(z)=\int_{E_0}^{z}\psi_j(s)\,ds,
	\qquad j=1,2,\ldots,n,
\end{equation}
are normalized by the relations
\begin{equation}
	\int_{\mathbf{a}_l}d\omega_j(\mathcal{P})
	=
	2\int_{\widehat{E}_{l-1}}^{E_l}\psi_j(z)\,dz
	=
	\delta_{jl},
	\qquad j,l=1,2,\ldots,n.
\end{equation}

The $\mathbf{b}$-period matrix is then defined by
\begin{equation}
	B_{jl}
	=
	\int_{\mathbf{b}_l}d\omega_j(\mathcal{P})
	=
	2\sum_{k=0}^{l-1}\int_{E_k}^{\widehat{E}_k}\psi_j(z_+)\,dz.
\end{equation}
and the Abel map $\mathbf{A}:\mathcal{X}\to Jac\{\mathcal{X}\}$ is defined by
\begin{equation}
	A_j(\mathcal{P})
	=
	\int_{\mathcal{P}_0}^{\mathcal{P}}\omega_j(\mathcal{Q}),
	\qquad j=1,2,\ldots,n,
\end{equation}
where the base point $\mathcal{P}_0$ is fixed differently, namely by $\pi(\mathcal{P}_0)=E_0$.

Similarly, the Abel maps $\mathbf{A}(z)$ on the upper sheet of $\mathcal{X}$ have the following properties:
\begin{align}
	\mathbf{A}_-(z)+\mathbf{A}_+(z)
	&=0 \pmod{\mathbb{Z}^n},
	&
	z&\in\Gamma_j=(E_j,\widehat{E}_j),
	&
	j&=1,\ldots,n,\\
	\mathbf{A}_-(z)-\mathbf{A}_+(z)
	&=-B\mathbf{e}_j,
	&
	z&\in\widehat{\Gamma}_j=(\widehat{E}_{j-1},E_j),
	&
	j&=1,\ldots,n,
\end{align}
where $\mathbf{e}_j=(0,\ldots,0,1,0,\ldots,0)$ is the $j$-th basis vector in $\mathbb{C}^n$. Define, for $s=1,2$,
\begin{equation}
	F_s(z)
	=
	\frac{\theta(\mathbf{A}(z)+\mathbf{c}+\mathbf{d}_s)}
	{\theta(\mathbf{A}(z)+\mathbf{d}_s)},
	\qquad
	H_s(z)
	=
	\frac{\theta(-\mathbf{A}(z)+\mathbf{c}+\mathbf{d}_s)}
	{\theta(-\mathbf{A}(z)+\mathbf{d}_s)},
	\qquad
	z\in\mathbb{C}\setminus\Gamma,
	\tag{3.16}
\end{equation}
where $\mathbf{c}=(c_0,c_1,\ldots,c_j,\ldots,c_n),\, j=0,1,\ldots,n,$ and
$\mathbf{d}_1=-\mathbf{w}(\mathcal{D}),	\quad	\mathbf{d}_2=\mathbf{w}(\mathcal{D}).$

From the jump relations for the Abel map, we obtain
\bee
 \begin{array}{lll}
 	F_{s-}(z)=H_{s+}(z),&
	H_{s-}(z)=F_{s+}(z),& z\in\Gamma_j,\,\, j=0,\ldots,n, \v\\
	F_{s-}(z)=e^{2\pi\mathrm{i}c_j}F_{s+}(z),	&
	H_{s-}(z)=e^{-2\pi\mathrm{i}c_j}H_{s+}(z),	& z\in\widehat{\Gamma}_j,,\,\, j=1,\ldots,n.
 \end{array}
 \ene
Next, define the $2\times2$ matrix-valued function
\begin{equation}
	N(z)=
	\begin{pmatrix}
		F_1(z) & H_1(z)\\
		F_2(z) & H_2(z)
	\end{pmatrix}.
\end{equation}
Then $N$ has the jump relation: $N_-(z)=N_+(z)J_{BN}(z)$,  where the jump matrix $J_{BN}(z)$ is defined as
\bee
 J_{BN}(z)=\left\{\begin{array}{ll}
	 \mathrm{i}\sigma_1
, &	z\in(E_j,\widehat{E}_j),\,\, j=0,1,\ldots,n, \v\\
	e^{-\mathrm{i}(y\widehat{C}_j^f+t\widehat{C}_j^g+\widehat{\phi}_j)\sigma_3}
, &	z\in(\widehat{E}_{j-1},E_j),\,\, j=1,2,\ldots,n.
 \end{array} \right.
 \ene

Let
\[
c_j:=-\frac{y\widehat{C}_j^f+t\widehat{C}_j^g+\widehat{\phi}_j}{2\pi},
\qquad j=1,\ldots,n.
\]
Proceeding as in Case A, we define $M$ by
\begin{equation}
	\begin{aligned}
		M(z)=
		\frac{1}{2}
		\begin{pmatrix}
			F_1^{-1}(\infty) & 0\\
			0 & H_2^{-1}(\infty)
		\end{pmatrix}
		\begin{pmatrix}
			(\varkappa(z)+\varkappa^{-1}(z))F_1(z)
			&
			(\varkappa(z)-\varkappa^{-1}(z))H_1(z)
			\\[1ex]
			(\varkappa(z)-\varkappa^{-1}(z))F_2(z)
			&
			(\varkappa(z)+\varkappa^{-1}(z))H_2(z)
		\end{pmatrix}.
	\end{aligned}
\end{equation}
In this case, the jumps of $M(z)$ take the form $M_-(z)=M_+(z)J_{BM}(z)$,  where the jump matrix $J_{BM}(z)=J_{BN}(z)$
as required in Case B. Thus, the function $M(y,t,z)$ is the solution of the RH problem \ref{RH-B-M}. The finite-genus algebro-geometric solution $q^{(AG)}(x,t)$ of the WKI equation \eqref{WKI} can be obtained similarly to Case A.

\section{The Jost functions and signature table}\label{JostFuc}

Next, we mainly consider the long-time asymptotics of the solution $q$ to the Cauchy problem for the defocusing WKI equation (\ref{WKI}) with the algebro-geometric background \eqref{background} under the perturbation condition \eqref{qRestrict} in Case A, where $q^{(AG)}(x,t)=q^{(AG)}(x(y,t),t)$ is given by (\ref{AG-A}) with $x=x(y,t)$ given by (\ref{AG-A2}).

Let $\mu^{\pm}$ be the solutions of the following Volterra integral equations:
\begin{equation} \label{mu-g}
	\begin{aligned}
		\mu^\pm(y,t;z) & =e^{i(f_0y+g_0t)\widehat{\sigma}_3}M(y,t;z) \\
		&\qquad +\int_{\pm\infty}^x\Upsilon(s,t,y;z)Q(q-q^{(AG)})(s,t)\mu^\pm(s,t;z)e^{-i(f(z)-f_0)(s-y)\sigma_3}\mathrm{d}s,
	\end{aligned}
\end{equation}
where [cf. (\ref{psi-AG})]
\begin{equation}  \no
	\begin{aligned}
		\Upsilon(s,t,y;z)&=\Psi^{(AG)}(y,t;z)\Psi^{(AG)}(s,t;z)^{-1} \\
		& =e^{i(f_0y+g_0t)\sigma_3}M(y,t;z)e^{if(z)(s-y)\sigma_3}M(s,t;z)^{-1}e^{-i(f_0s+g_0t)\sigma_3},
	\end{aligned}
\end{equation}
and $\mu^\pm(y,t;z)=\left(\mu_1^\pm(y,t;z),\mu_2^\pm(y,t;z)\right)$. Then one can show that the column vectors $\mu_1^-(y,t;z)$ and $\mu_2^+(y,t;z)$ exist uniquely for $z\in\overline{\mathbb{C}^+}\setminus\{E_j,\widehat{E}_j\}_{j=0}^n$, are analytic in $\mathbb{C}^+$, and admit continuous extensions to $z\in\overline{\mathbb{C}^+}\setminus\{E_j,\widehat{E}_j\}_{j=0}^n$. Similarly, the column vectors $\mu_1^+(y,t;z)$ and $\mu_2^-(y,t;z)$ exist uniquely for $z\in\overline{\mathbb{C}^-}\setminus\{E_j,\widehat{E}_j\}_{j=0}^n$, are analytic in $\mathbb{C}^-$, and admit continuous extensions to $z\in\overline{\mathbb{C}^-}\setminus\{E_j,\widehat{E}_j\}_{j=0}^n$.

Next, we define
\begin{equation}
	\Psi^\pm(y,t;z)=G(x(y,t),t)e^{\eta\widehat{\sigma}_3}\mu^\pm(y,t;z),
\end{equation}
which are two linearly independent solutions of the Lax pair \eqref{Lax}, where $G(x(y,t),t)$ is defined by (\ref{G}). Thus, there exists a scattering matrix $S(z)=(S_{ij}(z))_{2\times 2}$, independent of $(y,t)$, such that
\begin{equation}\label{psi-S}
	\Psi^+(y,t;z)=\Psi^-(y,t;z)S(z), \quad z\in\mathbb{R}\setminus(\cup_{j=0}^n[E_j,\widehat{E}_j]),
\end{equation}
and $\det S(z)=1$ for $z\in\mathbb{R}\setminus(\cup_{j=0}^n[E_j,\widehat{E}_j])$. It follows from (\ref{psi-S}) and $\Psi^{\pm}(y,t;z)
=(\Psi_1^{\pm}(y,t;z), \Psi_2^{\pm}(y,t;z))$ that
\bee \label{rel}
\begin{array}{l}
 S_{11}(z)=|\Psi_1^+(y,t;z), \Psi_2^-(y,t;z)|, \v\\
 S_{22}(z)=|\Psi_1^-(y,t;z), \Psi_2^+(y,t;z)|, \v\\
  S_{12}(z)=|\Psi_2^+(y,t;z), \Psi_2^-(y,t;z)|, \v\\
   S_{21}(z)=|\Psi_1^-(y,t;z), \Psi_1^+(y,t;z)|.
\end{array} \ene
Similarly, the scattering data $S_{11}(z)$ and $S_{22}(z)$ are analytic in $\mathbb{C}^-$ and $\mathbb{C}^+$, respectively, and admit continuous extensions to $\overline{\mathbb{C}^-}\setminus\{E_j,\widehat{E}_j\}_{j=0}^n$ and $\overline{\mathbb{C}^+}\setminus\{E_j,\widehat{E}_j\}_{j=0}^n$, respectively. The remaining entries generally do not admit analytic continuations to $\mathbb{C}^{\pm}$.

Moreover, we have the following symmetry relations:
\begin{equation}
	S_{11}(z)=\overline{S_{22}(\overline{z})},\quad S_{12}(z)=\overline{S_{21}(\overline{z})},\quad z\in\mathbb{C}\setminus(\cup_{j=0}^n[E_j,\widehat{E}_j]).
\end{equation}
We define the reflection coefficients in terms of the scattering matrix $S(z)$ by
\begin{equation}
	\label{r}
	r_1(z)=\frac{S_{12}(z)}{S_{11}(z)},\quad r_2(z)=\frac{S_{21}(z)}{S_{11}(z)}, \quad z\in\mathbb{R}\setminus(\cup_{j=0}^n[E_j,\widehat{E}_j]).
\end{equation}
Next, we define $\theta$ by
\begin{equation}
	\label{theta}
	\theta(z)=\theta(z;\xi)=-(f(z)-f_0)\xi-(g(z)-g_0),\quad\xi=\frac{y}{t}.
\end{equation}

\textbf{The signature table for $\operatorname{Im}\theta$.} Following \cite{Fan2026}, for the phase $\theta(z;\xi)$ defined in \eqref{theta}, a point $z=z(\xi)\in\mathbb{R}$ is called a saddle point of $\theta$ if it satisfies one of the following cases:
\begin{equation}\label{z}
	\begin{array}{ll}
	 {\rm Case\,\, i)}\quad \dfrac{\partial\theta(z;\xi)}{\partial z}=0,&z\in\mathbb{R}\backslash(\cup_{j=0}^n(E_j,\widehat{E}_j)); \v\\
		{\rm Case\,\, ii)} \quad\operatorname{Im}\theta(z;\xi)=0,&z\in\cup_{j=0}^n(E_j,\widehat{E}_j).
	\end{array}
\end{equation}

It follows from \eqref{theta} and (\ref{fg}) that
\begin{equation} \label{theta}
	\dfrac{\partial\theta(z;\xi)}{\partial z}=-w^{-1}(z)\left[\xi\prod_{j=0}^n(z-z_j^f)+4\prod_{j=0}^{n+1}(z-z_j^g)\right],
\end{equation}
where $w(z)$ is defined by (\ref{RSurface}), so that the zeros of $d\theta(z;\xi)/dz$ are precisely the roots of
\begin{equation} \label{T-f}
	T(z;\xi)=\xi\prod_{j=0}^n(z-z_j^f)+4\prod_{j=0}^{n+1}(z-z_j^g)=0.
\end{equation}

Since $f(E_j)=f(\widehat{E}_j)$ and $g(E_j)=g(\widehat{E}_j)$ for $j=0,1,\ldots,n$, thus we have $\theta(E_j;\xi)=\theta(\widehat{E}_j;\xi)$ from \eqref{theta}. This implies that $T(z;\xi)$ has at least one zero in each interval $(E_j,\widehat{E}_j)$. Since $T(z;\xi)$ is a polynomial of degree $n+2$, there remains one additional zero. If the zero lies outside the union of these intervals, namely in $\mathbb{R}\setminus(\cup_{j=0}^n(E_j,\widehat{E}_j))$, then it corresponds to Case $i)$. On the other hand, if it lies inside one of the intervals $(E_j,\widehat{E}_j)$, then it corresponds to Case $ii)$. This gives rise to one saddle point, denoted by $z_1$.

For the Case $i)$, let $z_1=z_1(\xi)\in\mathbb{R}\setminus(\cup_{i=0}^n(E_i,\widehat{E}_i))$. In view of the identities
\begin{equation}
	T(E_j;\xi_j)=T(\widehat{E}_j;\widehat{\xi}_j)=0,\quad j=0,1,\ldots,n,
\end{equation}
where $\xi_j$ and $\widehat{\xi}_j$ are given by \eqref{xi}.

By the implicit function theorem, taking the interval $(\widehat{E}_j,E_{j+1})$ as an example, from \eqref{T-f}, we obtain
\begin{equation}
	\frac{dz_1(\xi)}{d\xi}=-\frac{\partial_\xi T(z_1;\xi)}{\partial_{z_1}T(z_1;\xi)}=-\frac{\prod_{j=0}^n(z_1-z_j^f)}{\partial_{z_1}T(z_1;\xi)}<0.
\end{equation}

Indeed, if $n-j$ is even, then both $\partial_{z_1}T(z_1;\xi)$ and $\prod_{k=0}^n(z_1-z_k^f)$ are positive; if $n-j$ is odd, then both quantities are negative. Hence $\partial_\xi z_1(\xi)<0$. Therefore, $z_1(\xi)$ is a monotonically decreasing function of $\xi$ on each interval of $\mathbb{R}\setminus(\cup_{i=0}^n(E_i,\widehat{E}_i))$.

Similarly, in the Case $ii)$, $z_1(\xi)$ is a monotonically decreasing function of $\xi$ on $\cup_{j=0}^n(E_j,\widehat{E}_j)$. Moreover, $z_1(\xi)$ tends to $E_j$ and $\widehat{E}_j$ as $\xi$ tends to $\xi_j$ and $\widehat{\xi}_j$, respectively. In particular, we have $z_1(\widehat{\xi}_j)=\widehat{E}_j$ and $z_1(\xi_j)=E_j$.

\section{Asymptotic analysis in transition region I}\label{Region-I}
In this section, we carry out the asymptotic analysis of the Riemann--Hilbert problem for $\mathcal{M}$ in order to derive the asymptotics of $q$ in transition region I. According to Definition~ \ref{def1}, we may restrict ourselves to
$-C\le(\xi-\widehat{\xi}_{j_1})t^{2/3}\leq0$
for some fixed $j_1\in\{0,\ldots,n\}$, since the analysis in the complementary half-region is analogous. For large $t$, one has $\xi_{j_1+1}<\xi<\widehat{\xi}_{j_1}$, and hence the saddle point satisfies $z_1\in[\widehat{E}_{j_1},E_{j_1+1})$. Here $\mathcal{M}$ is defined by
\begin{equation}
	\begin{aligned}
		\mathcal{M}(z)=\mathcal{M}(y,t;z) & =
		\begin{cases}
			\left(\dfrac{\mu_1^-(y,t;z)}{S_{22}(z)},\mu_2^+(y,t;z)\right),\quad & z\in\mathbb{C}^+, \\[1.2em]
			\left(\mu_1^+(y,t;z),\dfrac{\mu_2^-(y,t;z)}{S_{11}(z)}\right),\quad & z\in\mathbb{C}^-.
		\end{cases}
	\end{aligned}
\end{equation}

\begin{RH}
	$\mathcal{M}(z)$ satisfies the following Riemann--Hilbert problem:
	\begin{enumerate}[label=(\textbf{$\mathcal{M}$\arabic*}), leftmargin=*]
	\item \textbf{Normalization at infinity:} $\mathcal{M}(z)=I+\mathcal{O}(z^{-1})$ as $z\to\infty$.

	\item \textbf{Analyticity:} $\mathcal{M}(z)$ is analytic in $\mathbb{C}\setminus\mathbb{R}$.
		
		\item \textbf{Jump condition:} $\mathcal{M}(z)$ satisfies the jump relation $\mathcal{M}_+(z)=\mathcal{M}_-(z)J_{\mathcal{M}}(z)$ for $z\in\mathbb{R}$, where the jump matrix is given by
		\begin{equation} \label{eq:mathcal{M}_jump}
			J_{\mathcal{M}}(z)=
			\begin{cases}
				e^{i(f_0y+g_0t-(B_j^fy+B_j^gt+\phi_j)/2)\widehat{\sigma}_3}
				\begin{pmatrix}
					0 & -i \\
					-i & 0
				\end{pmatrix},&z\in(E_j,\widehat{E}_j),\,\, j=0,\ldots,n, \\
				\\
				\begin{pmatrix}
					1-|r_1(z)|^2 & r_1(z)e^{2it\theta(z)} \v\\
					-\overline{r_1(z)}e^{-2it\theta(z)} & 1
				\end{pmatrix},& z\in\mathbb{R}\setminus(\cup_{j=0}^n[E_j,\widehat{E}_j]).
			\end{cases}
		\end{equation}

		\item \textbf{Local behavior:} For $p\in\{E_j,\widehat{E}_j\}_{j=0}^n$, $\mathcal{M}(z)$ has the following local behavior:
		\begin{equation}
			\begin{aligned}
				\mathcal{M}(z)=\left\{
          \begin{array}{ll}
          (\mathcal{O}((z-p)^{1/4}),\mathcal{O}((z-p)^{-1/4})),& z\to p\mathrm{~from~}\mathbb{C}^+, \v\\
			(\mathcal{O}((z-p)^{-1/4}),\mathcal{O}((z-p)^{1/4})), & z\to p\mathrm{~from~}\mathbb{C}^-.
          \end{array}\right.
			\end{aligned}
		\end{equation}
	\end{enumerate}
\end{RH}

We next reconstruct $q(y,t)$ and $x(y,t)$. For this purpose, we need the asymptotic behavior of $\mathcal{M}(y,t;z)$ as $z\to0$, namely
\begin{equation}
	\begin{aligned}
		&\mathcal{M}
		=
		e^{\eta(x(y,t),t)\sigma_{3}}
		G^{-1}(x(y,t),t)
		 \left\{
		I
		+z\int_{-\infty}^{x}Q(s,t)+\mathrm{i}\sigma_{3}\bigl(p(s,t)-1\bigr)\,ds
				+O(z^2)
		\right\}.
	\end{aligned}
\end{equation}

This yields the reconstruction formulae
\begin{equation}\label{q2}
	\begin{aligned}
		q(x,t)=
		\lim_{z\to0}
		\frac{
			\partial_y\left(\mathcal{M}^{-1}(y,t;0)\mathcal{M}(y,t;z)\right)_{12}
		}{z}\left(1+\mathrm{i}\lim_{\substack{z\to 0}}
		\frac{\partial_y\left(\mathcal{M}^{-1}(y,t;0)\mathcal{M}(y,t;z)\right)_{11}}{z}\right)^{-1},
	\end{aligned}
\end{equation}
and
\begin{equation}\label{x2}
	\begin{aligned}
		x(y,t)=y+\mathrm{i}\lim_{z\to0}
		\frac{
			\left(\mathcal{M}^{-1}(y,t;0)\mathcal{M}(y,t;z)\right)_{11}-1
		}{z}.
	\end{aligned}
\end{equation}

Let $z_{j_1}\in(E_{j_1},\widehat{E}_{j_1})$ be fixed. We set
\begin{equation}\no
	\begin{aligned}
		&\Omega_1=\Omega_1(\xi)=\{z\in\mathbb{C}:0\leq\arg(z-z_1)\leq\varphi_1\},\,\,\Sigma_1=\Sigma_1(\xi)=z_1+e^{i\varphi_1}\mathbb{R}^+,\\
		&\Omega_2=\Omega_2(\xi)=\{z\in\mathbb{C}:\pi-\varphi_1\leq\arg(z-z_{j_1})\leq\pi\},\,\, \Sigma_2=\Sigma_2(\xi)=z_{j_1}+e^{i(\pi-\varphi_1)}\mathbb{R}^+,
	\end{aligned}
\end{equation}
where $\varphi_{1}$ is chosen such that
\begin{equation}\no
	\operatorname{Im}\theta(z)
	\begin{cases}
		<0,\quad z\in\Sigma_1\setminus\{z_1\}, \\
		>0,\quad z\in\Sigma_2\setminus\{z_{j_1}\},
	\end{cases}
\end{equation}
and $U^*$ denotes the complex conjugate of a region $U$.

Noting that the jump matrix $J_{\mathcal{M}}(z)$ admits the following factorizations:
\begin{equation}
	\begin{aligned}
		J_{\mathcal{M}}(z) & =
		\begin{pmatrix}
			1 & 0 \\
			-\frac{\overline{r_1(z)}e^{-2it\theta(z)}}{1-|r_1(z)|^2} & 1
		\end{pmatrix}(1-|r_1(z)|^2)^{\sigma_3}
		\begin{pmatrix}
			1 & \frac{r_1(z)e^{2it\theta(z)}}{1-|r_1(z)|^2} \\
			0 & 1
		\end{pmatrix} \\
		& =
		\begin{pmatrix}
			1 & r_1(z)e^{2it\theta(z)} \\
			0 & 1
		\end{pmatrix}
		\begin{pmatrix}
			1 & 0 \\
			-\overline{r_1(z)}e^{-2it\theta(z)} & 1
		\end{pmatrix},\quad z\in\mathbb{R}\setminus(\cup_{j=0}^n[E_j,\widehat{E}_j]),
	\end{aligned}
\end{equation}
and
\begin{equation}
	\begin{gathered}
		J_{\mathcal{M}}\left(z\right)=
		\begin{pmatrix}
			1 & r_1(z)e^{2it\theta_-(z)} \\
			0 & 1
		\end{pmatrix}
		\begin{pmatrix}
			0 & -ie^{i(t(\theta_+(z)+\theta_-(z))-\phi_j)} \\
			-ie^{-i(t(\theta_+(z)+\theta_-(z))-\phi_j)} & 0
		\end{pmatrix} \\
		\times
		\begin{pmatrix}
			1 & 0 \\
			-r_1(z)e^{-2i(t\theta_+(z)-\phi_j)} & 1
		\end{pmatrix} \\
		\quad\quad\quad~~=
		\begin{pmatrix}
			1 & 0 \\
			-\frac{r_1(z)e^{-2i(t\theta_-(z)-\phi_j)}}{1-(r_1(z))^2e^{2i\phi_j}} & 1
		\end{pmatrix}
		\begin{pmatrix}
			0 & -ie^{i(t(\theta_+(z)+\theta_-(z))-\phi_j)} \\
			-ie^{-i(t(\theta_+(z)+\theta_-(z))-\phi_j)} & 0
		\end{pmatrix} \\
		\times
		\begin{pmatrix}
			1 & \frac{r_1(z)e^{2it\theta_+(z)}}{1-(r_1(z))^2e^{2i\phi_j}} \\
			0 & 1
		\end{pmatrix},\quad z\in\cup_{j=0}^n(E_j,\widehat{E}_j).
	\end{gathered}
\end{equation}

Based on these factorizations, we define the matrix-valued function $\mathcal{M}^{(1)}(z)$ by
\begin{equation}
	\mathcal{M}^{(1)}(z)=e^{\delta(\infty)\sigma_3}\mathcal{M}(z)\mathcal{G}(z)e^{-\delta(z)\sigma_3},
\end{equation}
where
\begin{equation}
	\mathcal{G}(z)=
	\begin{cases}
		\begin{pmatrix}
			1 & 0 \\
			\overline{r_1(z)}e^{-2it\theta(z)} & 1
		\end{pmatrix},\quad z\in\Omega_1, \v\\
		\begin{pmatrix}
			1 & r_1(z)e^{2it\theta(z)} \\
			0 & 1
		\end{pmatrix},\quad z\in\Omega_1^*, \v\\
		\begin{pmatrix}
			1 & -\frac{r_1(z)e^{2it\theta(z)}}{1-|r_1(z)|^2} \\
			0 & 1
		\end{pmatrix},\quad z\in\Omega_2, \v\\
		\begin{pmatrix}
			1 & 0 \\
			-\frac{\overline{r_1(z)}e^{-2it\theta(z)}}{1-|r_1(z)|^2} & 1
		\end{pmatrix},\quad z\in\Omega_2^*, \v\\
		I, \quad \text{elsewhere,}
	\end{cases}
\end{equation}
and the scalar function $\delta(z)$ is defined by
\begin{equation}\label{delta1}
	\delta(z)=\frac{w(z)}{2\pi i}\left[\sum_{j=1}^n\delta_j\int_{E_j}^{\widehat{E}_j}\frac{i\mathrm{~d}s}{w_+(s)(s-z)}-\int_{(-\infty,E_{j_1})\setminus(\cup_{j=0}^{j_1-1}(E_j,\widehat{E}_j))}\frac{\log(1-|r_1(s)|^2)\mathrm{d}s}{w(s)(s-z)}\right],
\end{equation}
where the logarithm is taken on the principal branch, and the constants $\delta_j$, $j=1,\ldots,n$, are determined by the linear system
\begin{equation}\label{delta_1}
	\int_{(-\infty,E_{j_1})\setminus(\cup_{j=0}^{j_1-1}(E_j,\widehat{E}_j))}\frac{\log(1-|r_1(s)|^2)s^k\mathrm{d}s}{w(s)}=\sum_{j=1}^n\delta_j\int_{E_j}^{\widehat{E}_j}\frac{is^k\mathrm{d}s}{w_+(s)},\quad k=0,\ldots,n-1.
\end{equation}
Here it is understood that $(E_{j_1-1},\widehat{E}_{j_1-1})=\emptyset$ if $j_1=0$. In addition, one can prove that each $\delta_j$ is real.

\begin{figure}[!t]
	\centering
	\includegraphics[scale=0.3]{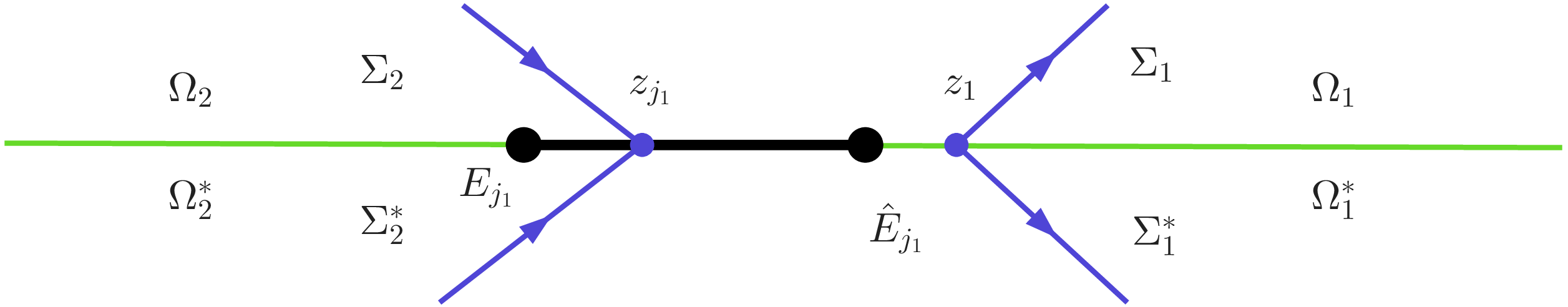}
	\caption{The contours $\Sigma^{(1)}$.}
	\label{fig:WKI9}
\end{figure}

\begin{RH}
	$\delta(z)$ satisfies the following Riemann--Hilbert problem:
	\begin{enumerate}[label=(\textbf{$\delta$\arabic*}), leftmargin=*]
		\item \textbf{Analyticity:} $\delta(z)$ is analytic in $\mathbb{C}\setminus(\cup_{j=0}^n[E_j,\widehat{E}_j]\cup(-\infty,\widehat{E}_{j_1}])$.
		
		\item \textbf{Jump condition:} $\delta(z)$ satisfies the jump relation
		\begin{equation} \label{eq:delta_jump}
			\begin{aligned}
				\delta_-(z)=\left\{\begin{array}{ll}
         \delta_+(z)+\log(1-|r_1(z)|^2),\quad & z\in(-\infty,E_{j_1})\setminus(\cup_{j=0}^{j_1-1}[E_j,\widehat{E}_j]), \\
				-\delta_+(z)+i\delta_j,\quad &z\in(E_j,\widehat{E}_j),\quad j=1,\ldots,n.
			\end{array}\right.
\end{aligned}
		\end{equation}
		
		\item \textbf{Normalization:} 			$\delta(z)=\delta(\infty)+\frac{\delta^{(1)}}{z}+\mathcal{O}(z^{-2}),\quad z\to\infty,$
				where
		\begin{equation}\label{delta_infty_1}
			\begin{aligned}
				\delta(\infty) =\frac{1}{2\pi i}\left[\int_{(-\infty,E_{j_1})\setminus(\cup_{j=0}^{j_1-1}(E_j,\widehat{E}_j))}\frac{\log(1-|r_1(s)|^2)s^n \mathrm{d}s}{w(s)} -\sum_{j=1}^n\delta_j\int_{E_j}^{\widehat{E}_j}\frac{is^n\mathrm{d}s}{w_+(s)}\right],
			\end{aligned}
		\end{equation}
				\begin{equation}
			\begin{aligned}
				\delta^{(1)} =&\delta(\infty)\sum_{j=0}^n(E_j+\widehat{E}_j) \\
				& -\frac{1}{2\pi i}\left[\int_{(-\infty,E_{j_1})\setminus(\cup_{j=1}^{j_1-1}(E_j,\widehat{E}_j))}\frac{\log(1-|r_1(s)|^2)s^{n+1} \mathrm{d}s}{w(s)}-\sum_{j=1}^n\delta_j\int_{E_j}^{\widehat{E}_j}\frac{is^{n+1}\mathrm{d}s}{w_+(s)}\right].
			\end{aligned}
		\end{equation}
		
		\item \textbf{Local behavior:} As $z\to p\in\{E_j\}_{j=0}^{j_1}\cup\{\widehat{E}_j\}_{j=0}^{j_1-1}$ from $\mathbb{C}^+$, we have
		\begin{equation}
			e^{\delta(z)}=\mathcal{O}((z-p)^{1/2}),
		\end{equation}
		and
		\begin{equation}
			\delta(z)=\frac{i}{2}\delta_{j_1}+\mathcal{O}((z-\widehat{E}_{j_1})^{1/2}),\quad z\to\widehat{E}_{j_1} \text{~from~} \mathbb{C}\setminus(-\infty,\widehat{E}_{j_1}).
		\end{equation}
		
		\item \textbf{Symmetry relation:} $\delta(z)$ satisfies the symmetry relation $\delta(z)=-\overline{\delta(\bar{z})}$.
	\end{enumerate}
\end{RH}

\begin{RH}
	$\mathcal{M}^{(1)}(z)$ satisfies the following Riemann-Hilbert problem:
	\begin{enumerate}[label=(\textbf{$\mathcal{M}^{(1)}$\arabic*}), leftmargin=*]
		\item \textbf{Analyticity:} $\mathcal{M}^{(1)}(z)$ is analytic in $\mathbb{C}\setminus\Sigma^{(1)}$, where	
			$\Sigma^{(1)}=(-\infty,z_1]\cup(\cup_{j=0}^n[E_j,\widehat{E}_j])\cup\Sigma_1\cup\Sigma_1^*\cup\Sigma_2\cup\Sigma_2^*$ (see Fig.~\ref{fig:WKI9}).
		
		\item \textbf{Jump condition:} $\mathcal{M}^{(1)}(z)$ satisfies the jump relation $\mathcal{M}^{(1)}_+(z)=\mathcal{M}^{(1)}_-(z)J_{\mathcal{M}}^{(1)}(z)$ for $z\in\Sigma^{(1)}$, where the jump matrix is given by
		\begin{equation} \label{eq:M1_jump}
			J_{\mathcal{M}}^{(1)}(z)=\begin{cases}
				-ie^{i(f_0y+g_0t-(B_j^fy+B_j^gt+\phi_j-\delta_j)/2)\widehat{\sigma}_3}\sigma_1,&z\in(E_j,\widehat{E}_j),\\[2.5ex]
				\begin{pmatrix}1&0\\-\overline{r_1(z)}e^{-2it\theta(z)-2\delta(z)}&1\end{pmatrix},&z\in\Sigma_1,\\[2.5ex]
				\begin{pmatrix}1&r_1(z)e^{2it\theta(z)+2\delta(z)}\\0&1\end{pmatrix},&z\in\Sigma_1^*,\\[2.5ex]
				\begin{pmatrix}1&\dfrac{r_1(z)e^{2it\theta(z)+2\delta(z)}}{1-|r_1(z)|^2}\\0&1\end{pmatrix},&z\in\Sigma_2,\\[2.5ex]
				\begin{pmatrix}1&0\\ \dfrac{-\overline{r_1(z)}e^{-2it\theta(z)-2\delta(z)}}{1-|r_1(z)|^2}&1\end{pmatrix},&z\in\Sigma_2^*,\\[2.5ex]
				\begin{pmatrix}1-|r_1(z)|^2&r_1(z)e^{2it\theta(z)+2\delta(z)}\\ -\overline{r_1(z)}e^{-2it\theta(z)-2\delta(z)}&1\end{pmatrix},&z\in(\widehat{E}_{j_1},z_1),
			\end{cases}
		\end{equation}
		with $\phi_0=\delta_0=0$ and $j=0,\dots,n$.
		
		\item \textbf{Normalization at infinity:} As $z\to\infty$, we have $\mathcal{M}^{(1)}(z)=I+\mathcal{O}(z^{-1})$.
		
		\item \textbf{Local behavior:} For $p\in\{E_j,\widehat{E}_j\}_{j=0}^n\setminus\{\widehat{E}_{j_1}\}$, one has
		$
			\mathcal{M}^{(1)}(z)=\mathcal{O}((z-p)^{-1/4}),\quad z\to p.$
	\end{enumerate}
\end{RH}

\subsection{Global parametrix}

We now construct the global parametrix, which captures the leading asymptotic behavior of the Riemann--Hilbert problem. We seek a global parametrix $\mathcal{M}^{(glo)}$ solving the model Riemann--Hilbert problem obtained by neglecting the exponentially small jumps on the lens contours. More precisely, we consider the following problem.

\begin{RH}
	$\mathcal{M}^{(glo)}(z)$ satisfies the following Riemann--Hilbert problem:
	\begin{enumerate}[label=(\textbf{$\mathcal{M}^{(glo)}$\arabic*}), leftmargin=*]
		\item \textbf{Analyticity:} $\mathcal{M}^{(glo)}(z)$ is analytic in $\mathbb{C}\setminus(\cup_{j=0}^n[E_j,\widehat{E}_j])$.
		
		\item \textbf{Jump condition:} $\mathcal{M}^{(glo)}(z)$ satisfies the jump relation $\mathcal{M}^{(glo)}_+(z)=\mathcal{M}^{(glo)}_-(z)J_{\mathcal{M}}^{(glo)}(z)$ for $z\in\cup_{j=0}^n(E_j,\widehat{E}_j)$, where the jump matrix is given by
		\begin{equation} \label{eq:M^{(glo)}_jump}
			J_{\mathcal{M}}^{(glo)}(z)=
			-ie^{i[f_0y+g_0t-(B_j^fy+B_j^gt+\phi_j-\delta_j)/2]\widehat{\sigma}_3}\sigma_1,\quad z\in(E_j,\widehat{E}_j).
		\end{equation}
		
		\item \textbf{Normalization at infinity:} As $z\to\infty$, we have $\mathcal{M}^{(glo)}(z)=I+\mathcal{O}(z^{-1})$.
		
		\item \textbf{Local behavior:} For $p\in\{E_j,\widehat{E}_j\}_{j=0}^n$, one has
		$
			\mathcal{M}^{(glo)}(z)=\mathcal{O}((z-p)^{-1/4}),\quad z\to p.
		$
	\end{enumerate}
\end{RH}

As in the Riemann--Hilbert problem for $M$, the above problem can be solved explicitly using the planar matrix Baker--Akhiezer function. More precisely, we have
\begin{equation}\label{M_glo}
	\mathcal{M}^{(glo)}(z)=e^{-i(f_0y+g_0t)\sigma_3}D^{(glo)}(\infty)^{-1}D^{(glo)}(z)e^{i(f_0y+g_0t)\sigma_3},
\end{equation}
where $D^{(glo)}(z)=\left(D_{ij}^{(glo)}(z)\right)_{i,j=1,2}$ with
\begin{equation} \no
	D_{j1}^{(glo)}(z)=\frac{1}{2}(\kappa(z)+(-1)^{j-1}\kappa(z)^{-1})\frac{\theta(\mathbf{A}(z)+\mathbf{c^*}+\mathbf{d}_j)}
	{\theta(\mathbf{A}(z)+\mathbf{d}_j)},\quad j=1,2,
\end{equation}
\begin{equation} \no
	D_{j2}^{(glo)}(z)=\frac{1}{2}(\kappa(z)+(-1)^j\kappa(z)^{-1})\frac{\theta(-\mathbf{A}(z)+\mathbf{c^*}+\mathbf{d}_j)}
	{\theta(-\mathbf{A}(z)+\mathbf{d}_j)},\quad j=1,2,
\end{equation}
where $\mathbf{c^*}=(c_0^*,c_1^*,\ldots,c_j^*,\ldots,c_n^*)$, $j=0,1,\ldots,n$, and
\begin{equation}\no
	c_j^*=-\frac{yC_j^f+tC_j^g+\phi_j-\delta_j}{2\pi},
	\qquad j=1,\ldots,n.
\end{equation}

Consequently, we obtain
\begin{equation}\label{M_glo_12}
	\begin{aligned}
	q^{(AG)}(x,t)=q^{(AG)}(x,t;\boldsymbol{E},\boldsymbol{\widehat{E}},\boldsymbol{\phi}-\boldsymbol{\delta})=\lim_{z\to0}
	\frac{
		\partial_y\left({\mathcal{M}^{(glo)}}^{-1}(y,t;0){\mathcal{M}^{(glo)}}(y,t;z)\right)_{12}
	}{z}\\
	\times\left(1+\mathrm{i}\lim_{\substack{z\to 0}}
	\frac{\partial_y\left({\mathcal{M}^{(glo)}}^{-1}(y,t;0){\mathcal{M}^{(glo)}}(y,t;z)\right)_{11}}{z}\right)^{-1}
	\end{aligned}
\end{equation}
with \begin{equation}
	\begin{aligned}
		x(y,t)
		=y+\mathrm{i}\lim_{z\to0}
		\frac{\left({\mathcal{M}^{(glo)}}^{-1}(y,t;0){\mathcal{M}^{(glo)}}(y,t;z)\right)_{11}-1
		}{z}.
	\end{aligned}
\end{equation}

\subsection{Local parametrix}
The global parametrix $\mathcal{M}^{(glo)}(z)$ constructed in the previous subsection gives an accurate approximation to $\mathcal{M}^{(1)}(z)$ away from the endpoint $\widehat{E}_{j_1}$. At the endpoint $\widehat{E}_{j_1}$, however, this approximation fails because the jump matrices for $\mathcal{M}^{(1)}(z)$ do not converge uniformly to the identity and may exhibit singular behavior. To resolve this issue, we construct a local parametrix in a small neighborhood of this point.

We define the small open disk $U_{1}=\{z\in\mathbb{C}: |z-\widehat{E}_{j_1}|<\varepsilon_1\}$ centered at $\widehat{E}_{j_1}$, with fixed radius $\varepsilon_1>0$ chosen sufficiently small. In particular, we may choose
\begin{equation}
	\varepsilon_1=\min\left\{\frac{1}{2}(\widehat{E}_{j_1}-z_{j_1}),\frac{1}{2}(E_{j_1+1}-\widehat{E}_{j_1}),2(z_1-\widehat{E}_{j_1})t^\epsilon\right\},\quad\epsilon\in(1/3,2/3).
\end{equation}

\begin{RH}
	$\mathcal{M}^{(loc)}(z)$ satisfies the following Riemann--Hilbert problem:
	\begin{enumerate}[label=(\textbf{$\mathcal{M}^{(loc)}$\arabic*}), leftmargin=*]
		\item \textbf{Analyticity:} $\mathcal{M}^{(loc)}(z)$ is analytic in $U_1\setminus\Sigma^{(1)}$.
		
		\item \textbf{Jump condition:} $\mathcal{M}^{(loc)}(z)$ satisfies the jump relation $\mathcal{M}^{(loc)}_+(z)=\mathcal{M}^{(loc)}_-(z)J_{\mathcal{M}}^{(1)}(z)$ for $z\in\Sigma^{(1)}\cap U_1$.
		
		\item \textbf{Matching condition:} As $t\to\infty$, $\mathcal{M}^{(loc)}(z)$ matches $\mathcal{M}^{(glo)}(z)$ on the boundary $\partial U_1$ of $U_1$.
		
		\item \textbf{Local behavior:} At $\widehat{E}_{j_1}$, one has
		$\mathcal{M}^{(loc)}(z)=\mathcal{O}((z-\widehat{E}_{j_1})^{-1/4}),\quad z\to \widehat{E}_{j_1}.$
		
	\end{enumerate}
\end{RH}

To construct the solution, we need a suitable conformal map that captures the local behavior of the phase function near the endpoint. For $\theta$, we have the following local expansion:
\begin{equation}
	\theta(z,\xi)=\theta^{(0,\widehat{j}_1)}+(\xi-\widehat{\xi}_{j_1})\theta^{(1,\widehat{j}_1)}(z-\widehat{E}_{j_1})^{1/2}+\frac{2}{3}\theta^{(3,\widehat{j}_1)}(z-\widehat{E}_{j_1})^{3/2}+O\left((z-\widehat{E}_{j_1})^{5/2}\right),
\end{equation}
as $z\to\widehat{E}_{j_1}$, for fixed $\xi$, where
\begin{equation}\no
	\theta^{(0,\widehat{j}_1)}=\theta(\widehat{E}_{j_1};\xi)=(f_0-B_{j_1}^f/2)\xi+g_0-\frac{1}{2}B_{j_1}^g,
\quad
	\theta^{(1,\widehat{j}_1)}=-\frac{2\prod_{j=0}^n(\widehat{E}_{j_1}-z_j^f)}{\widehat{w}^{(j_1)}(\widehat{E}_{j_1})},
\end{equation}
\begin{equation}\no
	\theta^{(3,\widehat{j}_1)}=\frac{1}{\widehat{w}^{(j_1)}(\widehat{E}_{j_1})^2}
\left[T(\widehat{E}_{j_1};\xi)(\widehat{w}^{(j_1)})^{\prime}(\widehat{E}_{j_1})-\partial_zT(\widehat{E}_{j_1};\xi)
\widehat{w}^{(j_1)}(\widehat{E}_{j_1})\right],
\end{equation}
with
\begin{equation} \no
	w^{(i)}(z)=\left[-\prod_{k=0}^n(z-\widehat{E}_k)\cdot\prod_{k=0,...,n,k\neq i}(z-E_k)\right]^{\frac{1}{2}},
\end{equation}
and
\begin{equation}\no
	\widehat{w}^{(i)}(z)=\left[\prod_{k=0}^n(z-E_k)\cdot\prod_{k=0,...,n,k\neq i}(z-\widehat{E}_k)\right]^{\frac{1}{2}}.
\end{equation}
This local expansion motivates the definition
\begin{equation}
	\zeta(z)=\zeta(z;\xi)=\left(\frac{3it}{2}(\theta(\widehat{E}_{j_1})+(\xi-\widehat{\xi}_{j_1})\theta^{(1,\widehat{j}_1)}(z-\widehat{E}_{j_1})^{1/2}-\theta(z))\right)^{2/3},\quad z\in U_1,
\end{equation}
which is a one-to-one conformal map in $U_1$ with respect to $z$. Moreover, it is readily seen that
\begin{equation}
	\zeta(\widehat{E}_{j_1})=0,\quad\zeta^{\prime}(\widehat{E}_{j_1})=-|\theta^{(3,\widehat{j}_1)}(\xi)|^{2/3}t^{2/3}<0.
\end{equation}

\subsection{Construction of a model Riemann--Hilbert problem}

In this subsection, we explicitly construct the matrix function $\mathcal{M}^{(loc)}(z)$ using Painlev\'{e} XXXIV functions. The construction is carried out in three steps: First of all, we solve a standard model Riemann--Hilbert problem in the auxiliary $\zeta$-plane; Next, we establish the error estimate; Finally, we construct $\mathcal{M}^{(loc)}(z)$.
We use the standard Painlev\'{e} XXXIV parametrix $\widehat\Psi(\zeta)$, as described in \cite{Its2008,Its2009,Fan2026}.
\begin{RH}
	The model function $\widehat\Psi(\zeta)$ is characterized by the following Riemann--Hilbert problem:
	\begin{enumerate}[label=(\textbf{$\widehat\Psi$\arabic*}), leftmargin=*]
		\item \textbf{Analyticity:} $\widehat\Psi(\zeta)$ is analytic in $\mathbb{C}\setminus \Sigma^{(loc,\zeta)}$ and continuous up to the boundary. The jump contour $\Sigma^{(loc,\zeta)}$ is defined by
		\begin{equation}
			\begin{aligned}
				\Sigma_1^{(loc)} &=(\zeta(\Sigma_1^*\cap U_1))\cup\{\zeta(\Sigma_1^*\cap\partial U_1)+e^{(\pi-\varphi)i}\mathbb{R}^+\}, \\
				\Sigma^{(loc,\zeta)} &=(\zeta(z_1),+\infty)\cup\Sigma_1^{(loc)}\cup\Sigma_1^{(loc)*}.
			\end{aligned}
		\end{equation}
		
		\item \textbf{Jump condition:} On the contour $\Sigma^{(loc,\zeta)}$, $\widehat\Psi(\zeta)$ satisfies the jump relation $\widehat\Psi_+(\zeta)=\widehat\Psi_-(\zeta)J^{(loc,1)}(\zeta)$, where the jump matrices are given by
		\begin{equation}
			J^{(loc,1)}(\zeta)=e^{it\theta(\widehat{E}_{j_1})\widehat{\sigma}_3}
			\begin{cases}
				e^{-i(\phi_{j_{1}}-\delta_{j_{1}})\widehat{\sigma}_{3}/2}
				\begin{pmatrix}
					0 & i \\
					i & 0
				\end{pmatrix}, & \zeta\in\mathbb{R}^{+}, \\
				\begin{pmatrix}
					1 & e^{-i\phi_{j_{1}}+2\widehat{\theta}(\zeta)+i\delta_{j_{1}}} \\
					0 & 1
				\end{pmatrix}, & \zeta\in\Sigma_{1}^{(loc)}, \\
				\begin{pmatrix}
					1 & 0 \\
					e^{i\phi_{j_{1}}+2\widehat{\theta}(\zeta)-i\delta_{j_{1}}} & 1
				\end{pmatrix}, & \zeta\in\Sigma_{1}^{(loc)*}, \\
				\begin{pmatrix}
					1 & 0 \\
					e^{i\phi_{j_{1}}+2\widehat{\theta}(\zeta)-i\delta_{j_{1}}} & 1
				\end{pmatrix}
				\begin{pmatrix}
					1 & e^{-i\phi_{j_{1}}+2\widehat{\theta}(\zeta)+i\delta_{j_{1}}} \\
					0 & 1
				\end{pmatrix}, & \zeta\in(\zeta(z_{1}),0),
			\end{cases}
		\end{equation}
		where
		\begin{equation}\no
			\widehat{\theta}(\zeta)=s\zeta^{1/2}+\frac{2}{3}\zeta^{3/2},\quad
			s=-\frac{\theta^{(1,\widehat{j}_1)}}{|\theta^{(3,\widehat{j}_1)}(\xi)|^{1/3}}(\xi-\widehat{\xi}_{j_1})t^{2/3}\in\mathbb{R}.
		\end{equation}
		
		\item \textbf{Asymptotics at infinity:} As $\zeta\to\infty$, we have
		\begin{equation} \no
			\widehat\Psi(\zeta)=\left(I+\mathcal{O}\left(\zeta^{-1}\right)\right)\frac{\zeta^{-\frac{1}{4}\sigma_3}}{\sqrt{2}}
			\begin{pmatrix}
				1 & i \\
				i & 1
			\end{pmatrix}\mathcal{G}_1(\zeta),
		\end{equation}
		where
		\begin{equation}\no
			\mathcal{G}_1(\zeta)=
			\begin{cases}
				 i\sigma_2e^{\pi i\sigma_3/4}e^{-i(t\theta(\widehat{E}_{j_1})-\phi_{j_1}+\delta_{j_1})\sigma_3/2}, & \zeta\in\mathbb{C}^+, \\
				e^{\pi i\sigma_3/4}e^{-i(t\theta(\widehat{E}_{j_1})-\phi_{j_1}+\delta_{j_1})\sigma_3/2}, & \zeta\in\mathbb{C}^-.
			\end{cases}
		\end{equation}
		
		\item \textbf{Local behavior:} As $\zeta\to0$, $\widehat\Psi(\zeta)$ has the local behavior $\widehat\Psi(\zeta)=\mathcal{O}(\zeta^{-1/4})$.
	\end{enumerate}
\end{RH}

This model RH problem admits an explicit solution in terms of the Painlev\'{e} XXXIV function $M^{(P_{34})}(\zeta;s,\alpha,\omega)$. To this end, let $\alpha=-1/4$ and $\omega=0$, and define
\begin{equation}
	\widehat\Psi(\zeta)=
	\begin{pmatrix}
		1 & 0 \\
		ia(s) & 1
	\end{pmatrix}M^{(P_{34})}(\zeta;s,-1/4,0)e^{\widehat{\theta}(\zeta)\sigma_3}\mathcal{G}_1(\zeta)\mathcal{G}_2(\zeta),
\end{equation}
where
\begin{equation}
	\begin{aligned}
		\mathcal{G}_2(\zeta)=e^{it\theta(\widehat{E}_{j_1})\widehat{\sigma}_3}
		\begin{cases}
			\begin{pmatrix}
				1 & e^{-i\phi_{j_1}+2\widehat{\theta}(\zeta)+i\delta_{j_1}} \\
				0 & 1
			\end{pmatrix}, & \zeta\in\Omega_1^{(loc)}, \\
			\begin{pmatrix}
				1 & 0 \\
				e^{i\phi_{j_1}+2\widehat{\theta}(\zeta)-i\delta_{j_1}} & 1
			\end{pmatrix}, & \zeta\in\Omega_1^{(loc)*}, \\
			I, & \text{elsewhere,}
		\end{cases}
	\end{aligned}
\end{equation}
with
\begin{equation}\label{a}
	a(s)=\int_{-\infty}^s\left(w(\zeta)+\frac{\zeta}{2}\right)\mathrm{d}\zeta.
\end{equation}
Here $w(x)$ is the unique solution of the Painlev\'{e} XXXIV equation \cite{Its2008,Its2009}
\begin{equation}
	w^{\prime\prime}(x)=4w^2(x)+2xw(x)+\frac{w'^2(x)-1/4}{2w(x)}.
\end{equation}

Note that, as $t\to\infty$, $J^{(1)}(\zeta)$ is well approximated by $J^{(loc,1)}(\zeta)$ for $\zeta\in\Sigma_1^{(loc)}\cap\zeta(U_1)$. It is therefore natural to expect that $\mathcal{M}^{(loc)}$ is well approximated by $\widehat\Psi$ for large $t$. We begin by introducing a Riemann--Hilbert problem whose jump matrix is the same as that of $\mathcal{M}^{(loc)}$ and whose large-$\zeta$ asymptotics agree with those of $\widehat\Psi$.

\begin{RH}
	${\mathcal{M}}^{(loc,0)}(\zeta)$ satisfies the following Riemann--Hilbert problem:
	\begin{enumerate}[label=(\textbf{${\mathcal{M}}^{(loc,0)}$\arabic*}), leftmargin=*]
		\item \textbf{Analyticity:} ${\mathcal{M}}^{(loc,0)}(\zeta)$ is analytic in $\mathbb{C}\setminus \Sigma^{(loc,\zeta)}$.
		
		\item \textbf{Jump condition:} On the contour $\Sigma^{(loc,\zeta)}$, ${\mathcal{M}}^{(loc,0)}(\zeta)$ satisfies the jump relation ${\mathcal{M}}^{(loc,0)}_+(\zeta)={\mathcal{M}}^{(loc,0)}_-(\zeta)J_{\mathcal{M}}^{(1)}(\zeta)$.
		
		\item \textbf{Asymptotics at infinity:} As $\zeta\to\infty$, we have
		\begin{equation}
			{\mathcal{M}}^{(loc,0)}(\zeta)=\left(I+\mathcal{O}\left(\zeta^{-1}\right)\right)\frac{\zeta^{-\frac{1}{4}\sigma_3}}{\sqrt{2}}
			\begin{pmatrix}
				1 & i \\
				i & 1
			\end{pmatrix}\mathcal{G}_1(\zeta).
		\end{equation}
		
		\item \textbf{Local behavior:} As $\zeta\to0$, ${\mathcal{M}}^{(loc,0)}(\zeta)$ has the local behavior ${\mathcal{M}}^{(loc,0)}(\zeta)=\mathcal{O}(\zeta^{-1/4})$.
	\end{enumerate}
\end{RH}

Set $\Xi(\zeta)=\mathcal{M}^{(loc,0)}(\zeta)\widehat\Psi(\zeta)^{-1}$. Following \cite{Fan2026}, one obtains that, as $t\to\infty$, $\Xi(\zeta)$ exists uniquely and satisfies
\begin{equation}
	\Xi(\zeta)=I+\mathcal{O}(t^{-1/3}).
\end{equation}
Hence the asymptotic expansion of $\mathcal{M}^{(loc,0)}(\zeta)$ is
\begin{equation}
	\mathcal{M}^{(loc,0)}(\zeta)=\left(I+\frac{\mathcal{M}_1^{(loc,0)}}{\zeta}+O\left(\zeta^{-2}\right)\right)\frac{\zeta^{-\frac{1}{4}\sigma_3}}{\sqrt{2}}
	\begin{pmatrix}
		1 & i \\
		i & 1
	\end{pmatrix}\mathcal{G}_1(\zeta).
\end{equation}

We now construct $\mathcal{M}^{(loc)}(z)$. With the aid of $\mathcal{M}^{(loc,0)}(\zeta)$, we define
\begin{equation}
	\mathcal{M}^{(loc)}(z)=H_1(z)t^{\sigma_3/6}\mathcal{M}^{(loc,0)}(\zeta(z)),\quad z\in U_1,
\end{equation}
where
\begin{equation}
	H_1(z)=\mathcal{M}^{(glo)}(z)\mathcal{G}_1(\zeta(z))^{-1}\frac{1}{\sqrt{2}}\begin{pmatrix}
		1 & -i \\
		-i & 1
	\end{pmatrix}\left((\widehat{E}_{j_1}-z)|\theta^{(3,\widehat{j}_1)}(\widehat{\xi}_{j_1})|^{2/3}\right)^{\sigma_3/4}.
\end{equation}
The function $H_1(z)$ is analytic in $U_1$; see \cite{Fan2026}. Finally, using $\mathcal{M}^{(glo)}(z)$ and $\mathcal{M}^{(loc)}(z)$, we obtain
\begin{equation}
	\mathcal{M}^{(loc)}(z)\mathcal{M}^{(glo)}(z)^{-1}=I-\frac{t^{1/3}}{\zeta(z)}H_1(z)
	\begin{pmatrix}
		0 & ia(s) \\
		0 & 0
	\end{pmatrix}H_1(z)^{-1}+\mathcal{O}(t^{1/3-2\epsilon}),
\end{equation}
as $t\to\infty$, for $z\in\partial U_1$.

\subsection{The small-norm Riemann--Hilbert problem}
The final transformation involves the global parametrix $\mathcal{M}^{(glo)}(z)$ and the local parametrix $\mathcal{M}^{(loc)}(z)$. We define the matrix-valued function $E(z)$ by
\begin{equation}
	E(z)=
	\begin{cases}
		\mathcal{M}^{(1)}(z)\mathcal{M}^{(glo)}(z)^{-1},\quad z\in\mathbb{C}\setminus U_1, \\
		\mathcal{M}^{(1)}(z)\mathcal{M}^{(loc)}(z)^{-1},\quad z\in U_1.
	\end{cases}
\end{equation}

\begin{figure}[!t]
	\centering
	\includegraphics[scale=0.23]{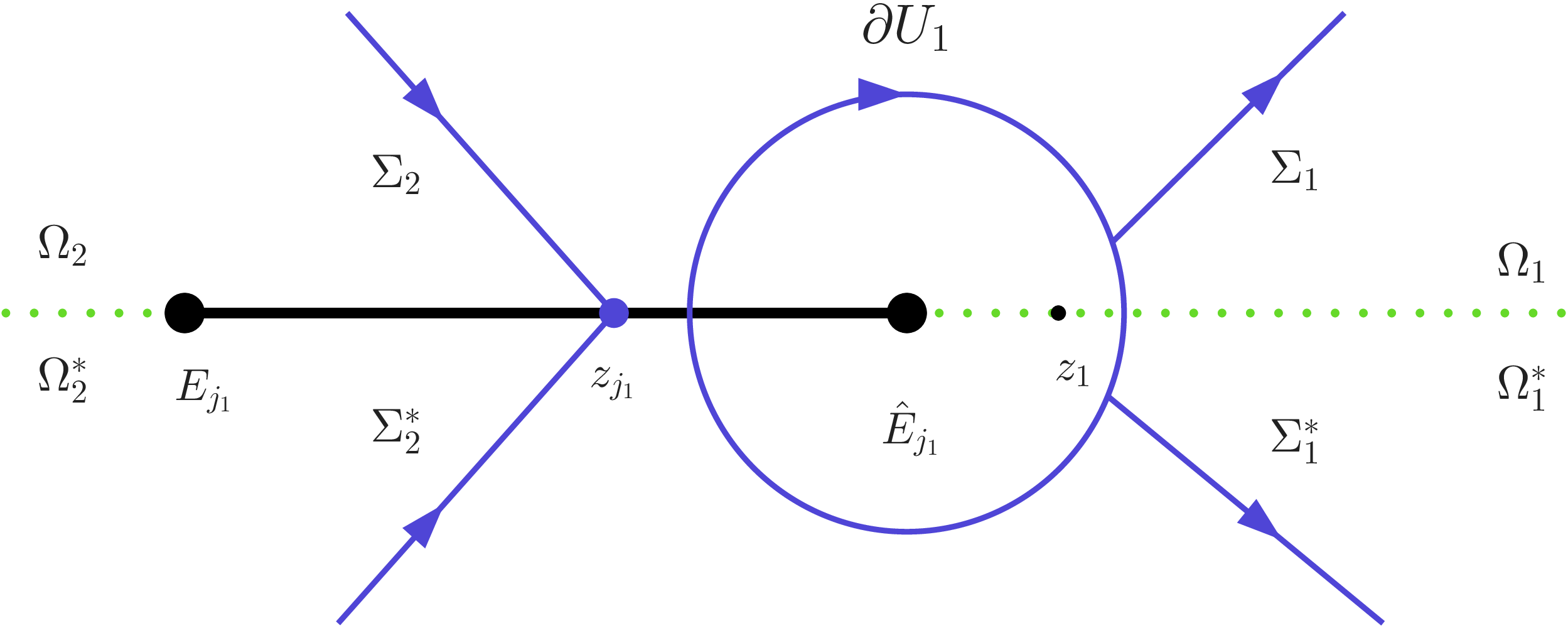}
	\caption{The contours $\Sigma^{(E)}$.}
	\label{fig:WKI10}
\end{figure}

\begin{RH}
	$E(z)$ satisfies the following Riemann--Hilbert problem:

	\begin{enumerate}[label=(\textbf{$E$\arabic*}), leftmargin=*]

		\item \textbf{Analyticity:} $E(z)$ is analytic in $\mathbb{C}\setminus \Sigma^{(E)}$, where the jump contour $\Sigma^{(E)}$ is defined by (see Fig.~\ref{fig:WKI10})
		$\Sigma^{(E)}=\partial U_1\cup\Sigma_2\cup\Sigma_2^*\cup\Sigma_1\cup\Sigma_1^*\setminus U_1.$

		\item \textbf{Jump condition:} On the contour $\Sigma^{(E)}$, $E(z)$ satisfies the jump relation $E_+(z)=E_-(z)J_{E}(z)$, where the jump matrix is given by
		\begin{equation}
			J_{E}(z)=\left\{
			\begin{array}{ll}\mathcal{M}^{(glo)}(z)J_{\mathcal{M}}^{(1)}(z)\mathcal{M}^{(glo)}(z)^{-1}, &
         z\in\Sigma^{(E)}\setminus\partial U_1, \v\\
				\mathcal{M}^{(loc)}(z)\mathcal{M}^{(glo)}(z)^{-1}, & z\in\partial U_1.
			\end{array}\right.
		\end{equation}
		
		\item \textbf{Asymptotics at infinity:} As $z\to\infty$, we have $E(z)=I+\mathcal{O}(z^{-1})$.
		
		\item \textbf{Local behavior:} As $z\to\widehat{E}_{j_1}$, $E(z)$ has the local behavior $E(z)=\mathcal{O}((z-\widehat{E}_{j_1})^{-1/2})$.
	\end{enumerate}
\end{RH}

A direct calculation shows that
\begin{equation}\label{J_E-I}
	\parallel J_E(z)-I\parallel_{L^p}=
	\begin{cases}
		\mathcal{O}(\exp\left\{-ct^{\epsilon/2}\right\}), & z\in\Sigma^{(E)}\setminus\partial U_1, \\
		\mathcal{O}(t^{-\kappa_p}), & z\in\partial U_1 ,
	\end{cases}
\end{equation}
for some positive constant $c$, with $\kappa_\infty=\epsilon-1/3$ and $\kappa_2=\epsilon/2$; see \cite{Fan2026}. Since the jump matrix is uniformly close to the identity, the standard small-norm Riemann--Hilbert theory (see, e.g., \cite{Deift1993}) implies that the Riemann--Hilbert problem for $E$ has a unique solution for sufficiently large positive $t$. Moreover,
\begin{equation}
	E(z)=I+\frac{1}{2\pi i}\int_{\Sigma^{(E)}}\frac{\varpi(\lambda)(J_E(\lambda)-I)}{\lambda-z}\mathrm{d}\lambda,
\end{equation}
where $\varpi\in I+L^2(\Sigma^{(E)})$ is the unique solution of the Fredholm-type equation
\begin{equation}
	\varpi=I+\mathcal{C}_E\varpi,
\end{equation}
and $\mathcal{C}_E{:}L^2(\Sigma^{(E)})\to L^2(\Sigma^{(E)})$ is the integral operator defined by $\mathcal{C}_E(f)(z)=\mathcal{C}_-\left(f(J_E(z)-I)\right)$, with $\mathcal{C}_-$ being the Cauchy projection operator on $\Sigma^{\large(E)}$ defined by
\begin{equation}
	\mathcal{C}_{\pm}f(z)=\lim_{z^{\prime}\to z\in\Sigma^{\large(E)}}\frac{1}{2\pi i}\int_{\Sigma^{\large(E)}}\frac{f(\zeta)}{\zeta-z^{\prime}}\mathrm{d}\zeta.
\end{equation}

Furthermore, to reconstruct $q(y,t)$, it is necessary to study the asymptotic behavior of $E(z)$ as $z\to0$ and the large-time asymptotic behavior of $E(0)$. In view of the estimate \eqref{J_E-I}, as $t\to\infty$, it suffices to compute the contribution from $\partial U_1$, since the contributions from the remaining parts of the contour are exponentially small. First, as $z\to0$, we have
\begin{equation}
	E(z)=E(0)+E_1z+O(z^2),
\end{equation}
where
\begin{equation}\no
	E(0)
	=
	I
	+
	\frac{1}{2\pi\mathrm{i}}
	\int_{\Sigma^{(E)}}
	\frac{(I+\varpi(\lambda))(J_E(\lambda)-I)}{\lambda}\,d\lambda,
\end{equation}
and
\begin{equation}\no
	E_1
	=
	-\frac{1}{2\pi\mathrm{i}}
	\int_{\Sigma^{(E)}}
	\frac{(I+\varpi(\lambda))(J_E(\lambda)-I)}{\lambda^2}\,d\lambda.
\end{equation}
Then, as $t\to\infty$, the asymptotic behavior of $E(0)$ and $E_1$ is given by
\begin{equation}\label{E_0}
	\begin{aligned}
		E(0)
		&=
		I
		+\frac{1}{2\mathrm{i}\pi}
		\int_{\partial{U}_{1}}
		\left(J_E(\lambda)-I\right)\,d\lambda
		+\mathcal{O}(t^{-\epsilon})
		\\
		&=
		I
		+\frac{t^{-1/3}}{|\theta^{(3,\widehat{j}_1)}(\widehat{\xi}_{j_1})|^{2/3}}H_1(\widehat{E}_{j_1})
		\begin{pmatrix}
			0 & ia \\
			0 & 0
		\end{pmatrix}H_1(\widehat{E}_{j_1})^{-1}+\mathcal{O}(t^{-\epsilon}),
	\end{aligned}
\end{equation}
and
\begin{equation}\no
	E_1
	=
	-\frac{t^{-1/3}}{|\theta^{(3,\widehat{j}_1)}(\widehat{\xi}_{j_1})|^{2/3}\widehat{E}_{j_1}}H_1(\widehat{E}_{j_1})
	\begin{pmatrix}
		0 & ia \\
		0 & 0
	\end{pmatrix}H_1(\widehat{E}_{j_1})^{-1}+\mathcal{O}(t^{-\epsilon}).
\end{equation}
Moreover, from \eqref{E_0}, a direct calculation shows that
\begin{equation}\label{E_0-1}
	E(0)^{-1}=I+\mathcal{O}(t^{-1/3}).
\end{equation}

\section{Asymptotic analysis in transition region II}\label{Region-II}

In this and the next sections, we use the same notation as before, with the meaning understood from the corresponding context.

In this section, we carry out the asymptotic analysis of the Riemann--Hilbert problem for $\mathcal{N}(z)$ in order to derive the long-time asymptotics of $q$ in transition region II. According to Definition \ref{def1}, we may restrict ourselves to
$0\leq(\xi-\xi_{j_1})t^{2/3}\leq C$
for some fixed $j_1\in\{0,\ldots,n\}$, since the analysis in the complementary half-region is analogous. For large $t$, one has $\xi_{j_1}<\xi<\widehat{\xi}_{j_1-1}$, and hence the saddle point satisfies $z_1\in[\widehat{E}_{j_1-1},E_{j_1})$. Here $\mathcal{N}(z)$ is defined by
\begin{equation}
	\mathcal{N}(z)=\mathcal{N}(y,t;z)  =
	\begin{cases}
		\left(\mu_1^-(y,t;z),\dfrac{\mu_2^+(y,t;z)}{S_{22}(z)}\right),\quad & z\in\mathbb{C}^+, \\[1.2em]
		\left(\dfrac{\mu_1^+(y,t;z)}{S_{11}(z)},\mu_2^-(y,t;z)\right),\quad & z\in\mathbb{C}^-.
	\end{cases}
\end{equation}

\begin{RH}
	$\mathcal{N}(z)$ satisfies the following Riemann--Hilbert problem:
	\begin{enumerate}[label=(\textbf{$\mathcal{N}$\arabic*}), leftmargin=*]
		\item \textbf{Analyticity:} $\mathcal{N}(z)$ is analytic in $\mathbb{C}\setminus\mathbb{R}$.
		
		\item \textbf{Jump condition:} $\mathcal{N}_+(z)=\mathcal{N}_-(z)J_{\mathcal{N}}(z)$ for $z\in\mathbb{R}$, where the jump matrix is given by
		\begin{equation} \label{eq:N_jump}
			J_{\mathcal{N}}(z)=
			\begin{cases}
				e^{i(f_0y+g_0t-(B_j^fy+B_j^gt+\phi_j)/2)\widehat{\sigma}_3}
				\begin{pmatrix}
					0 & -i \\
					-i & 0
				\end{pmatrix},&z\in(E_j,\widehat{E}_j),j=0,\ldots,n, \\
				\\
				\begin{pmatrix}
					1 & \overline{r_2(z)}e^{2it\theta(z)} \\
					-r_2(z)e^{-2it\theta(z)} & 1-|r_2(z)|^2
				\end{pmatrix},& z\in\mathbb{R}\setminus(\cup_{j=0}^n[E_j,\widehat{E}_j]).
			\end{cases}
		\end{equation}
		
		\item \textbf{Normalization at infinity:} As $z\to\infty$, $\mathcal{N}(z)=I+\mathcal{O}(z^{-1})$.
		
		\item \textbf{Local behavior:} For $p\in\{E_j,\widehat{E}_j\}_{j=0}^n$, $\mathcal{N}(z)$ has the following local behaviors:
		\begin{equation}
						\mathcal{N}(z)=\left\{
      \begin{array}{ll}
       (\mathcal{O}((z-p)^{-1/4}),\mathcal{O}((z-p)^{1/4})),& z\to p\mathrm{~from~}\mathbb{C}^+,\v\\
	 (\mathcal{O}((z-p)^{1/4}),\mathcal{O}((z-p)^{-1/4})),& z\to p\mathrm{~from~}\mathbb{C}^-.
       \end{array}\right.
				\end{equation}
	\end{enumerate}
\end{RH}

Similarly, we obtain the following reconstruction formulae:
\begin{equation}\label{q3}
	\begin{aligned}
		\overline{q}(x,t)=-
		\lim_{z\to0}
		\frac{
			\partial_y\left(\mathcal{N}^{-1}(y,t;0)\mathcal{N}(y,t;z)\right)_{21}
		}{z}\left(1+\mathrm{i}\lim_{\substack{z\to 0}}
		\frac{\partial_y\left(\mathcal{N}^{-1}(y,t;0)\mathcal{N}(y,t;z)\right)_{11}}{z}\right)^{-1},
	\end{aligned}
\end{equation}
and
\begin{equation}\label{x3}
	\begin{aligned}
		x(y,t)=y+\mathrm{i}\lim_{z\to0}
		\frac{
			\left(\mathcal{N}^{-1}(y,t;0)\mathcal{N}(y,t;z)\right)_{11}-1
		}{z}.
	\end{aligned}
\end{equation}

Let $z_{j_1}\in(E_{j_1},\widehat{E}_{j_1})$ be fixed, and
\begin{equation}
	\begin{aligned}
		&\Omega_3=\Omega_3(\xi)=\{z\in\mathbb{C}:\pi-\varphi_2\leq\arg(z-z_1)\leq\pi\},\quad \Sigma_3=\Sigma_3(\xi)=z_1+e^{i\varphi_2}\mathbb{R}^+,\\
		&\Omega_4=\Omega_4(\xi)=\{z\in\mathbb{C}:0\leq\arg(z-z_{j_1})\leq\varphi_2\},\quad \Sigma_4=\Sigma_4(\xi)=z_{j_1}+e^{i(\pi-\varphi_2)}\mathbb{R}^+,
	\end{aligned}
\end{equation}
where $\varphi_{2}$ is chosen such that
\begin{equation}
	\operatorname{Im}\theta(z)
	\begin{cases}
		<0,\quad z\in\Sigma_3\setminus\{z_1\}, \\
		>0,\quad z\in\Sigma_4\setminus\{z_{j_1}\}.
	\end{cases}
\end{equation}

Similarly, the jump matrix $J_{\mathcal{N}}$ admits the following factorizations:
\begin{equation}
	\begin{aligned}
		J_{\mathcal{N}}(z) & =
		\begin{pmatrix}
			1 & \frac{r_2(z)e^{2it\theta(z)}}{1-|r_2(z)|^2} \v\\
			0 & 1
		\end{pmatrix}(1-|r_2(z)|^2)^{\sigma_3}
		\begin{pmatrix}
			1 & 0 \v\\
			-\frac{\overline{r_2(z)}e^{-2it\theta(z)}}{1-|r_2(z)|^2} & 1
		\end{pmatrix} \\
		& =
		\begin{pmatrix}
			1 & 0 \v\\
			-r_2(z)e^{-2it\theta(z)} & 1
		\end{pmatrix}
		\begin{pmatrix}
			1 & \overline{r_2(z)}e^{2it\theta(z)} \v\\
			0 & 1
		\end{pmatrix},\quad z\in\mathbb{R}\setminus(\cup_{j=0}^n[E_j,\widehat{E}_j]),
	\end{aligned}
\end{equation}
and
\begin{equation}
	\begin{aligned}
		J_{\mathcal{N}}\left(z\right)& =
		\begin{pmatrix}
			1 & 0 \v \\
			-r_2(z)e^{-2it\theta_-(z)} & 1
		\end{pmatrix}
		\begin{pmatrix}
			0 & -ie^{i(t(\theta_+(z)+\theta_-(z))-\phi_j)} \v\\
			-ie^{-i(t(\theta_+(z)+\theta_-(z))-\phi_j)} & 0
		\end{pmatrix}  \v\\
		& \qquad\qquad \times
		\begin{pmatrix}
			1 & r_2(z)e^{2i(t\theta_+(z)-\phi_j)} \v\\
			0 & 1
		\end{pmatrix} \\
		& =
		\begin{pmatrix}
			1 & \frac{r_2(z)e^{2i(t\theta_-(z)-\phi_j)}}{1-(r_2(z))^2e^{-2i\phi_j}} \v\\
			0 & 1
		\end{pmatrix}
		\begin{pmatrix}
			0 & -ie^{i(t(\theta_+(z)+\theta_-(z))-\phi_j)} \v\\
			-ie^{-i(t(\theta_+(z)+\theta_-(z))-\phi_j)} & 0
		\end{pmatrix} \v\\
		&\qquad\qquad \times
		\begin{pmatrix}
			1 & -\frac{r_2(z)e^{-2it\theta_+(z)}}{1-(r_2(z))^2e^{-2i\phi_j}} \v\\
			0 & 1
		\end{pmatrix},\quad z\in\cup_{j=0}^n(E_j,\widehat{E}_j).
	\end{aligned}
\end{equation}

Based on these factorizations, we define the matrix-valued function $\mathcal{N}^{(1)}(z)$ by
\begin{equation}
	\mathcal{N}^{(1)}(z)=e^{\delta(\infty)\sigma_3}\mathcal{N}(z)\mathcal{G}(z)e^{-\delta(z)\sigma_3},
\end{equation}
where
\begin{equation}
	\mathcal{G}(z)=
	\begin{cases}
		\begin{pmatrix}
			1 & -\overline{r_2(z)}e^{2it\theta(z)} \\
			0 & 1
		\end{pmatrix},\quad z\in\Omega_3, \v\\
		\begin{pmatrix}
			1 & 0 \\
			-r_2(z)e^{-2it\theta(z)} & 1
		\end{pmatrix},\quad z\in\Omega_3^*, \v\\
		\begin{pmatrix}
			1 & 0 \\
			\frac{r_2(z)e^{-2it\theta(z)}}{1-|r_2(z)|^2} & 1
		\end{pmatrix},\quad z\in\Omega_4, \v\\
		\begin{pmatrix}
			1 & \frac{\overline{r_2(z)}e^{2it\theta(z)}}{1-|r_2(z)|^2} \\
			0 & 1
		\end{pmatrix},\quad z\in\Omega_4^*, \v\\
		I, \quad \text{elsewhere,}
	\end{cases}
\end{equation}
and the function $\delta(z)$ is defined by
\begin{equation}\label{delta2}
	\delta(z)=\frac{w(z)}{2\pi i}\left[\sum_{j=1}^n\delta_j\int_{E_j}^{\widehat{E}_j}\frac{i\mathrm{~d}s}{w_+(s)(s-z)}+\int_{(\widehat{E}_{j_1},+\infty)\setminus(\cup_{j=j_1}^{n}(E_j,\widehat{E}_j))}\frac{\log(1-|r_2(s)|^2)\mathrm{~d}s}{w(s)(s-z)}\right],
\end{equation}
where the logarithm is taken on the principal branch, and the constants $\delta_j$, $j=1,\ldots,n$, are determined by the linear system
\begin{equation}\label{delta_2}
	\int_{(\widehat{E}_{j_1},+\infty)\setminus(\cup_{j=j_1}^{n}(E_j,\widehat{E}_j))}\frac{\log(1-|r_2(s)|^2)s^k\mathrm{d}s}{w(s)}+\sum_{j=1}^n\delta_j\int_{E_j}^{\widehat{E}_j}\frac{is^k\mathrm{d}s}{w_+(s)}=0,\quad k=0,\ldots,n-1.
\end{equation}

\begin{RH}
	$\delta(z)$ satisfies the following Riemann--Hilbert problem:
	\begin{enumerate}[label=(\textbf{$\delta$\arabic*}), leftmargin=*]
		\item \textbf{Analyticity:} $\delta(z)$ is analytic in $\mathbb{C}\setminus(\cup_{j=0}^n[E_j,\widehat{E}_j]\cup[z_1,+\infty))$.
		
		\item \textbf{Jump condition:} $\delta(z)$ satisfies the jump relation
		\begin{equation}
		\delta_-(z)=\left\{\begin{aligned}
				&\delta_+(z)+\log(1-|r_2(z)|^2),\quad &z\in(\widehat{E}_{j_1},+\infty)\setminus(\cup_{j=j_1}^{n}[E_j,\widehat{E}_j]), \\
				&-\delta_+(z)+i\delta_j,\quad &z\in(E_j,\widehat{E}_j),\quad j=1,\ldots,n.
			\end{aligned} \right.
		\end{equation}
		
		\item \textbf{Normalization at infinity:} As $z \to \infty$,
		 $
			\delta(z)=\delta(\infty)+\frac{\delta^{(1)}}{z}+\mathcal{O}(z^{-2}),\quad z\to\infty,$
				where
		\begin{equation}\label{delta_infty_2}
			\begin{aligned}
				\delta(\infty) =-\frac{1}{2\pi i}\left[\int_{(\widehat{E}_{j_1},+\infty)\setminus(\cup_{j=j_1}^{n}(E_j,\widehat{E}_j))}\frac{\log(1-|r_2(s)|^2)s^n \mathrm{d}s}{w(s)} +\sum_{j=1}^n\delta_j\int_{E_j}^{\widehat{E}_j}\frac{is^n\mathrm{d}s}{w_+(s)}\right],
			\end{aligned}
		\end{equation}
		and
		\begin{equation}
			\begin{aligned}
				&\delta^{(1)} =-\delta(\infty)\sum_{j=0}^n(E_j+\widehat{E}_j) \\
				&\qquad\qquad -\frac{1}{2\pi i}\left[\int_{(\widehat{E}_{j_1},+\infty)\setminus(\cup_{j=j_1}^{n}(E_j,\widehat{E}_j))}\frac{\log(1-|r_2(s)|^2)s^{n+1} \mathrm{d}s}{w(s)}+\sum_{j=1}^n\delta_j\int_{E_j}^{\widehat{E}_j}\frac{is^{n+1}\mathrm{d}s}{w_+(s)}\right].
			\end{aligned}
		\end{equation}
		
		\item \textbf{Local behaviors:}
		\begin{equation}
			e^{\delta(z)}=\mathcal{O}((z-p)^{-1/2}),\quad z\to p\in\{E_j\}_{j=j_1+1}^{n}\cup\{\widehat{E}_j\}_{j=j_1}^{n}\,\, \text{from}\,\, \mathbb{C}^+,
		\end{equation}
		and
		\begin{equation}
			\delta(z)=\frac{i}{2}\delta_{j_1}+\mathcal{O}((z-E_{j_1})^{1/2}),\quad z\to E_{j_1}\,\, \text{~from~} \,\, \mathbb{C}\setminus(-\infty,E_{j_1}).
		\end{equation}
		
		\item \textbf{Symmetry relation:} $\delta(z)$ satisfies the symmetry relation $\delta(z)=-\overline{\delta(\bar{z})}$.
	\end{enumerate}
\end{RH}

\begin{RH}
	$\mathcal{N}^{(1)}(z)$ satisfies the following Riemann--Hilbert problem:
	\begin{enumerate}[label=(\textbf{$\mathcal{N}^{(1)}$\arabic*}), leftmargin=*]
		\item \textbf{Analyticity:} $\mathcal{N}^{(1)}(z)$ is analytic in $\mathbb{C}\setminus\Sigma^{(1)}$, where $\Sigma^{(1)}$, is defined by (see Fig.~\ref{fig:WKI11})
	$			\Sigma^{(1)}=[z_1,E_{j_1}]\cup(\cup_{j=0}^n[E_j,\widehat{E}_j])\cup\Sigma_3\cup\Sigma_3^*\cup\Sigma_4\cup\Sigma_4^*.
	$
		
		\item \textbf{Jump condition:} $\mathcal{N}^{(1)}_+(z)=\mathcal{N}^{(1)}_-(z)J_{\mathcal{N}}^{(1)}(z)$ for $z\in\Sigma^{(1)}$, where the jump matrix is given by
		\begin{equation}
			J_{\mathcal{N}}^{(1)}(z)=\begin{cases}
				e^{i(f_0y+g_0t-(B_j^fy+B_j^gt+\phi_j-\delta_j)/2)\widehat{\sigma}_3}\begin{pmatrix}0&-i\\-i&0\end{pmatrix},&z\in(E_j,\widehat{E}_j),\\[2.5ex]
				\begin{pmatrix}1&\overline{r_2(z)}e^{2it\theta(z)+2\delta(z)}\\0&1\end{pmatrix},&z\in\Sigma_3,\\[2.5ex]
				\begin{pmatrix}1&0\\-r_2(z)e^{-2it\theta(z)-2\delta(z)}&1\end{pmatrix},&z\in\Sigma_3^*,\\[2.5ex]
				\begin{pmatrix}1&0\\\dfrac{-r_2(z)e^{-2it\theta(z)-2\delta(z)}}{1-|r_2(z)|^2}&1\end{pmatrix},&z\in\Sigma_4,\\[3.5ex]
				\begin{pmatrix}1&\dfrac{\overline{r_2(z)}e^{2it\theta(z)+2\delta(z)}}{1-|r_2(z)|^2}\\ 0&1\end{pmatrix},&z\in\Sigma_4^*,\\[3ex]
				\begin{pmatrix}1&\overline{r_2(z)}e^{2it\theta(z)+2\delta(z)}\\ -r_2(z)e^{-2it\theta(z)-2\delta(z)}&1-|r_2(z)|^2\end{pmatrix},&z\in(z_1,E_{j_1}).
			\end{cases}
		\end{equation}
		
		\item \textbf{Normalization at infinity:} As $z\to\infty$, $\mathcal{N}^{(1)}(z)=I+\mathcal{O}(z^{-1})$.
		
		\item \textbf{Local behavior:} For $p\in\{E_j,\widehat{E}_j\}_{j=0}^n\setminus\{E_{j_1}\}$, we have
		$
			\mathcal{N}^{(1)}(z)=\mathcal{O}((z-p)^{-1/4}),\quad z\to p.
		$
	\end{enumerate}
\end{RH}

\begin{figure}[!t]
	\centering
	\includegraphics[scale=0.3]{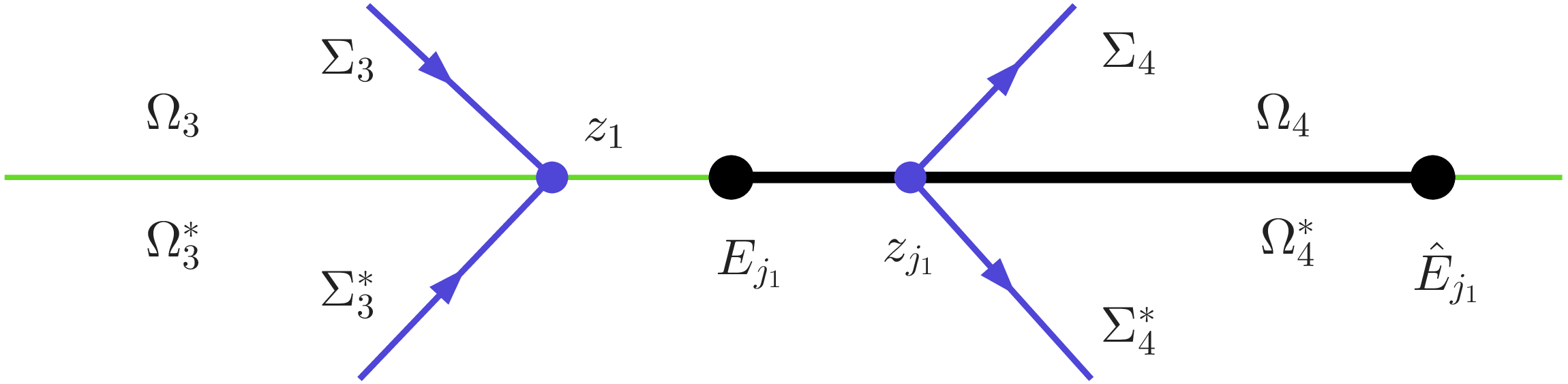}
	\caption{The contours $\Sigma^{(1)}$.}
	\label{fig:WKI11}
\end{figure}

We set
\begin{equation}
	\mathcal{N}^{(1)}(z)=
	\begin{cases}
		E(z)\mathcal{N}^{(glo)}(z), & z\in\mathbb{C}\setminus U_2, \\
		E(z)\mathcal{N}^{(loc)}(z), & z\in U_2.
	\end{cases}
\end{equation}
Here $\mathcal{N}^{(glo)}(z)$ and $\mathcal{N}^{(loc)}(z)$ are the global and local parametrices for $\mathcal{N}^{(1)}(z)$, respectively, and $U_{2}=\{z\in\mathbb{C}: |z-E_{j_1}|<\varepsilon_2\}$ is centered at $E_{j_1}$, with fixed radius $\varepsilon_2>0$ chosen sufficiently small. In particular, we may choose
\begin{equation}
	\varepsilon_2=\min\left\{\frac{1}{2}(z_{j_1}-E_{j_1}),\frac{1}{2}(E_{j_1+1}-\widehat{E}_{j_1}),2(z_1-\widehat{E}_{j_1})t^\epsilon\right\},\quad\epsilon\in(1/3,2/3).
\end{equation}

As in transition region I, $\mathcal{N}^{(glo)}(z)$ satisfies the same global model Riemann--Hilbert problem, whose solution is explicitly given by \eqref{M_glo}.

\begin{RH}
	$\mathcal{N}^{(loc)}(z)$ satisfies the following Riemann--Hilbert problem:
	\begin{enumerate}[label=(\textbf{$\mathcal{N}^{(loc)}$\arabic*}), leftmargin=*]
		\item \textbf{Analyticity:} $\mathcal{N}^{(loc)}(z)$ is analytic in $U_2\setminus\Sigma^{(1)}$.
		
		\item \textbf{Jump condition:}  $\mathcal{N}^{(loc)}_+(z)=\mathcal{N}^{(loc)}_-(z)J_{\mathcal{N}}^{(1)}(z)$ for $z\in\Sigma^{(1)}\cap U_2$.
		
		\item \textbf{Matching condition:} As $t\to\infty$, $\mathcal{N}^{(loc)}(z)$ matches $\mathcal{N}^{(glo)}(z)$ on the boundary $\partial U_2$ of $U_2$.
		
		\item \textbf{Local behavior:} $\mathcal{N}^{(loc)}(z)=\mathcal{O}((z-E_{j_1})^{-1/4})$ as $z\to E_{j_1}$.

	\end{enumerate}
\end{RH}

Similarly, to construct the solution, we need a suitable conformal map that captures the local behavior of the phase function near the endpoint. For $\theta$, we have the following local expansion:
\begin{equation}
	\theta(z)=\theta^{(0,j_1)}+(\xi-\xi_{j_1})\theta^{(1,j_1)}(E_{j_1}-z)^{1/2}+\frac{2}{3}\theta^{(3,j_1)}(E_{j_1}-z)^{3/2}+O\left((E_{j_1}-z)^{5/2}\right),
\end{equation}
as $z\to E_{j_1}$, for fixed $\xi$, where
\begin{equation}\no
	\theta^{(0,j_1)}=f_0\xi+g_0-\frac{1}{2}(B_{j_1}^f\xi+ B_{j_1}^g),\qquad
	\theta^{(1,j_1)}=-\frac{2\prod_{j=0}^n(E_{j_1}-z_{j_1}^f)}{w^{(j_1)}(E_{j_1})},
\end{equation}
\begin{equation}\no
	\theta^{(3,j_1)}=\frac{1}{w^{(j_1)}(E_{j_1})^2}\left[T(E_{j_1};\xi)
(w^{(j_1)})^{\prime}(E_{j_1})-\partial_zT(E_{j_1};\xi)w^{(j_1)}(E_{j_1})\right].
\end{equation}
This local expansion motivates the definition
\begin{equation}
	\zeta(z)=\left(\frac{3it}{2}(\theta(E_{j_1})+(\xi-\xi_{j_1})\theta^{(1,j_1)}(E_{j_1}-z)^{1/2}-\theta(z))\right)^{2/3},\quad z\in U_2,
\end{equation}
which is a one-to-one conformal map in $U_2$ with respect to $z$. Moreover, it is readily seen that
\begin{equation}
	\zeta(E_{j_1})=0,\quad\zeta^{\prime}(E_{j_1})=-|\theta^{(3,j_1)}(\xi_{j_1})|^{2/3}t^{2/3}<0.
\end{equation}

As in transition region I, the local parametrix $\mathcal{N}^{(loc)}(z)$ is given by
\begin{equation}
	\mathcal{N}^{(loc)}(z)=\tilde{H}_1(z)t^{\sigma_3/6}
	\begin{pmatrix}
		1 & 0 \\
		ia(s) & 1
	\end{pmatrix}M^{(P_{34})}(\zeta(z);s,-1/4,0)e^{\widehat{\theta}(\zeta(z))\sigma_3}\mathcal{G}_1(\zeta(z))\mathcal{G}_2(\zeta(z)),
\end{equation}
where
\begin{equation}\no
	\mathcal{G}_1(\zeta)=
	\begin{cases}
		\begin{pmatrix}
			0 & -1 \\
			1 & 0
		\end{pmatrix}e^{-\pi i\sigma_3/4}e^{-i(f_0y+g_0t-(B_{j_1}^fy+B_{j_1}^gt+\delta_{j_1})\sigma_3/2)}, & \zeta\in\mathbb{C}^-, \\
		e^{-\pi i\sigma_3/4}e^{-i(f_0y+g_0t-(B_{j_1}^fy+B_{j_1}^gt+\delta_{j_1})\sigma_3/2)}, & \zeta\in\mathbb{C}^+,
	\end{cases}
\end{equation}
\begin{equation}\no
	\begin{aligned}
		\mathcal{G}_2(\zeta)=e^{it\theta(E_{j_1})\widehat{\sigma}_3}
		\begin{cases}
			\begin{pmatrix}
				1 & -\overline{r_2(E_{j_1})}e^{2\theta(\zeta)+i\delta_{j_1}} \\
				0 & 1
			\end{pmatrix}, & \zeta\in\Omega_3^{(loc)}, \\
			\begin{pmatrix}
				1 & 0 \\
				-r_2(E_{j_1})e^{2\theta(\zeta)-i\delta_{j_1}} & 1
			\end{pmatrix}, & \zeta\in\Omega_3^{(loc)*}, \\
			I, & \text{elsewhere,}
		\end{cases}
	\end{aligned}
\end{equation}
and
\begin{equation}
	\tilde{H}_1(z)=\mathcal{N}^{(glo)}(z)\mathcal{G}_1(\zeta(z))^{-1}\frac{1}{\sqrt{2}}\begin{pmatrix}
		1 & -i \\
		-i & 1
	\end{pmatrix}\left((z-E_{j_1})|\theta^{(3,j_1)}(\xi)|^{2/3}\right)^{\sigma_3/4}.
\end{equation}
Using $\mathcal{N}^{(glo)}(z)$ and $\mathcal{N}^{(loc)}(z)$, we obtain
\begin{equation}
	\mathcal{N}^{(loc)}(z)\mathcal{N}^{(glo)}(z)^{-1}=I-\frac{t^{1/3}}{\zeta(z)}\tilde{H}_1(z)
	\begin{pmatrix}
		0 & ia(s) \\
		0 & 0
	\end{pmatrix}\tilde{H}_1(z)^{-1}+\mathcal{O}(t^{1/3-2\epsilon}),
\end{equation}
as $t\to\infty$, for $z\in\partial U_2$.

\begin{figure}[!t]
	\centering
	\includegraphics[scale=0.25]{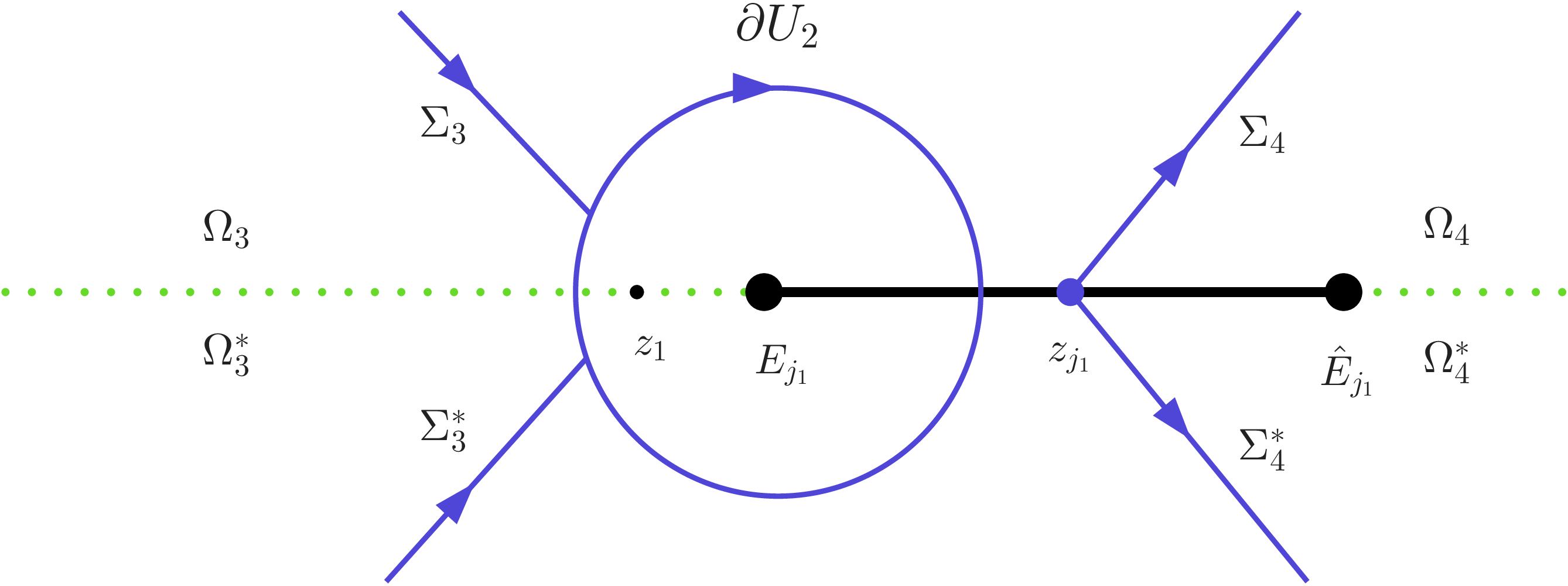}
	\caption{The contours $\Sigma^{(E)}$.}
	\label{fig:WKI12}
\end{figure}

\begin{RH}
	$E(z)$ satisfies the following Riemann--Hilbert problem:
	\begin{enumerate}[label=(\textbf{$E$\arabic*}), leftmargin=*]
		\item \textbf{Analyticity:} $E(z)$ is analytic in $\mathbb{C}\setminus \Sigma^{(E)}$, where the jump contour $\Sigma^{(E)}$ is defined by $
			\Sigma^{(E)}=\partial U_2\cup\Sigma_3\cup\Sigma_3^*\cup\Sigma_4\cup\Sigma_4^*\setminus U_2$ (see  Fig.~\ref{fig:WKI12}).

		\item \textbf{Jump condition:} $E_+(z)=E_-(z)J_{E}(z)$, where the jump matrix is given by
		\begin{equation}
			J_{E}(z)=\left\{
			\begin{array}{ll}
\mathcal{N}^{(glo)}(z)J_\mathcal{N}^{(1)}(z)\mathcal{N}^{(glo)}(z)^{-1}, & z\in\Sigma^{(E)}\setminus\partial U_2, \v\\
				\mathcal{N}^{(loc)}(z)\mathcal{N}^{(glo)}(z)^{-1}, & z\in\partial U_2.
			\end{array}\right.
		\end{equation}
		
		\item \textbf{Asymptotics at infinity:} As $z\to\infty$, $E(z)=I+\mathcal{O}(z^{-1})$.
		
		\item \textbf{Local behavior:} $E(z)=\mathcal{O}((z-E_{j_1})^{-1/2})$ \,\ as\,\, $z\to E_{j_1}$.
	\end{enumerate}
\end{RH}

By the standard small-norm Riemann--Hilbert theory, there exists a unique solution to the Riemann--Hilbert problem for $E$ for sufficiently large positive $t$. Moreover, as $z\to0$, we have
\begin{equation}
	E(z)=E(0)+E_1z+O(z^2),
\end{equation}
where
\begin{equation} \no
	E(0)
	=
	I
	-\frac{t^{-1/3}}{|\theta^{(3,j_1)}(\xi_{j_1})|^{2/3}}\tilde{H}_1(E_{j_1})
	\begin{pmatrix}
		0 & ia \\
		0 & 0
	\end{pmatrix}\tilde{H}_1(E_{j_1})^{-1}+\mathcal{O}(t^{-\epsilon}),
\end{equation}
and
\begin{equation} \no
	E_1
	=
	\frac{t^{-1/3}}{|\theta^{(3,j_1)}(\xi_{j_1})|^{2/3}E_{j_1}}\tilde{H}_1(E_{j_1})
	\begin{pmatrix}
		0 & ia \\
		0 & 0
	\end{pmatrix}\tilde{H}_1(E_{j_1})^{-1}+\mathcal{O}(t^{-\epsilon}).
\end{equation}
Similarly, a direct calculation gives
\begin{equation} \no
	E(0)^{-1}=I+\mathcal{O}(t^{-1/3}).
\end{equation}

\section{Asymptotic analysis in the Zakharov-Manakov region III and fast-decay region IV}\label{RegionIII-IV}

In this section, we carry out the asymptotic analysis of the Riemann--Hilbert problem for $\mathcal{M}$ when $\xi$ belongs to the Zakharov--Manakov region, that is,
$\xi\in(-\infty,\widehat{\xi}_n)\cup_{j=1}^n(\xi_j,\widehat{\xi}_{j-1})\cup(\xi_0,+\infty)$, and to the fast-decay region, that is, $\xi\in\cup_{j=0}^n(\widehat{\xi}_j,\xi_j)$. In the Zakharov--Manakov region, the saddle point satisfies $z_1\in\mathbb{R}\setminus(\cup_{j=0}^n[E_j,\widehat{E}_j])$, whereas in the fast-decay region it satisfies $z_1\in\cup_{j=0}^n(E_j,\widehat{E}_j)$. We set
\begin{equation}
	\begin{aligned}
		&\Omega_1=\Omega_1(\xi)=\{z\in\mathbb{C}:0\leq\arg(z-z_1)\leq\varphi_1\},\quad \Sigma_1=\Sigma_1(\xi)=z_1+e^{i\varphi_1}\mathbb{R}^+,\\
		&\Omega_2=\Omega_2(\xi)=\{z\in\mathbb{C}:\pi-\varphi_1\leq\arg(z-z_1)\leq\pi\},\quad
\Sigma_2=\Sigma_2(\xi)=z_1+e^{i(\pi-\varphi_1)}\mathbb{R}^+,
	\end{aligned}
\end{equation}
where $\varphi_{1}$ is chosen such that
\begin{equation}
	\operatorname{Im}\theta(z)
	\begin{cases}
		<0, \quad z\in\Sigma_1\setminus\{z_1\}, \\
		>0,\quad z\in\Sigma_2\setminus\{z_1\}.
	\end{cases}
\end{equation}

Similarly, the first transformation for $\mathcal{M}$ is defined by
\begin{equation}
	\mathcal{M}^{(1)}(z)=e^{\delta(\infty)\sigma_3}\mathcal{M}(z)\mathcal{G}(z)e^{-\delta(z)\sigma_3},
\end{equation}
where
\begin{equation}
	\mathcal{G}(z)=
	\begin{cases}
		\begin{pmatrix}
			1 & 0 \\
			\overline{r_1(z)}e^{-2it\theta(z)} & 1
		\end{pmatrix},\quad z\in\Omega_1, \v\\
		\begin{pmatrix}
			1 & r_1(z)e^{2it\theta(z)} \\
			0 & 1
		\end{pmatrix},\quad z\in\Omega_1^*, \v\\
		\begin{pmatrix}
			1 & -\frac{r_1(z)e^{2it\theta(z)}}{1-|r_1(z)|^2} \\
			0 & 1
		\end{pmatrix},\quad z\in\Omega_2, \v\\
		\begin{pmatrix}
			1 & 0 \\
			-\frac{\overline{r_1(z)}e^{-2it\theta(z)}}{1-|r_1(z)|^2} & 1
		\end{pmatrix},\quad z\in\Omega_2^*, \v\\
		I, \quad \text{elsewhere.}
	\end{cases}
\end{equation}

The function $\delta(z)$ is defined by
\begin{equation}\label{delta3}
	\delta(z)=\frac{w(z)}{2\pi i}\left[\sum_{j=1}^n\delta_j\int_{E_j}^{\widehat{E}_j}\frac{i\mathrm{~d}s}{w_+(s)(s-z)}-\int_{(-\infty,z_1)\setminus(\cup_{j=0}^{n}(E_j,\widehat{E}_j))}\frac{\log(1-|r_1(s)|^2)\mathrm{~d}s}{w(s)(s-z)}\right],
\end{equation}
where the logarithm is taken on the principal branch, and the constants $\delta_j$, $j=1,\ldots,n$, are determined by the linear system
\begin{equation}\label{delta_3}
	\int_{(-\infty,z_1)\setminus(\cup_{j=0}^{n}(E_j,\widehat{E}_j))}\frac{\log(1-|r_1(s)|^2)s^k\mathrm{d}s}{w(s)}-\sum_{j=1}^n\delta_j\int_{E_j}^{\widehat{E}_j}\frac{is^k\mathrm{d}s}{w_+(s)}=0,\quad k=0,\ldots,n-1.
\end{equation}

\begin{RH}
	$\delta(z)$ satisfies the following Riemann--Hilbert problem:
	\begin{enumerate}[label=(\textbf{$\delta$\arabic*}), leftmargin=*]
		\item \textbf{Analyticity:} $\delta(z)$ is analytic in $\mathbb{C}\setminus(\cup_{j=0}^n[E_j,\widehat{E}_j]\cup(-\infty,z_1])$.
		
		\item \textbf{Jump condition:} $\delta(z)$ satisfies the jump relation
		\begin{equation}
			\delta_-(z)=\left\{\begin{aligned}
				& \delta_+(z)+\log(1-|r_1(z)|^2),\quad &z\in(-\infty,z_1)\setminus(\cup_{j=0}^{n}[E_j,\widehat{E}_j]), \\
				&-\delta_+(z)+i\delta_j,\quad &z\in(E_j,\widehat{E}_j),\quad j=1,\ldots,n.
			\end{aligned}\right.
		\end{equation}
		
		\item \textbf{Normalization at infinity:} $
			\delta(z)=\delta(\infty)+\frac{\delta^{(1)}}{z}+\mathcal{O}(z^{-2}),\quad z\to\infty,$
				where
		\begin{equation}\label{delta_infty_3}
			\begin{aligned}
				\delta(\infty) =\frac{1}{2\pi i}\left[\int_{(-\infty,z_1)\setminus(\cup_{j=0}^{n}(E_j,\widehat{E}_j))}\frac{\log(1-|r_1(s)|^2)s^n \mathrm{d}s}{w(s)} -\sum_{j=1}^n\delta_j\int_{E_j}^{\widehat{E}_j}\frac{is^n\mathrm{d}s}{w_+(s)}\right],
			\end{aligned}
		\end{equation}
		and
		\begin{equation}
			\begin{aligned}
				\delta^{(1)} & =\delta(\infty)\sum_{j=0}^n(E_j+\widehat{E}_j) \\
				& \quad -\frac{1}{2\pi i}\left[\int_{(-\infty,z_1)\setminus(\cup_{j=0}^{n}(E_j,\widehat{E}_j))}\frac{\log(1-|r_1(s)|^2)s^{n+1} \mathrm{d}s}{w(s)}-\sum_{j=1}^n\delta_j\int_{E_j}^{\widehat{E}_j}\frac{is^{n+1}\mathrm{d}s}{w_+(s)}\right].
			\end{aligned}
		\end{equation}
	\end{enumerate}
\end{RH}

\begin{RH}
	$\mathcal{M}^{(1)}(z)$ satisfies the following Riemann--Hilbert problem:
	\begin{enumerate}[label=(\textbf{$\mathcal{M}^{(1)}$\arabic*}), leftmargin=*]
		\item \textbf{Analyticity:} $\mathcal{M}^{(1)}(z)$ is analytic in $\mathbb{C}\setminus\Sigma^{(1)}$, where $\Sigma^{(1)}$ is defined by (see Figs.~\ref{fig:WKI13} and~\ref{fig:WKI14})
		$	\Sigma^{(1)}=(\cup_{j=0}^n[E_j,\widehat{E}_j])\cup\cup_{i=1}^2(\Sigma_i\cup\Sigma_i^*).$
		
		\item \textbf{Jump condition:} $\mathcal{M}^{(1)}_+(z)=\mathcal{M}^{(1)}_-(z)J_{\mathcal{M}}^{(1)}(z)$ for $z\in\Sigma^{(1)}$, where the jump matrix is given by
		\begin{equation}
			J_{\mathcal{M}}^{(1)}(z)=\begin{cases}
				-ie^{i[f_0y+g_0t-(B_j^fy+B_j^gt+\phi_j-\delta_j)/2]\widehat{\sigma}_3}\sigma_1,&z\in(E_j,\widehat{E}_j),\\[2.5ex]
				\begin{pmatrix}1&0\\-\overline{r_1(z)}e^{-2it\theta(z)-2\delta(z)}&1\end{pmatrix},&z\in\Sigma_1,\\[2.5ex]
				\begin{pmatrix}1&r_1(z)e^{2it\theta(z)+2\delta(z)}\\0&1\end{pmatrix},&z\in\Sigma_1^*,\\[2.5ex]
				\begin{pmatrix}1&\dfrac{r_1(z)e^{2it\theta(z)+2\delta(z)}}{1-|r_1(z)|^2}\\0&1\end{pmatrix},&z\in\Sigma_2,\\[2.5ex]
				\begin{pmatrix}1&0\\ \dfrac{-\overline{r_1(z)}e^{-2it\theta(z)-2\delta(z)}}{1-|r_1(z)|^2}&1\end{pmatrix},&z\in\Sigma_2^*.
			\end{cases}
		\end{equation}
		
		\item \textbf{Normalization at infinity:} As $z\to\infty$, we have $\mathcal{M}^{(1)}(z)=I+\mathcal{O}(z^{-1})$.
		
		\item \textbf{Local behavior:} For $p\in\{E_j,\widehat{E}_j\}_{j=0}^n$, we have
		$	\mathcal{M}^{(1)}(z)=\mathcal{O}((z-p)^{-1/4}),\quad z\to p.$
	\end{enumerate}
\end{RH}

\begin{figure}[!t]
	\centering
	\includegraphics[scale=0.3]{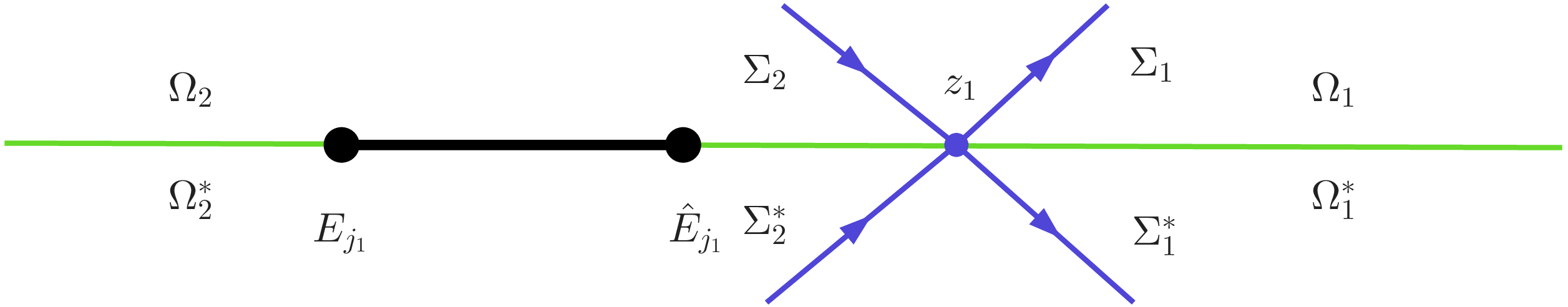}
	\caption{The contours $\Sigma^{(1)}$ in the Zakharov--Manakov region.}
	\label{fig:WKI13}
\end{figure}

\begin{figure}[!t]
	\centering
	\includegraphics[scale=0.3]{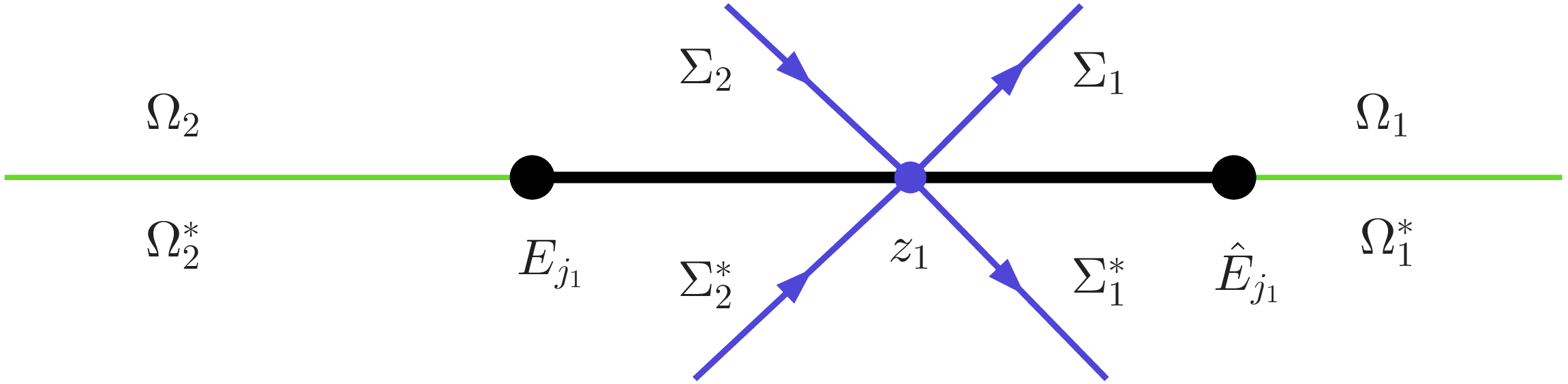}
	\caption{The contours $\Sigma^{(1)}$ in the fast-decay region.}
	\label{fig:WKI14}
\end{figure}

We set
\begin{equation}
	\mathcal{M}^{(1)}(z)=
	\begin{cases}
		E(z)\mathcal{M}^{(glo)}(z), & z\in\mathbb{C}\setminus U_3, \\
		E(z)\mathcal{M}^{(loc)}(z), & z\in U_3.
	\end{cases}
\end{equation}
Here $\mathcal{M}^{(glo)}(z)$ and $\mathcal{M}^{(loc)}(z)$ are the global and local parametrices for $\mathcal{M}^{(1)}(z)$, respectively. The neighborhood $U_3$ is defined by
\begin{equation}
	U_3=
	\begin{cases}
		\{z:|z-z_1|\leq \varepsilon_3\}, & \xi\text{ belongs to the Zakharov--Manakov region,} \\
		\emptyset, & \xi\text{ belongs to the fast-decay region,}
	\end{cases}
\end{equation}
with
\begin{equation}
	\varepsilon_3=\min_{j=0,...,n}\left\{\frac{|\widehat{E}_{j}-z_1|}{2},\frac{|E_j-z_1|}{2}\right\}.
\end{equation}

As in transition region I, $\mathcal{M}^{(glo)}(z)$ satisfies the same global model Riemann--Hilbert problem, whose solution is explicitly given by \eqref{M_glo}. For the local parametrix, it remains only to consider the case in which $\xi$ belongs to the Zakharov--Manakov region.

\begin{RH}
	$\mathcal{M}^{(loc)}(z)$ satisfies the following Riemann--Hilbert problem:
	\begin{enumerate}[label=(\textbf{$\mathcal{M}^{(loc)}$\arabic*}), leftmargin=*]
		\item \textbf{Analyticity:} $\mathcal{M}^{(loc)}(z)$ is analytic in $U_3\setminus\Sigma^{(1)}$.
		
		\item \textbf{Jump condition:}  $\mathcal{M}^{(loc)}_+(z)=\mathcal{M}^{(loc)}_-(z)J_{\mathcal{M}}^{(1)}(z)$ for $z\in\Sigma^{(1)}\cap U_3$.
		
		\item \textbf{Matching condition:} As $t\to\infty$, $\mathcal{M}^{(loc)}(z)$ matches $\mathcal{M}^{(glo)}(z)$ on the boundary $\partial U_3$ of $U_3$.
	\end{enumerate}
\end{RH}

To construct the solution, we need a suitable conformal map that captures the local behavior of the phase function near the saddle point. We have the following local expansion:
\begin{equation}
	\theta(z)=\theta(z_1)-\theta^{(z_1,2)}(\xi)(z-z_1)^2+\mathcal{O}\left((z-z_1)^3\right),\quad z\to z_1,
\end{equation}
where
\begin{equation}\label{theta_2}
	\theta^{(z_1,2)}(\xi):=\frac{\partial_zT(z_1;\xi)}{w(z_1)}>0.
\end{equation}

Let
\begin{equation}
	\zeta=2(z-z_1)\sqrt{\theta^{(z_1,2)}(\xi)}t^{1/2}.
\end{equation}
Then this local Riemann--Hilbert problem can be solved explicitly using the parabolic-cylinder parametrix given in \cite{Fan2026,Its1981}. To this end, we obtain
\begin{equation}
	\mathcal{M}^{(loc)}(z)=I+\frac{
		t^{-1/2}}{2(z-z_1)\sqrt{\theta^{(z_1,2)}(\xi)}}\begin{pmatrix}
		0 & \beta_{21} \\
		\beta_{12} & 0
	\end{pmatrix}+\mathcal{O}(t^{-1}),\quad t\to\infty,
\end{equation}
where
\begin{equation}\label{beta}
	\beta_{12}=\frac{\sqrt{2\pi}e^{\frac{1}{4}(\pi i-\log(1-|r_0|^2))}}{r_0\Gamma(i\log(1-|r_0|^2)/2\pi)},\quad\beta_{21}=\frac{\log(1-|r_0|^2)}{2\pi\beta_{12}},
\end{equation}
\begin{equation}\no
	\begin{aligned}
		r_{0}& =-\overline{r_1(z_1)}\left(2\sqrt{\theta^{(z_1,2)}(\xi)}\right)^{\frac{i}{2\pi}\log(1-|r_1(z_1)|^2)} \\
		&\quad  \times \exp\left[{\frac{w(z_1)}{2\pi i}\sum_{j=0}^n\delta_j\int_{E_j}^{\widehat{E}_j}\frac{i\mathrm{d}s}{w_+(s)(s-z_1)}+\frac{\log(1-|r_1(z_1)|^2)\log(\varepsilon_3)}{2\pi i}}-2it\theta(z_1)\right] \\
		&\quad \times \exp\left\{\frac{w(z_1)}{2\pi i}\int_{(-\infty,z_1)\setminus(\cup_{j=0}^n(E_j,\widehat{E}_j))}\left[\frac{\chi_{(z_1-\varepsilon_3,z_1)}
\log(1-|r_1(z_1)|^2)}{w(z_1)(s-z_1)}-\frac{\log(1-|r_1(s)|^2)}{w(s)(s-z_1)}\right]\mathrm{d}s
\right\},
	\end{aligned}
\end{equation}
and $\chi_I$ denotes the characteristic function of the interval $I$.

\begin{figure}[!t]
	\centering
	\includegraphics[scale=0.22]{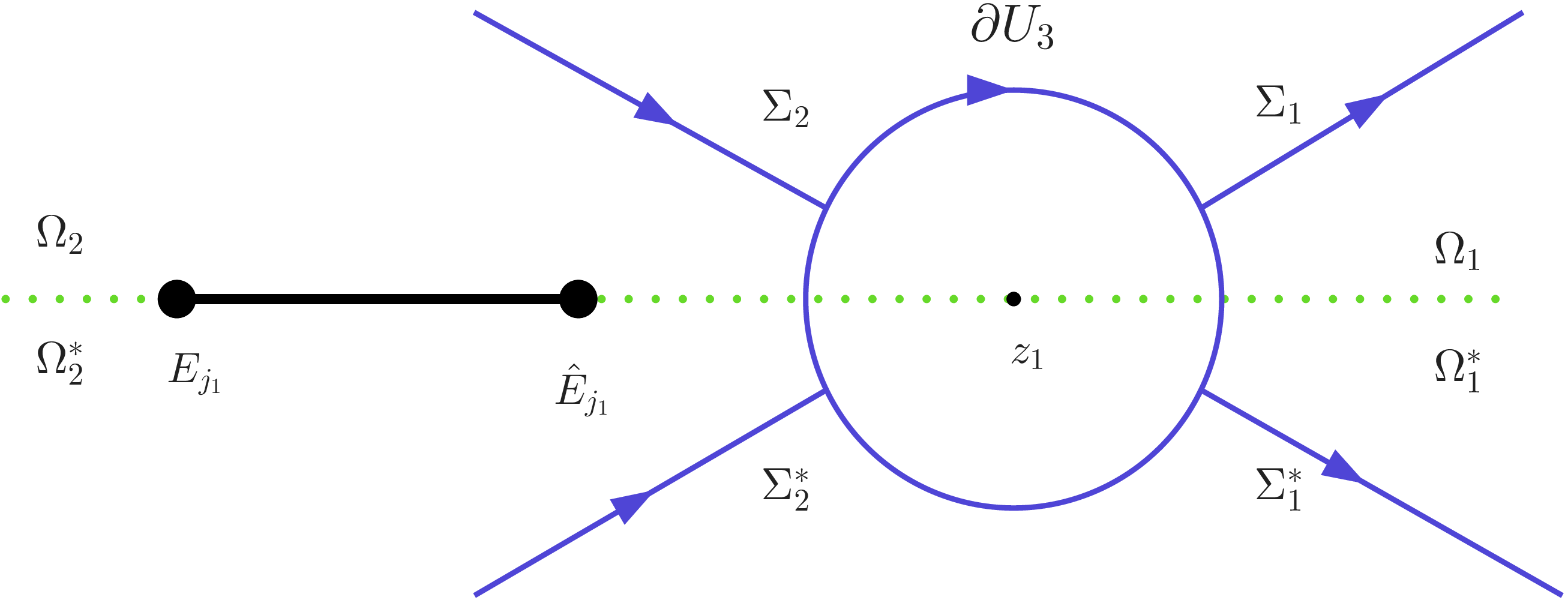}
	\caption{The contours $\Sigma^{(E)}$.}
	\label{fig:WKI15}
\end{figure}

\begin{RH}
	$E(z)$ satisfies the following Riemann--Hilbert problem:
	\begin{enumerate}[label=(\textbf{$E$\arabic*}), leftmargin=*]
		\item \textbf{Analyticity:} $E(z)$ is analytic in $\mathbb{C}\setminus \Sigma^{(E)}$, where the jump contour $\Sigma^{(E)}$ is defined by
		$\Sigma^{(E)}=\cup_{i=1}^2(\Sigma_i\cup\Sigma_i^*)\cup\partial U_3\setminus U_3$ (see Fig.~\ref{fig:WKI15}).
			
		\item \textbf{Jump condition:} On the contour $\Sigma^{(E)}$, $E(z)$ satisfies the jump relation $E_+(z)=E_-(z)J_{E}(z)$, where the jump matrix is given by
		\begin{equation}
			J_{E}(z)=\left\{
			\begin{array}{ll}\mathcal{M}^{(glo)}(z)J_\mathcal{M}^{(1)}(z)\mathcal{M}^{(glo)}(z)^{-1}, & z\in\Sigma^{(E)}\setminus\partial U_3, \v\\
				\mathcal{M}^{(loc)}(z)\mathcal{M}^{(glo)}(z)^{-1}, & z\in\partial U_3.
			\end{array}\right.
		\end{equation}
		
		\item \textbf{Asymptotics at infinity:} As $z\to\infty$,  $E(z)=I+\mathcal{O}(z^{-1})$.
	\end{enumerate}
\end{RH}

By the standard small-norm Riemann--Hilbert theory, there exists a unique solution to the Riemann--Hilbert problem for $E$ for sufficiently large positive $t$. When $\xi$ belongs to the fast-decay region, as $z\to0$, we have
\begin{equation}\label{E_2}
	E(z)=I+\mathcal{O}(t^{-1}).
\end{equation}

When $\xi$ belongs to the Zakharov--Manakov region, as $z\to0$, we have
\begin{equation}
	E(z)=E(0)+E_1z+O(z^2),
\end{equation}
where
\begin{equation}\no
	E(0)
	=
	I
	+\frac{t^{-1/2}}{2z_1\sqrt{\theta^{(z_1,2)}(\xi)}}\mathcal{M}^{(glo)}(z_1)
	\begin{pmatrix}
		0 & \beta_{21} \\
		\beta_{12} & 0
	\end{pmatrix}\mathcal{M}^{(glo)}(z_1)^{-1}+\mathcal{O}(t^{-1}),
\end{equation}
and
\begin{equation}\label{E_3}
	E_1
	=
	-\frac{t^{-1/2}}{2z_1^2\sqrt{\theta^{(z_1,2)}(\xi)}}\mathcal{M}^{(glo)}(z_1)
	\begin{pmatrix}
		0 & \beta_{21} \\
		\beta_{12} & 0
	\end{pmatrix}\mathcal{M}^{(glo)}(z_1)^{-1}+\mathcal{O}(t^{-1}).
\end{equation}
Similarly, a direct calculation gives
\begin{equation} \no
	E(0)^{-1}=I+\mathcal{O}(t^{-1/2}).
\end{equation}

\section{Proof of Theorem~\ref{theom}}\label{Proof-Regions}

We now establish the uniform asymptotic expansions for the solution $q(y,t)$ of the defocusing WKI equation (\ref{WKI}) by inverting the sequence of transformations performed in the previous sections.

i) For the asymptotics of $q(y,t)$ in transition region I, the transformations $\mathcal{M}\mapsto\mathcal{M}^{(1)}\mapsto E$ yield, as $t\to+\infty$,
\begin{equation}
	\mathcal{M}(z)=e^{-\delta(\infty)\sigma_3}E(z)\mathcal{M}^{(glo)}(z)e^{\delta(z)\sigma_3}\mathcal{G}(z)^{-1},\quad z\in\mathbb{C}\setminus U_1.
\end{equation}
Using $E(z)=E(0)+E_1z+O(z^2)$ as $z\to0$, we obtain
\begin{equation}
	\mathcal{M}(0)=e^{-\delta(\infty)\sigma_3}E(0)\mathcal{M}^{(glo)}(0)e^{\delta(0)\sigma_3}\mathcal{G}(0)^{-1},
\end{equation}
and
\begin{equation}
	\mathcal{M}(z)=e^{-\delta(\infty)\sigma_3}\left(E(0)+E_1z+\cdots\right)\mathcal{M}^{(glo)}(z)e^{\delta(z)\sigma_3}\mathcal{G}(z)^{-1}.
\end{equation}
Then a direct calculation gives
\begin{equation}
	\begin{aligned}
		\mathcal{M}(0)^{-1}\mathcal{M}(z)
		&=		\mathcal{G}(0)e^{-\delta(0)\sigma_3}{\mathcal{M}^{(glo)}}^{-1}(0)\mathcal{M}^{(glo)}(z)
    e^{\delta(z)\sigma_3}\mathcal{G}(z)^{-1} \v\\
		& \qquad +
		\mathcal{G}(0)e^{-\delta(0)\sigma_3}{\mathcal{M}^{(glo)}}^{-1}(0)E(0)^{-1}
   E_1\mathcal{M}^{(glo)}(z)e^{\delta(z)\sigma_3}\mathcal{G}(z)^{-1}z
		+\mathcal{O}(t^{-\epsilon}).
	\end{aligned}
\end{equation}
Therefore, by the reconstruction formulae \eqref{q2}, \eqref{x2}, together with \eqref{M_glo_12} and \eqref{E_0-1}, we obtain, as $t\to\infty$,
\begin{equation}\no
	\begin{aligned}
		q(x,t)=
		e^{-2\delta(0)}q^{(AG)}(x,t;\boldsymbol{E},\boldsymbol{\widehat{E}},\boldsymbol{\phi}-\boldsymbol{\delta})
		-e^{-2\delta(0)}
		\left(\frac{\partial}{\partial y}\mathcal{F}_{12}\right)\left(1+\mathrm{i}\frac{\partial}{\partial y}\mathcal{F}_{11}\right)
		+\mathcal{O}(t^{-\epsilon}),
	\end{aligned}
\end{equation}
and
\begin{equation}\no
	x=y+\mathrm{i}\mathcal{F}_{11}+\mathcal{O}(t^{-\epsilon}),
\end{equation}
where $\mathcal{F}=\left(\mathcal{F}_{ij}\right)_{i,j=1,2}$ with
\begin{equation}\label{mathcalF}
	\mathcal{F}=\frac{t^{-1/3}}
{|\theta^{(3,\widehat{j}_1)}(\widehat{\xi}_{j_1})|^{2/3}\widehat{E}_{j_1}}[\mathcal{G}(0){\mathcal{M}^{(glo)}}^{-1}(0)H_1(\widehat{E}_{j_1})
	\begin{pmatrix}
		0 & ia \\
		0 & 0
	\end{pmatrix}H_1(\widehat{E}_{j_1})^{-1}\mathcal{M}^{(glo)}(0)\mathcal{G}(0)^{-1}].
\end{equation}

ii) Similarly, in transition region II, we have
\begin{equation}\no
	\begin{aligned}
		q(x,t)=
		e^{-2\delta(0)}q^{(AG)}(x,t;\boldsymbol{E},\boldsymbol{\widehat{E}},\boldsymbol{\phi}-\boldsymbol{\delta})
		+e^{-2\delta(0)}
		\left(\frac{\partial}{\partial y}\mathcal{\tilde{F}}_{12}\right)\left(1+\mathrm{i}\frac{\partial}{\partial y}\mathcal{\tilde{F}}_{11}\right)
		+\mathcal{O}(t^{-\epsilon}),
	\end{aligned}
\end{equation}
and
\begin{equation} \no
	x=y+\mathrm{i}\mathcal{\tilde{F}}_{11}+\mathcal{O}(t^{-\epsilon}),
\end{equation}
where
\begin{equation}\label{mathcaltildeF}
	\mathcal{\tilde{F}}=\frac{t^{-1/3}}{|\theta^{(3,j_1)}(\xi_{j_1})|^{2/3}E_{j_1}}
\mathcal{G}(0){\mathcal{N}^{(glo)}}^{-1}(0)\tilde{H}_1(E_{j_1})
	\begin{pmatrix}
		0 & ia \\
		0 & 0
	\end{pmatrix}\tilde{H}_1(E_{j_1})^{-1}\mathcal{N}^{(glo)}(0)\mathcal{G}(0)^{-1}.
\end{equation}

iii) For the asymptotics of $q(x,t)$ in the Zakharov-Manakov region III, using \eqref{E_3}, we obtain
\begin{equation} \no
	\begin{aligned}
		q(x,t)=
		e^{-2\delta(0)}q^{(AG)}(x,t;\boldsymbol{E},\boldsymbol{\widehat{E}},\boldsymbol{\phi}-\boldsymbol{\delta})
		+e^{-2\delta(0)}
		\left(\frac{\partial}{\partial y}\mathcal{J}_{12}\right)\left(1+\mathrm{i}\frac{\partial}{\partial y}\mathcal{J}_{11}\right)
		+\mathcal{O}(t^{-1}),
	\end{aligned}
\end{equation}
and
\begin{equation} \no
	x=y+\mathrm{i}\mathcal{J}_{11}+\mathcal{O}(t^{-1}),
\end{equation}
where
\begin{equation}\label{mathcalJ}
	\begin{aligned}
		\mathcal{J}=-\frac{t^{-1/2}}{2z_1^2\sqrt{\theta^{(z_1,2)}(\xi)}}&\mathcal{G}(0)
{\mathcal{M}^{(glo)}}^{-1}(0)\mathcal{M}^{(glo)}(z_1)\\
		\times&\begin{pmatrix}
			0 & \beta_{21} \\
			\beta_{12} & 0
		\end{pmatrix}\mathcal{M}^{(glo)}(z_1)^{-1}\mathcal{M}^{(glo)}(0)\mathcal{G}(0)^{-1}.
	\end{aligned}
\end{equation}

iv) Finally, for the asymptotics of $q(x,t)$ in the fast-decay region IV, using \eqref{E_2}, we have
\begin{equation} \no
	\begin{aligned}
		q(x,t)=
		e^{-2\delta(0)}q^{(AG)}(x,t;\boldsymbol{E},\boldsymbol{\widehat{E}},\boldsymbol{\phi}-\boldsymbol{\delta})
		+\mathcal{O}(t^{-1}),
	\end{aligned}
\end{equation}
and
\begin{equation} \no
	x=y+\mathcal{O}(t^{-1}).
\end{equation}
This completes the proof of Theorem~\ref{theom}.


\section*{Acknowledgements}
This work was partially supported by National Natural Science Foundation of China under Grant No. 12471242, and Beijing Natural Science Foundation under Grant No. 1262023.

\end{document}